\documentclass[11pt]{article}

\usepackage{amsmath,amsthm,amsfonts,amssymb}
\usepackage{mathtools}
\usepackage{bbm}
\usepackage{graphicx}
\usepackage{booktabs}
\usepackage{array}
\usepackage{algorithm}
\usepackage{algorithmic}
\usepackage[dvipsnames]{xcolor}
\usepackage{pgfplots}
\usepackage[margin=1in]{geometry}
\usepackage{authblk}
\usepackage{natbib}
\usepackage{enumitem}
\usepackage{setspace}
\usepackage{hyperref}

\pgfplotsset{compat=1.18}
\usepgfplotslibrary{groupplots,fillbetween}
\usetikzlibrary{arrows.meta,positioning,calc}

\hypersetup{
  colorlinks=true,
  linkcolor=blue,
  citecolor=blue,
  urlcolor=blue
}
\renewcommand{\d}{\,\mathrm{d}}
\newcommand{\e}{\mathrm{e}}
\newcommand{\dsquare}{\mathop{\square}\displaylimits}
\newcommand{\dboxplus}{\mathop{\boxplus}\displaylimits}
\def\P{\mathbb{P}}

\newcommand{\E}{\mathbb{E}}
\newcommand{\R}{\mathbb{R}}
\newcommand{\N}{\mathbb{N}}
\newcommand{\Sh}{\phi}
\newcommand{\id}{\mathbbm{1}}
\newcommand{\bm}[1]{\boldsymbol{#1}}
\newcommand{\Q}{\mathbb{Q}}

\newcommand{\cx}{\le_{\mathrm{cx}}}
\newcommand{\X}{\mathcal{X}}
\newcommand{\A}{\mathbb{A}}
\newcommand{\IR}{\mathrm{IR}}
\newcommand{\AF}{\mathrm{AF}}
\newcommand{\Opt}{\mathrm{Opt}}
\newcommand{\Core}{\mathrm{Core}}
\newcommand{\Weber}{\mathcal{W}}

\newcommand{\Nuc}{\mathrm{Nuc}}

\DeclareMathOperator*{\argmin}{arg\,min}
\DeclareMathOperator*{\argmax}{arg\,max}
\DeclareMathOperator{\VaR}{VaR}
\DeclareMathOperator{\ES}{ES}
\DeclareMathOperator{\Cov}{Cov}
\DeclareMathOperator{\Var}{Var}
\DeclareMathOperator{\Cor}{Cor}
\DeclareMathOperator{\dom}{dom}

\theoremstyle{plain}
\newtheorem{theorem}{Theorem}
\newtheorem{proposition}{Proposition}
\newtheorem{lemma}{Lemma}
\newtheorem{corollary}{Corollary}

\theoremstyle{definition}
\newtheorem{definition}{Definition}
\newtheorem{example}{Example}

\theoremstyle{definition}
\newtheorem{assumption}{Assumption}

\theoremstyle{remark}

\title{Designing entry-monotone risk-sharing pools}

\author[1]{Christopher Blier-Wong}
\author[2]{Jean-Gabriel Lauzier}

\affil[1]{Department of Statistical Sciences, University of Toronto, Canada}
\affil[2]{Department of Economics, Memorial University of Newfoundland, Canada}

\date{\today}

\begin{document}
\maketitle

\begin{abstract}
While risk pooling lowers the total cost of risk, efficiency alone does not make a pool viable. Participants need terms that ensure their participation, that are immune to subgroups breaking away, and that allow new members to join. Under cash-additive risk measures, the minimum cost of a coalition's risk determines the value created by that coalition, and deterministic side payments redistribute that value among participants. Institutional risk sharing is thus a transferable-utility cooperative game. We prove that the game is totally balanced whenever the risk measures are convex (agents are risk averse), so every coalition has a nonempty core and stable allocations always exist. We then analyze entry monotonicity through Population-Monotonic Allocation Schemes \citep{sprumont1990population}, a strong requirement that is notoriously difficult to construct and has received limited attention in risk sharing. We find several structural conditions that ensure that either the Arrow--Debreu pricing surplus allocation rule or the proportional-cost surplus allocation rule satisfies this entry-monotonicity property, the latter being a novel cooperative notion we propose. These verifiable structural conditions naturally arise in pooled (re)insurance and credit portfolios, providing pool designers with a practical toolkit for building risk pools that remain stable and attractive as they expand. 
\end{abstract}

\section{Introduction}\label{sec:introduction}

Risk pooling is a central mechanism for managing large losses that are difficult for any single organization to absorb. Countries and regional authorities pool catastrophe and climate-related disaster risks \citep{worldbank2017sovereign,ciullo2023increasing}; firms use insurance and contractual arrangements to share supply-chain disruption risks \citep{Tomlin2006,dong2012managing}, 
and insurers form pools to manage losses whose scale or volatility exceeds the capacity of a single balance sheet \citep{embrechts2018quantilebaseda,feng2024expansion}. In all of these settings, pooling can create value by reducing the total cost of risk \citep{Borch1962,Wilson1968,Pratt2000}. 

However, the value created by risk pooling may be unequally distributed among participants depending on the details of the risk-sharing arrangement. Risk pooling, therefore, involves two decisions. The first is the (pure) risk allocation decision: who bears which random losses after the pool is formed. The second is the implementation decision: which premiums, prices, or side payments among participants make that allocation acceptable. In a risk exchange economy, the two are jointly determined by equilibrium prices.
The resulting exchange is efficient in the sense of Pareto optimality, as no participant can be made better off without making someone else worse off. In equilibrium, the market also determines how the surplus is divided, and the outcome might not be ``fair'' in any intuitive sense. Can risk pools sometimes do ``better'' than market mechanisms? If not, should risk pools emulate risk markets and use standard pricing mechanisms?

This paper analyzes the implementation problem for institutional risk sharing from the perspective of cooperative game theory. The intended applications are financial institutions, insurers, reinsurers, firms, governments, regional pools, and other organizations that evaluate risk in monetary terms. We work with cash-additive monetary risk measures, with particular attention to distortion risk measures \citep{Artzner1999,Wang2000,FollmerSchied2002,Acerbi2002}. These evaluations are natural in institutional settings because they encode tail-sensitive capital charges, premium principles, and risk-management objectives in monetary units \citep{follmer2016stochastic}. They are less intended as a model of household-level informal risk sharing, where expected utility, consumption smoothing, or behavioural models are more appropriate.

The pivotal insight of our approach is that when deterministic cash transfers between agents are feasible, and agents evaluate positions in a common monetary numeraire using cash-additive criteria, then the induced risk-sharing game is a transferable-utility (TU) game. Cash additivity gives the implementation problem an exact monetary structure: under the standard loss convention, adding a deterministic cost $c$ to an agent shifts that agent's evaluation by exactly $c$. Balanced deterministic transfers thus do not change either the aggregate random allocation or the planner's aggregate objective, but they redistribute welfare gains across participants. Once an efficient random allocation has been found, the remaining question is how to redistribute the surplus generated by pooling using cash payments. This redistribution problem is the central question that we tackle. 

While an allocation can be efficient in the sense of minimizing the aggregate evaluated cost of risk, it can nevertheless fail as a business or policy mechanism if one participant is made worse off, if a subgroup can object to the proposed terms, or if incumbents lose when a new participant is admitted.  We show that when all agents use convex risk measures (so they are risk averse), the risk-sharing game is always totally balanced, and each subcoalition has a nonempty core. This result, paralleling the Shapley-Shubik Theorem \citep{ShapleyShubik1969}, sheds light on the cooperative nature of risk-sharing games: when agents are risk averse, one can always find allocations that are stable in the sense that no coalition can profitably block them.

Having established the general stability of risk pools when agents are risk averse, we then ask whether they are entry monotone, meaning that incumbents never lose when a new participant is admitted. We surmise that entry monotonicity is the most natural notion of fairness that one can desire besides coalitional stability. Answering this question amounts to asking whether risk-sharing games admit a \textit{Population Monotonic Allocation Scheme} (PMAS) \citep{sprumont1990population}. PMAS is a refinement of the notion of core allocation that is notoriously difficult to construct because it requires the underlying TU game to be totally balanced. 

We focus on two surplus allocation rules. The first is a novel allocation rule called the proportional-cost rule, which assigns each agent a share of the surplus proportional to the risk they contribute to the pool. We interpret the proportional-cost rule as a purely cooperative solution, since it is not generally guaranteed to yield a surplus allocation consistent with equilibrium. The second is the welfare allocation generated by the Arrow--Debreu pricing measure in a fictitious economy consisting only of the agents of a subcoalition; clearly, the Arrow--Debreu pricing rule is consistent with equilibrium behaviour by construction.

Both rules have a key operational advantage compared to other PMAS: they have the property that they allow for ``myopic extensions of the pool" in the sense that they do not require the knowledge of the full coalition composition beforehand. They can thus be written into a contract or the pool's constitution at the time of formation, without requiring any knowledge of future potential participants. To the best of our knowledge, this paper is the first to propose such a notion of PMAS that satisfies this notion of myopic extensions.\footnote{The term myopic extension explicitly evokes both myopic policies in one-step-ahead decision processes, like in dynamic programming, and the notion of extendability in probability theory.}

We provide many structural conditions that make these surplus allocation rules PMAS. For example, we show that when normalized coalition risks are ordered by convex order, then the proportional-cost rule allocation is a PMAS. Similarly, the Arrow--Debreu allocation rule yields a PMAS when the expected cost of agents' risk contributions decreases as the coalition extends.
These conditions cover many cases used in practice, such as a broad family of elliptical risks used in quantitative risk management \citep{mcneil2015quantitative}, and comonotonic risks, the worst-case dependence structure in a given Fr\'echet class.

A practical takeaway from our approach is that the surplus allocation rules we analyze need not be PMAS under the same circumstances. One can easily find examples in which the proportional-cost allocation rule is a PMAS, but the Arrow--Debreu allocation rule is not, and vice versa. That is, there are circumstances in which a purely cooperative approach is justified because the underlying coalition-wise pricing measure does not guarantee entry monotonicity, and other circumstances in which it does, allowing the pool designer to simply emulate a market for risk. We note with interest that this observation opens the possibility of studying assortative matching, risk, and optimal pool design in a market for risk-sharing pools. 


\subsection{Contributions and literature}\label{ss:contributions}

Our first contribution is in the definition of risk-sharing games. Let $\zeta$ denote idiosyncratic risks and $\rho$ denote monetary risk measures, such as the Expected Shortfall in regulatory capital reserve \citep{Artzner1999, Basel2019,wangAxiomaticFoundationExpected2021}. For a coalition $A\subseteq[n]$ of financial institutions, the infimal cost achievable by pooling and redistributing risks is
$$C(A):=\inf_{(X_i)_{i\in A}}\left \{\sum_{i\in A}\rho_i(X_i) \text{  subject to  } \sum_{i\in A}X_i= \sum_{i\in A}\zeta_i\right\}.$$
Finding solutions to the problem above is the central problem of the standard risk-sharing literature \citep{barrieu2005infconvolutiona,acciaio2007optimal,jouini2008optimal,filipovic2008optimal,embrechts2018quantilebaseda,Liu2024LambdaVaR}, and we say that the risk allocation $(X_i)_{i\in A}$ is efficient if it solves the problem. 

By construction, the aggregate cost $C(A)$ of such a coalition is less than the sum of the individual costs. We thus define its welfare gain as
$$\nu(A):= \sum_{i\in A}\rho_i(\zeta_i)-C(A).$$
By cash additivity, balanced side payments do not affect the value of a coalition and $\nu(A)$ is uniquely determined. The function $\nu$ is therefore the \textit{(scalar) characteristic function of a TU game}. We show in Lemma \ref{lemma:gain-monotone} that $\nu$ is a zero-normalized \citep{sprumont1990population} monotone and bounded set function \citep{MarinacciMontrucchio2004Ambiguity}.

Three well-known consequences immediately stem from the TU structure of risk-sharing games. First, the utility possibility set of individually rational Pareto optimal allocations is a convex polyhedron; in fact, it is a simplex in our case (Theorem \ref{thm:transfer-IR-geometry}). Second, any such allocation can be implemented through budget-balanced side payments. The third consequence is operational: one can often separate the problem of efficiency (solving for an efficient \textit{risk allocation} to obtain $C(A)$) and the problem of redistributing the welfare surplus (finding adequate side payments). Operationally, this means that a coalition of financial institutions can rely on their quantitative analyst division to efficiently allocate risk and leave the executive members to bargain over the price of participating in the pool.\footnote{Interesting exceptions arise when one uses a ``mechanical'' pricing rule like actuarial fairness, which is not based on game theoretical foundations. In this case, it is possible to have two different efficient allocations that differ in their welfare split, and thus in their coalitional stability; see Appendix~\ref{ss:actuarial-fairness} for a longer discussion.}

The next natural question is how the core of risk-sharing games behaves, where the core is defined as the subset of welfare allocations that no coalition can block. We show in Theorem \ref{thm:convex-implies-core} that the risk-sharing game is totally balanced when risk measures are convex, so each coalition has a nonempty core. This intuitive result establishes that one can always find a stable, mutually beneficial risk-sharing arrangement when agents are risk averse, paralleling the Shapley-Shubik Theorem \citep{ShapleyShubik1969} in market games without uncertainty but with concave utility functions. Section~\ref{ss:dual-infconv-pricing} gives the general pricing measure construction using the usual dual approach, and Appendix~\ref{app:distortion-background} records its lower-envelope formulation when specializing to distortion risk measures.

The distinction between efficiently allocating risk and fairly allocating the surplus of risk pooling is especially important because pools are not static. Catastrophe pools admit new countries, financial syndicates add counterparties, reinsurance arrangements expand capacity, and institutional risk-sharing mechanisms might evolve as new exposures appear. For example, a catastrophe pool may lower the risk-adjusted total cost by admitting a new country, while still requiring premium terms that prevent incumbents from bearing an unattractive share of the entrant's hurricane or earthquake exposure. The \textit{Caribbean Catastrophe Risk Insurance Facility}'s staggered, and legally segmented expansion illustrates the operational complexity of admitting heterogeneous members under a common catastrophe-risk-pooling platform \citep{worldbank2017sovereign,CCRIFAbout}. Similar issues arise in operations research models of supply-chain disruption, profit-sharing agreements, central counterparty risk allocation, and risk-pooling expansion \citep{dong2012managing,FuSimTeo2018,ghamami2019submodular,feng2024expansion}. 

New participants can increase aggregate welfare through two channels: diversification and risk-bearing capacity. Diversification arises when the entrant's loss is imperfectly correlated with the existing aggregate loss, so that pooling can reduce risk-adjusted cost, capital intensity, or variability relative to stand-alone coverage. Risk-bearing capacity arises when the entrant is willing to absorb some of the aggregate risk at a lower evaluated cost than the incumbents. An example of risk-bearing capacity is reinsurance, where an external risk bearer supplies capacity for a premium and absorbs part of the aggregate exposure without adding new underlying exposure. But aggregate value creation does not imply that every incumbent benefits. A transfer rule (i.e., a price or side payments) must therefore address not only implementability for a fixed pool but also guarantee that expansion is beneficial to everyone. Despite their practical importance, the cooperative-game properties of canonical risk-sharing transfer rules, particularly their behaviour under coalition expansion, have not been systematically studied in the existing literature.

This motivates our extensive analysis of stable expansion and PMAS \citep{sprumont1990population} in Section~\ref{sec:stable-expansion-dual}. A PMAS requires a coherent family of core allocations for all coalitions, chosen so that incumbents weakly benefit whenever a coalition expands, a stronger requirement than merely requiring the core property. Most risk/surplus allocation rules do not generally satisfy this property. This includes not only the proportional-cost and the coalition-wise pricing rule analyzed in the main text, but also the actuarially fair rule, the coalition-wise Shapley rule, and the coalition-wise nucleolus; see Appendix~\ref{ss:canonical-not-PMAS} for more details.

We identify conditions that generate PMAS surplus allocation rules. Convex-order consistency gives a PMAS for the proportional-cost rule: when each incumbent's coalitional average improves in convex order as the coalition expands, convexity of the risk measures ensures that no incumbent is made worse off (Theorem~\ref{thm:PC-PMAS}). For dual pricing, we find a cross-coalition compatibility criterion: prices selected by larger coalitions must support the larger-coalition allocation relative to the allocations available in smaller coalitions (Theorem~\ref{thm:cmrs-pricing-PMAS}). 

We collect some special cases that satisfy the conditions under which proportional-cost or dual pricing rules are PMAS, focusing on the most useful cases for institutional risk-sharing. For elliptical losses, Arrow--Debreu pricing yields a PMAS under a cross-monotonicity condition involving correlations with coalition aggregate risks (Corollary \ref{cor:elliptical-PMAS}), generalizing the results of \citet{chenhuwang2017stable}. For comonotonic endowments, the Arrow--Debreu pricing rule is also a PMAS if agents evaluate risk through distortion risk measures (Corollary \ref{cor:comonotonic-PMAS}). More generally, we provide conditions under which convex ordering of endowments or conditional means supports PMAS constructions for independent risks, exchangeable pools, and some specific families of random vectors that are relevant for (re)insurance, including mean-parametrized convolution semigroups, Dirichlet--Liouville radial risks, and frailty-type portfolios (Appendix \ref{app:additional-pmas-examples}).

Although cooperative allocation, fair pricing, and surplus division are well-developed themes in the operations research literature, only a small part of the risk-sharing literature explicitly analyzes risk-sharing through cooperative game theory. This literature provides only partial results, and no previous paper has fully leveraged the TU structure of risk-sharing games. See \citep{denault2001coherent,grechuk2013cooperative,chenhuwang2017stable,feng2024expansion} for risk sharing, and \citep{filipovic2008equilibrium, boonen2024paretoefficient} for the core property of equilibrium allocations. A complementary line of research approaches the pricing problem from an equilibrium perspective. \citet{boonen2016pricing} separate allocation and pricing in bilateral risk sharing under comonotonic additive preferences; \citet{boonen2021competitive} study competitive equilibria in comonotonic markets with dual-utility agents; \citet{boonen2024paretoefficient} analyze Pareto-efficient contracts in a centralized insurance market, and \citet{ghossoub2026efficiency} characterize efficiency and equilibrium support in pure-exchange economies with risk-averse monetary utilities. Our lower-envelope Arrow--Debreu pricing rule in Appendix \ref{app:distortion-background} is closely related to this equilibrium logic, but our interpretation differs. We use the pricing measure to implement efficient multilateral risk-sharing cooperative arrangements through deterministic transfers and to prove the coalitional stability of the induced welfare split, but we do not necessarily interpret those as an equilibrium outcome in a market.

Risk-sharing is related to a large operations research literature on cooperative cost allocation and surplus division. Axiomatic shared-cost allocation and equitable cost-sharing prices study how common costs should be divided among participants \citep{BilleraHeath1982,MirmanTauman1982}. Optimization-induced cooperative games also arise in stochastic inventory centralization and newsvendor settings, where duality, core stability, and entry monotonicity are used to construct stable allocation rules \citep{ChenZhang2009,ChenGaoHuWang2019}. Related surplus-division questions appear in supply chains through cost-savings allocation, profit-sharing agreements, and Nash bargaining \citep{LengParlar2008,NagarajanBassok2008,FuSimTeo2018}. Institutional risk sharing fits this broader research line: pooling creates a joint surplus, and the central operational question is whether that surplus can be allocated through transfers that are individually rational, coalitionally stable, and robust to myopic extension.

This perspective is also consistent with operational research work in which cooperative games are generated by optimization models under uncertainty: stochastic-programming duality yields constructive core allocations in inventory centralization games \citep{ChenZhang2009}, while dynamic linear programming games have been extended to risk-averse players using coherent conditional risk measures \citep{TorielloUhan2017}. Our total-balancedness condition is related to the result of \citet{anily2014subadditive}, who show that subadditivity together with degree-one homogeneity implies total balancedness for regular cooperative cost games.

Finally, a related paper on pool expansion is \citet{feng2024expansion}, which studies entropic risk measures with jointly normal losses and introduces strong and weak consensus conditions for one-step expansion under actuarial and equilibrium pricing. While they are motivated by a notion similar to our ``myopic extensions", their framework focuses on entry analysis without characterizing full PMAS requirements over the entire coalition lattice. Our intended applications are institutional, including nations, reinsurers, and financial syndicates, which have the capacity to evaluate subcoalition deviations. For such participants, a transfer rule that fails PMAS is harder to defend against subgroup objections, and we surmise that the appropriate stability benchmark is consistency across the entire coalition lattice rather than one-step expansion alone.

The remainder of this paper is structured as follows. Section~\ref{sec:preliminaries} introduces monetary and distortion risk measures and the efficient risk-sharing problem. Section~\ref{sec:main-results} defines risk-sharing games and studies individual rationality, core stability, and dual Arrow--Debreu pricing. Section~\ref{sec:stable-expansion-dual} studies stable expansion and PMAS, providing general criteria and constructions for several dependence structures. Section~\ref{sec:discussion} concludes. Proofs of the results stated in the main text are deferred to Appendix~\ref{app:main-text-proofs}. Appendix~\ref{app:distortion-background} records the lower-envelope specialization for distortion risk measures, Appendix~\ref{sec:transfer-rules} derives the properties of additional transfer rules, Appendix~\ref{app:additional-pmas-examples} collects further PMAS examples, and Appendix~\ref{sec:examples} illustrates the mechanisms in numerical examples.

\section{Preliminaries}\label{sec:preliminaries}

\subsection{Risk measures}\label{ss:risk-measures}

Let $(\Omega,\mathcal F,\P)$ be an atomless probability space. We denote by $\mathcal B(\mathbb R)$ the Borel $\sigma$-algebra on $\mathbb R$, namely the smallest $\sigma$-algebra containing the open subsets of $\mathbb R$. Unless otherwise stated, random variables are real-valued Borel-measurable functions on $\Omega$. Fix $p\in[1,\infty]$ and set $\X:=L^p(\Omega,\mathcal F,\P)$. Random variables represent losses, so positive values correspond to costs.

An extended-valued functional $\rho:\X\to(-\infty,\infty]$ is \emph{proper} if its effective domain $\dom\rho:=\{Y\in\X:\rho(Y)<\infty\}$ is nonempty. Fix a positive integer $n$ and let $[n]=\{1,\dots,n\}$. For each $i\in[n]$, agent $i$ evaluates losses through a proper risk measure $\rho_i:\X\to(-\infty,\infty]$. Consider the following properties:
\begin{enumerate}[label=\textup{(A\arabic*)},ref=A\arabic*]
\item\label{ax:monotonicity} \emph{Monotonicity}: $Y\le Y'$ a.s.\ implies $\rho(Y)\le\rho(Y')$.
\item\label{ax:cash-additivity} \emph{Cash additivity}: $\rho(0)\in\R$ and $\rho(Y+c)=\rho(Y)+c$ for all $Y\in\X$ and $c\in\R$.
\item\label{ax:subadditivity} \emph{Subadditivity}: $\rho(Y+Y')\le\rho(Y)+\rho(Y')$ for all $Y,Y'\in\X$.
\item\label{ax:pos-homogeneity} \emph{Positive homogeneity}: $\rho(\lambda Y)=\lambda\rho(Y)$ for all $Y\in\X$ and $\lambda\ge 0$.
\item\label{ax:convex} \emph{Convexity}: $\rho(\lambda Y+(1-\lambda)Y')\le \lambda\rho(Y)+(1-\lambda)\rho(Y')$ for all $Y,Y'\in\X$ and $\lambda\in[0,1]$.
\item\label{ax:lsc} \emph{Lower semicontinuity}: $\rho$ is lower semicontinuous on $\X$.
\item\label{ax:law-invariance} \emph{Law invariance}: $\rho(Y)=\rho(Y')$ whenever $Y$ and $Y'$ have the same distribution.
\item\label{ax:comonotonic_additivity} \emph{Comonotonic additivity}: $\rho(Y+Y')=\rho(Y) +\rho(Y')$ for all comonotonic $Y,Y'\in\X$.
\end{enumerate}
Throughout the paper, each $\rho_i$, $i\in[n]$, is proper and satisfies \eqref{ax:cash-additivity}; the initial losses introduced below satisfy $\zeta_i\in\dom\rho_i$ for every $i\in[n]$. A functional satisfying \eqref{ax:monotonicity} and \eqref{ax:cash-additivity} is called a \emph{monetary risk measure}. It is \emph{convex} if it satisfies \eqref{ax:convex}, and \emph{coherent} if it additionally satisfies \eqref{ax:subadditivity} and \eqref{ax:pos-homogeneity}; every coherent risk measure is convex. Convexity is imposed where total balancedness is proved, and monotonicity is used only when explicitly stated.

Convex risk measures admit a standard dual description in terms of supporting linear functionals. To state the version used below, fix a separating dual pair $(\X,\mathcal Z)$, so that $\mathcal Z$ is the chosen dual partner of $\X$, with pairing $\langle Y,\pi\rangle:=\E[\pi Y]$ whenever this expectation is well-defined. If a proper convex risk measure $\rho:\X\to(-\infty,\infty]$ is lower semicontinuous for the weak topology $\sigma(\X,\mathcal Z)$ induced by this dual pair, i.e.\ the topology generated by the maps $Y\mapsto\E[\pi Y]$ for $\pi\in\mathcal Z$, then the Fenchel--Moreau theorem gives the dual representation
\begin{equation}\label{eq:individual-dual-representation}
    \rho(Y)
    =
    \sup_{\pi\in\mathcal Z}
    \big\{\E[\pi Y]-\alpha_\rho(\pi)\big\},
    \qquad Y\in\X,
\end{equation}
where the minimal penalty is the convex conjugate
\begin{equation}\label{eq:minimal-penalty}
    \alpha_\rho(\pi):=\sup_{Y\in\X}\{\E[\pi Y]-\rho(Y)\},
    \qquad \pi\in\mathcal Z.
\end{equation}
For the agent-specific risk measures, write $\alpha_i:=\alpha_{\rho_i}$ for $i\in[n]$. On the effective domain of the penalty, cash additivity forces $\E[\pi]=1$, and monotonicity forces $\pi\ge0$; the two together make $\pi$ a probability density.

\emph{Distortion risk measures} are monetary risk measures satisfying \eqref{ax:pos-homogeneity} and \eqref{ax:comonotonic_additivity}. 
Let
\begin{equation*}
\mathcal{H}:=\{h:[0,1]\to\R:\ h \text{ is non-decreasing with } h(0)=0 \text{ and }h(1)=1\}.
\end{equation*}
A function $h\in\mathcal{H}$ is called a probability distortion, the set function $h\circ \P$ is called a capacity, and a distortion risk measure $\rho_h$ is the Choquet integral defined as 
\begin{equation}\label{eq:choquet-integral}
\rho_h(Y)=\int Y \d h\circ \P
=\int_{0}^{\infty} h(\P(Y>x))\,\d x
+\int_{-\infty}^{0} \left[ h(\P(Y>x))-h(1) \right]\,\d x,
\end{equation}
where the first integral is in the sense of Choquet and the last two are in the sense of Lebesgue. Under the convention that $Y\in \X$ represents a loss, if the distortion function $h$ is concave, then the distortion risk measure $\rho_h$ is subadditive, and thus coherent. 

\subsection{Optimal allocations}\label{ss:optimal-allocations}

Each agent $i\in[n]$ holds an initial loss $\zeta_i$. For a coalition $A\subseteq[n]$, set $X_A:=\sum_{i\in A}\zeta_i$; in particular, $X:=X_{[n]}$ is the aggregate loss of the grand coalition. Here and elsewhere, when a coalition-indexed object has no coalition subscript, it refers to the grand coalition $[n]$. For a loss $Y\in\X$ and $1 \leq m \leq n$, write
$$\A_m(Y):=\left\{(Y_i)_{i\in [m]}\in\X^m:\ \sum_{i\in [m]}Y_i=Y\right\}.$$
We interpret $\bm{X}=(X_1,\dots,X_n)\in\A_n(X)$ as a risk-sharing arrangement, with original no-sharing allocation $\bm{\zeta}:=(\zeta_1,\dots,\zeta_n)\in\A_n(X)$.

\begin{definition}[Comonotonicity]\label{def:comonotonicity}
Random variables $Y,Z\in\X$ are \emph{comonotonic} if
$$\bigl(Y(\omega)-Y(\omega')\bigr)\bigl(Z(\omega)-Z(\omega')\bigr)\ge 0 \quad\text{for $(\P\otimes\P)$-a.s.\ $(\omega,\omega')\in\Omega\times\Omega$.}$$
A collection $(Y_1,\dots,Y_n)\in\X^n$ is comonotonic if each pair $(Y_i,Y_j)$, $i,j\in[n]$, is comonotonic.
\end{definition}

By Denneberg's Lemma \citep[Proposition~4.5]{denneberg1994nonadditive}, $Y$ and $Z$ are comonotonic if and only if there exists a random variable $W$ and nondecreasing functions $f,g$ such that $Y=f(W)$ and $Z=g(W)$ almost surely. The same representation extends to $n$-tuples: $(Y_1,\dots,Y_n)$ is comonotonic if and only if there exists $W$ and nondecreasing $f_1,\dots,f_n$ with $Y_i=f_i(W)$ for all $i\in[n]$. In particular, one may take $W=\sum_{i\in [n]} Y_i$.

For a nonempty coalition $A\subseteq[n]$, let $\A_{|A|}^+(Y):=\{\bm{Y}\in\A_{|A|}(Y):\ \bm{Y} \text{ is comonotonic}\}$. The coalition inf-convolution is
\begin{equation}\label{eq:coalition-infconv-functional}
\dsquare_{i\in A}\rho_i (Y):=\inf_{\bm{Y}\in\A_{|A|}(Y)}\sum_{i\in A}\rho_i(Y_i), \qquad Y\in\X,
\end{equation}
and the coalition cost is
\begin{equation}\label{eq:coalition-cost}
    C(A):=\dsquare_{i\in A}\rho_i (X_A).
\end{equation}
Set $X_\varnothing:=0$ and $C(\varnothing):=0$. The associated comonotonic inf-convolution is
\begin{equation}\label{eq:comonotonic-infconv}
\dboxplus_{i\in [n]} \rho_i(X)
:=\inf_{\bm X\in\A_n^+(X)}\sum_{i=1}^n \rho_i(X_i).
\end{equation}

For every nonempty coalition $A\subseteq[n]$ and loss $Y\in\X$, define
$$\Opt_A(Y):=\argmin_{\bm{Y}\in\A_{|A|}(Y)}\sum_{i\in A}\rho_i(Y_i)\subseteq\A_{|A|}(Y).$$
In particular, $\Opt(X)$ denotes the efficient allocations of the aggregate loss $X$ for the grand coalition.

Unless explicitly stated otherwise, the following assumptions are in force throughout the paper.
\begin{assumption}\label{ass:standing}
Let $(\Omega,\mathcal F,\P)$ be an atomless standard probability space, $p\in[1,\infty]$,
and $\mathcal X=L^p(\Omega,\mathcal F,\mathbb P)$. For every $i\in[n]$, assume 
$\rho_i:\mathcal X\to(-\infty,\infty]$
is proper, cash additive, convex, lower semicontinuous, and law invariant, and
$\zeta_i\in\operatorname{dom}\rho_i$.
\end{assumption}

The following proposition is the Comonotonic Improvement Theorem \citep{landsberger1994comonotone, ludkovski2008comonotonicity}. It states that any non-comonotonic allocation can be Pareto-improved by a comonotonic rearrangement when the evaluations are convex, so the infimum over all allocations coincides with the infimum over comonotonic ones.

\begin{proposition}[Comonotonic Improvement Theorem]\label{prop:comonotonic-optimality}
Let $\rho_1,\dots,\rho_n$ be convex risk measures. Then
$\dsquare_{{i\in [n]}}\rho_i(X)=\dboxplus_{i \in [n]} \rho_i(X)$ for $X\in\X.$
\end{proposition}

The following standard exactness result gives finite, attained coalition inf-convolutions. We use the loss-convention version of \citet[Theorem~2.5]{filipovic2008optimal}.

\begin{proposition}\label{prop:infconv-attainment}
For every nonempty coalition $A\subseteq[n]$, the inf-convolution $\dsquare_{i\in A}\rho_i$ is proper, cash additive, convex, lower semicontinuous, and law invariant. Moreover, $\dsquare_{i\in A}\rho_i (X_A)$ is exact: for every $Y\in\X$ there exists $\bm{Y}\in\A_{|A|}(Y)$ such that
$$\dsquare_{i\in A}\rho_i (Y)=\sum_{i\in A}\rho_i(Y_i).$$
In particular, $C(A)=\dsquare_{i\in A}\rho_i (X_A)$ is finite, $\Opt_A(X_A)$ is nonempty, and $\Opt_A(X_A)$ contains a comonotonic allocation.
\end{proposition}

An allocation $\bm{X}\in\A_n(X)$ is \emph{Pareto optimal} if there is no $\bm{Y}\in\A_n(X)$ with $\rho_i(Y_i)\le\rho_i(X_i)$ for all $i\in[n]$ and strict inequality for some $i \in [n]$.

\emph{Individual rationality} (IR) requires that agents weakly prefer participation to autarky:
\begin{equation}\label{eq:IR}
\rho_i(X_i)\le \rho_i(\zeta_i),\qquad i\in [n].
\end{equation}
Let $\A_{n,\IR}(X):=\{\bm{X}\in\A_n(X):\ \bm{X} \text{ satisfies \eqref{eq:IR}}\}$. An \emph{individually rational Pareto optimal} allocation is any Pareto optimal allocation in $\A_{n,\IR}(X)$.

\section{Risk-sharing games}\label{sec:main-results}

We define the risk-sharing game generated by coalition-wise inf-convolution, i.e., efficient coalition-wise risk allocation. 
Because of cash additivity, the problem of implementing risk sharing is a TU cooperative game. We analyze two basic properties that an allocation must satisfy: individual rationality and coalitional stability. Section~\ref{ss:dual-infconv-pricing} develops the general dual-pricing rule, while Appendix~\ref{app:distortion-background} records the lower-envelope distortion specialization.

\subsection{Risk-sharing games as TU games}\label{ss:gain-game}

Consider a risk-sharing pool whose members may reallocate only their own aggregate endowment. The cost of this restricted pool is the coalition-wise inf-convolution of the members' risk measures, and the difference between this cost and autarky is the welfare gain available for redistribution.

\begin{definition}[Gain game]\label{def:coalitional-game}
For a coalition $A\subseteq[n]$, define its welfare gain by
\begin{equation}\label{eq:gain-game}
\nu(A):=\sum_{i\in A}\rho_i(\zeta_i)-C(A).
\end{equation}
\end{definition}

The gain game records the best welfare improvement that a coalition can obtain by reallocating only its own endowments. The next lemma shows that $\nu$ is a (i) zero-normalized, (ii) nonnegative, (iii) monotone, and (iv) bounded set function.

\begin{lemma}\label{lemma:gain-monotone}
The gain game in \eqref{eq:gain-game} satisfies:
\begin{enumerate}[label=\textup{(\roman*)},ref=\roman*]
\item $\nu(\{i\})=0$ for every $i\in [n]$;
\item $\nu(A)\ge 0$ for every $A\subseteq [n]$;
\item if $A\subseteq B\subseteq[n]$, then $\nu(A)\le \nu(B)$;
\item $\nu([n])<\infty$.
\end{enumerate}
\end{lemma}

Lemma \ref{lemma:gain-monotone} shows that (i) autarky yields zero gain, that (ii) no coalition loses from having the option to share internally, and that (iii) enlarging the coalition cannot reduce the attainable gain.

Since $\nu(A)$ is a scalar welfare gain attached to each coalition $A\subseteq[n]$, the pair $([n],\nu)$ is the characteristic-function form of the risk-sharing game. It is a transferable utility game because of cash additivity. Formally, a vector $\bm{c}\in\R^n$ on the balanced-budget hyperplane
\begin{equation}\label{eq:H}
H:=\left\{\bm{c}\in\R^n:\ \sum_{i\in [n]} c_i=0\right\}
\end{equation}
preserves feasibility, $\bm{X}+\bm{c}:=(X_1+c_1,\dots,X_n+c_n)\in\A_n(X)$, and shifts evaluations componentwise, $\rho_i(X_i+c_i)=\rho_i(X_i)+c_i$. We define a \textit{transfer rule} as a deterministic mapping
$((\rho_1,\zeta_1),\dots, (\rho_n, \zeta_n)) \longmapsto (c_1(X_1,\zeta_1),\dots, c_n(X_n,\zeta_n))$
attached to an optimal allocation $(X_1,\dots,X_n)$. Pareto optimality of $\bm{X}+\bm{c}$ is equivalent to $\sum_{i=1}^n\rho_i(X_i)=\dsquare_{i\in [n]}\rho_i$, so the question of which transfer rule to apply is a question of selection within the set of efficient allocations.

\begin{definition}[Transfer-equivalence classes]\label{def:transfer-classes}
For allocations $\bm{X},\bm{Y}\in\A_n(X)$ write $\bm{X}\sim \bm{Y}$ if $\bm{Y}-\bm{X}=\bm{c}\in H$. The equivalence classes defined by $\sim$ are called the transfer-equivalence classes, and we write $\Opt(X)_{/\sim}$ for the quotient space defined by $\sim$.
\end{definition}

Notice that a transfer-equivalence class always has a canonical pure risk allocation $\bm{X}_{0}\in \Opt(X)$, namely a representative whose components $X_{0,i}$ consist only of the random reallocation of risk and carry no additive cash component. Other elements of the class are obtained by applying a balanced transfer $\bm{c}\in H$ to $\bm{X}_0$. Notice that we also meaningfully refer to $\bm{X},\bm{Y} \in \Opt(X)_{/\sim}$ as different risk allocations when $ \bm{X}\nsim \bm{Y}$ so they do not belong to the same transfer-equivalence class. Finally, notice that $ \bm{X}\nsim \bm{Y}$ if and only if $\bm{X}_{0} \nsim \bm{Y}_{0}$  for $\bm{X} \sim \bm{X}_{0}$ and $\bm{Y} \sim \bm{Y}_{0}$. Taken together, these observations imply that we can slightly abuse the nomenclature and refer to different $\bm{X},\bm{Y} \in \Opt(X)_{/\sim}$ as either risk allocations or pure risk allocations whenever this is unambiguous.

While we write deterministic transfers in the form $\bm{X}\mapsto \bm{X}+\bm{c}$ with $\bm{c}\in H$ from \eqref{eq:H}, it is generally more convenient to track transfers through evaluation vectors and welfare gains.

\begin{definition}[Evaluation vector and welfare baselines]\label{def:evaluation-welfare}
For $\bm{X}=(X_1,\dots,X_n)\in\A_n(X)$ with finite component evaluations, define the evaluation vector
\begin{equation*}
\bm{r}(\bm{X}):=(\rho_1(X_1),\dots,\rho_n(X_n))\in\R^n.
\end{equation*}
Define agentwise welfare baselines (risk reductions) relative to the initial endowment $\bm{\zeta}$ by
\begin{equation*}
b_i(\bm{X}):=\rho_i(\zeta_i)-\rho_i(X_i),\qquad i\in [n].
\end{equation*}
\end{definition}

We reserve the symbol $w_i$ for the post-transfer welfare share assigned to agent $i\in[n]$, defined by $w_i:=b_i(\bm{X})-c_i$ when a balanced transfer $\bm{c}\in H$ is applied to $\bm{X}$. Since transfers act componentwise on evaluations, meaning $\bm{r}(\bm{X}+\bm{c})=\bm{r}(\bm{X})+\bm{c}$ for $\bm{c}\in H$, we have that $\{\bm{r}(\bm{X}+\bm{c}):\bm{c}\in H\}=\bm{r}(\bm{X})+H$ is an affine hyperplane in $\R^n$. 

The next theorem describes the geometry of individually rational allocations for an efficient risk allocation: the simplex below is the imputation simplex of the gain game, and the theorem recovers the well-known result that the utility possibility frontier of zero-normalized TU games is a simplex.

\begin{theorem}\label{thm:transfer-IR-geometry}
Fix a pure risk allocation $\bm{X}_0\in\Opt(X)$, let $b_i:=b_i(\bm{X}_0)$ for $i\in[n]$ be as in Definition~\ref{def:evaluation-welfare}, and write $W:=\nu([n])$ for the grand-coalition gain. Then:
\begin{enumerate}[label=\textup{(\roman*)},ref=\roman*]
\item\label{it:HIR} The set of transfers that make $\bm{X}_0$ individually rational is
\begin{equation}\label{eq:HIR}
H_{\IR}(\bm{X}_0)
:=\left\{\bm{c}\in\R^n:\ \sum_{i\in [n]} c_i=0,\ \ c_i\le b_i\ \forall i\in [n]\right\},
\end{equation}
and is nonempty. For $\bm{c}\in H$, the post-transfer welfare shares $w_i:=b_i(\bm{X}_0+\bm{c})$ satisfy $w_i=b_i-c_i$ for $i\in [n]$, and $\sum_{i\in [n]} w_i=W$.

\item\label{it:simplex} $H_{\IR}(\bm{X}_0)$ is in bijection with the imputation simplex $\Delta_W:=\{\bm{w}\in\R_+^n:\sum_{i\in [n]} w_i=W\}$ via $w_i=b_i-c_i$ and $c_i=b_i-w_i$ for $i\in[n]$. When $W>0$, the extreme points of $\Delta_W$ are $\{W\bm{e}^{(k)}:k\in [n]\}$, where $\bm{e}^{(k)}$ denotes the $k$th standard basis vector.
\end{enumerate}
\end{theorem}

Theorem~\ref{thm:transfer-IR-geometry} separates the risk-sharing problem into an efficiency part and an implementability part. Once a pure risk allocation is fixed, individual rationality reduces to a linear system in $\R^n$, and the welfare simplex $\Delta_W$ parameterizes admissible transfers. At the vertex $\bm{w}=W\bm{e}^{(k)}$, agent $k$ captures the entire grand-coalition gain while every other agent's participation constraint binds; when $W=0$, the simplex collapses to $\bm{0}$ and the unique IR transfer leaves every agent at their reservation level.

The bijection $c_i=b_i-w_i$ of Theorem~\ref{thm:transfer-IR-geometry}\,(\ref{it:simplex}) turns the welfare simplex into explicit transfer coordinates: $H_{\IR}(\bm{X}_0)$ is itself an $(n-1)$-dimensional simplex that can be written in closed form. The interior of the simplex corresponds to allocations in which every agent receives a strictly positive share.

\begin{corollary}\label{prop:HIR-vertices}
In the setting of Theorem~\ref{thm:transfer-IR-geometry}, suppose $W>0$. Then $H_{\IR}(\bm{X}_0)$ is an $(n-1)$-dimensional simplex with extreme points $\bm{c}^{(k)}:=\bm{b}-W\bm{e}^{(k)}$, $k=1,\dots,n$, and $H_{\IR}(\bm{X}_0)=\operatorname{conv}\{\bm{c}^{(1)},\dots,\bm{c}^{(n)}\}$. The set of IR-efficient allocations in the deterministic-transfer class of $\bm{X}_0$ is $\{\bm{X}_0+\bm{c}:\bm{c}\in H_{\IR}(\bm{X}_0)\}=\operatorname{conv}\{\bm{X}_0+\bm{c}^{(1)},\dots,\bm{X}_0+\bm{c}^{(n)}\}$.
\end{corollary}

Once an efficient transfer class is fixed, the random part of every allocation in $\bm{X}_0 + H$ is determined by efficiency, and deterministic transfers choose a point in the welfare simplex.
The welfare simplex itself $\Delta_W$ is uniquely determined from the characteristic function $\nu$, and so we refer to it as the IR-efficient set. However, when $\Opt(X)$ contains multiple transfer classes, the set of possible transfers $\cup_{\bm{X}_0\in \Opt(X)_\sim} H_{\IR}(\bm{X}_0)$ is a union of simplices, one per transfer-equivalence class, and this union need not be convex. Because of this lack of convexity, the sequel often avoids this possible ambiguity by referring only to the IR-efficient set and constructing the transfers directly through the vertices of $\Delta_W$.

\subsection{Core implementation and existence}\label{ss:core-implement}

The gain game $\nu$ associates with each coalition its aggregate welfare gain, and we now inquire which splits of the grand-coalition gain $\nu([n])$ are stable in the sense that no subgroup can improve by deviating. The classical answer is given by the core.

\begin{definition}[Core]\label{def:core}
The core of the TU game $([n],\nu)$ is
\begin{equation*}
\Core(\nu):=\left\{\bm{w}\in\R^n:\ \sum_{i\in [n]} w_i=\nu([n]),\ \sum_{i\in A} w_i\ge \nu(A)\ \forall A\subseteq [n]\right\}.
\end{equation*}
\end{definition}

If a proposed split lies outside the core, some coalition could ``walk away'' and do better on its own, so the arrangement is unstable. In risk sharing, core stability means that every subgroup receives at least as much benefit as it could achieve by pooling only among its own members, a strong yet natural participation constraint.

The next proposition records the basic correspondence: every core element is realized by a deterministic transfer on an efficient allocation, and the resulting allocation is IR.

\begin{proposition}\label{prop:core-implement}
In the setting of Theorem~\ref{thm:transfer-IR-geometry}, the following hold:
\begin{enumerate}[label=\textup{(\roman*)},ref=\roman*]
\item\label{it:core-nonneg} every $\bm{w}\in\Core(\nu)$ satisfies $w_i\ge 0$ for all $i\in [n]$;
\item\label{it:core-transfer} for any $\bm{w}\in\Core(\nu)$, the transfer
$c_i:=b_i-w_i,$ for $i\in [n]$, belongs to $H_{\IR}(\bm{X}_0)$, and the allocation $\bm{X}:=\bm{X}_0+\bm{c}$ is feasible, efficient, and IR. Its welfare baselines equal $\bm{w}$, meaning $b_i(\bm{X})=w_i$ for all $i\in [n]$.
\end{enumerate}
\end{proposition}

By cash additivity, any core element $\bm{w}$ translates into a transfer via $\bm{c}=\bm{b}-\bm{w}$ on the pure risk allocation $\bm{X}_0$, so coalitional stability reduces to how a fixed total surplus $W$ is split. Figure~\ref{fig:welfare-simplex} illustrates the welfare simplex and the core for three agents.

\begin{figure}[ht]
\centering
\begin{tikzpicture}
\begin{axis}[
  width=5cm, height=5cm,
  xmin=-0.18, xmax=1.12,
  ymin=-0.12, ymax=1.12,
  xlabel={$w_1$},
  ylabel={$w_2$},
  axis lines=middle,
  every axis x label/.style={at={(ticklabel* cs:1.02)}, anchor=west},
  every axis y label/.style={at={(ticklabel* cs:1.02)}, anchor=south},
  xtick=\empty,
  ytick=\empty,
  label style={font=\normalsize},
  clip=false,
  legend style={
    font=\normalsize,
    at={(1.5,0.5)},
    anchor=west,
    draw=gray!50,
    fill=white,
    fill opacity=0.92,
    text opacity=1,
    cells={anchor=west},
  },
]

\addplot[fill=gray!18, draw=none, forget plot]
  coordinates {(0,0) (1,0) (0,1) (0,0)} -- cycle;
\addplot[very thick, black, forget plot]
  coordinates {(0,0) (1,0) (0,1) (0,0)} -- cycle;

\addplot[fill=gray!45, draw=none, forget plot]
  coordinates {
    (0.35, 0)
    (0.70, 0)
    (0.70, 0.30)
    (0.40, 0.60)
    (0.00, 0.60)
    (0.00, 0.35)
    (0.35, 0)
  } -- cycle;

\addplot[very thick, black, forget plot]
  coordinates {(0.35, 0) (0.70, 0)};
\addplot[very thick, black, forget plot]
  coordinates {(0.00, 0.60) (0.00, 0.35)};
\addplot[very thick, black, forget plot]
  coordinates {(0.70, 0.30) (0.40, 0.60)};

\addlegendimage{area legend, fill=gray!18, draw=black, thick}
\addlegendentry{$\Delta_W$}
\addlegendimage{area legend, fill=gray!45, draw=black, thick}
\addlegendentry{$\mathrm{Core}(\nu)$}

\addplot[very thick, blue]
  coordinates {(0.00, 0.35) (0.35, 0)};
\addlegendentry{$w_1\!+\!w_2 = \nu(\{1,2\})$}
\addplot[very thick, red]
  coordinates {(0.40, 0.60) (0.00, 0.60)};
\addlegendentry{$w_2 = W - \nu(\{1,3\})$}
\addplot[very thick, teal]
  coordinates {(0.70, 0) (0.70, 0.30)};
\addlegendentry{$w_1 = W - \nu(\{2,3\})$}

\node[font=\normalsize, anchor=north west]
  at (axis cs:1.0, 0.0)
  {$W\bm{e}^{(1)}$};
\node[font=\normalsize, anchor=south east]
  at (axis cs:0.0, 1.0)
  {$W\bm{e}^{(2)}$};
\node[font=\normalsize, anchor=north east]
  at (axis cs:0.0, 0.0)
  {$W\bm{e}^{(3)}$};


\node[font=\normalsize, rotate=315, anchor=south east]
  at (axis cs:0.62, 0.38) {$w_3 = 0$};

\end{axis}
\end{tikzpicture}
\caption{The welfare simplex $\Delta_W$ for $n=3$ agents, projected onto
$(w_1, w_2)$ coordinates with $w_3 = W - w_1 - w_2$. The outer triangle
is the set of individually rational welfare-gain splits: $w_i \geq 0$
for all $i\in[3]$ and $\sum_{i\in [3]} w_i = W$. Each vertex $W\bm{e}^{(k)}$ assigns the entire
welfare gain $W$ to agent $k$ while all others bind their participation
constraints. The shaded hexagonal region is the core, defined by the
additional coalitional constraints $\sum_{i \in A} w_i \geq \nu(A)$ for
all $A \subsetneq [n]$. The coloured edges indicate the binding coalitional
constraints that define the core boundary.}
\label{fig:welfare-simplex}
\end{figure}
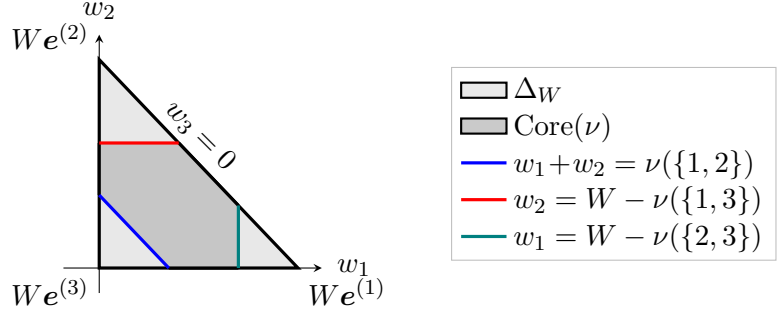
Similarly to the IR-efficient, the core is uniquely determined by the characteristic function $\nu$, and we can construct transfers directly through the vertices of $\Delta_W$.

We now turn to the conditions under which the core is nonempty. In general TU games, the Bondareva-Shapley Theorem states that the core is nonempty if and only if the game is balanced.

\begin{definition}\label{def:balanced}
Let $([n],\nu)$ be a TU game with $\nu(\varnothing)=0$.
\begin{enumerate}[label=\textup{(\roman*)},ref=\roman*]
\item The game $\nu$ is \emph{balanced} if for every $m\ge 1$, every $\lambda_1,\dots,\lambda_m\ge 0$, and every $A_1,\dots,A_m\subseteq [n]$,
$$\sum_{j:\, i\in A_j}\lambda_j=1\;\;\text{for all } i\in [n] \quad\Longrightarrow\quad \sum_{j\in [m]}\lambda_j\nu(A_j)\le\nu([n]).$$
\item The game $\nu$ is \emph{totally balanced} if for every nonempty coalition $B\subseteq [n]$, the restricted game $\nu^B:2^B\to\R$ defined by $\nu^B(A):=\nu(A)$ for subcoalitions $A\subseteq B$ is balanced.
\end{enumerate}
\end{definition}

The first part of the next theorem is the well-known Bondareva-Shapley Theorem; see e.g., \citet[Theorem~4.1]{MarinacciMontrucchio2004Ambiguity} for a proof. The second part immediately follows by definition, and we thus omit the proof.
\begin{theorem}[\citealp{Bondareva1963,Shapley1967}]\label{thm:bondareva-shapley}
Assume $\nu$ is a TU game with $\nu(\varnothing)=0$. Then:
\begin{enumerate}[label=\textup{(\roman*)},ref=\roman*]
\item The core $\Core(\nu)$ is nonempty if and only if $\nu$ is balanced.
\item The core of the cost game in \eqref{eq:coalition-cost} is nonempty if and only if every $m\ge 1$, $\lambda_1,\dots,\lambda_m\ge 0$, and $A_1,\dots,A_m\subseteq [n]$ satisfying
\begin{equation}\label{eq:balanced-weights}
\sum_{j:\, i\in A_j}\lambda_j=1,\qquad \forall\, i\in [n],
\end{equation}
also satisfy
\begin{equation}\label{eq:balancedness-cost}
\sum_{j\in [m]}\lambda_j C(A_j)\ge C([n]).
\end{equation}
\end{enumerate}
\end{theorem}

The cost formulation \eqref{eq:balancedness-cost} is natural because $C(A)$ is an infimum of total evaluations, so balancedness amounts to a consistency requirement across coalition-wise risk-sharing problems. 

The next theorem shows that the convexity of each risk measure $\rho_i$, $i\in[n]$, implies that $C$ (equivalently, $\nu$) is totally balanced.

\begin{theorem}\label{thm:convex-implies-core}
The cost game $C$ is totally balanced, hence $\Core(\nu^B)\neq\varnothing$ for every nonempty $B\subseteq[n]$.
\end{theorem}

The argument used in the proof of Theorem~\ref{thm:convex-implies-core} adapts \citet[Proposition~3.1]{csoka2009stable}, who verify balancedness for risk allocation problems under a single coherent measure by exploiting positive homogeneity and subadditivity. Here, coalition costs are defined via inf-convolution of heterogeneous functionals, so the proof constructs a feasible grand-coalition allocation from optimal subcoalition allocations and controls its cost using only the convexity of each $\rho_i$, $i\in[n]$, without positive homogeneity.\footnote{In the capital-allocation setting, where a single convex risk measure evaluates the coalition aggregate directly, \citet{tsanakas2009split} shows that convexity without positive homogeneity creates incentives for infinite fragmentation, so that the resulting game is not balanced; the inf-convolution formulation sidesteps this because convexity is exploited at the agent level through the optimization over feasible allocations.}

Theorem~\ref{thm:convex-implies-core} is important because it guarantees the existence of stable risk-sharing arrangements. When the core is empty, an efficient risk-sharing arrangement creates a positive aggregate surplus yet remains impossible to implement: some coalition could leave the grand coalition, reallocate internally, and obtain a larger welfare gain. In that case, the problem is not merely to choose a better price, as no deterministic transfer can make the proposed efficient allocation immune to all deviations across subgroups. Theorem~\ref{thm:convex-implies-core} shows that the core is always nonempty for standard risk-averse preferences, and this obstruction disappears. This existence result also clarifies the role of the transfer rules studied in Section~\ref{ss:dual-infconv-pricing} and Appendix~\ref{sec:transfer-rules}. Once the core is known to be nonempty, different transfer rules can be compared as points in the welfare simplex, with the core distinguishing those that are robust to coalitional deviations. A rule that falls in the core implements efficient sharing in a way that is robust to coalitional deviations; a rule outside the core may still be individually rational and efficient, but it is vulnerable to subgroup objections. Total balancedness is stronger still: it guarantees the same nonemptiness property for every possible subcoalition, which is the prerequisite for the entry-monotonicity analysis in Section~\ref{sec:stable-expansion-dual}.

\subsection{Dual pricing allocation}\label{ss:dual-infconv-pricing}

It is now clear that the inf-convolution defining the cost $C$ uniquely determines the welfare simplex $\Delta_W$ and the core $\Core(\nu)$, and that the efficiency of a random allocation $\bm{X}\in \textrm{Opt}(X)$ alone does not guarantee it is implementable. Although the risk-sharing literature has mostly focused on efficient random allocations, implementation also requires a well-specified and meaningful transfer rule.  Appendix~\ref{sec:transfer-rules} collects several such rules; here we focus on the Arrow--Debreu state-price rule, which is central to the main results of the paper.  This rule prices the original endowments with a common Arrow--Debreu state-price functional and uses the resulting cost shares as the target post-transfer evaluations.  It is generated by a dual optimizer of the grand-coalition inf-convolution: the optimizer supports the representative-agent risk measure at the aggregate loss, and its linear part supplies the state prices used to charge each agent for its endowment.

This construction is more specific than an arbitrary supporting price, but more general than the lower-envelope distortion formula that motivates much of the distortion risk measure analysis in Appendix \ref{app:distortion-background}. The aim is to isolate a pricing rule whose individual-rationality and core properties can be directly read off from dual feasibility and Fenchel inequalities, since the same object will later be selected, coalition by coalition, when we study stable expansion in Section~\ref{sec:stable-expansion-dual}. When the individual risk measures are distortion risk measures, the intersection of the dual envelopes yields the envelope associated with lower-envelope distortion, and the construction specializes to the well-known lower-envelope Arrow--Debreu pricing rule, which we recall in Appendix \ref{app:distortion-background}. For general convex risk measures, the same notation also keeps track of the penalty terms that replace zero-penalty coherent support.

The next proposition is the finite-coalition inf-convolution duality in \citet[Theorem~3.6]{barrieu2005infconvolutiona}. Their statement is written for two convex risk measures on bounded positions; here, we formulate the result on $\X=L^p(\Omega,\mathcal F,\P)$ and allow arbitrary finite coalitions. With the individual minimal penalties from \eqref{eq:minimal-penalty}, set
$$ \alpha_A(\pi):=\sum_{i\in A}\alpha_i(\pi),  \qquad \pi\in\mathcal Z,$$
for every nonempty coalition $A\subseteq[n]$; for the grand coalition, write $\alpha:=\alpha_{[n]}$.

\begin{proposition}\label{prop:dual-representation}
For every nonempty coalition $A\subseteq[n]$, the penalty function associated with $\dsquare_{i\in A}\rho_i$ is $\alpha_A$, and
$$ \dsquare_{i\in A}\rho_i (Y) = \sup_{\pi\in\mathcal Z}
    \left\{\E[\pi Y]-\alpha_A(\pi)\right\},
    \qquad Y\in\X.$$
\end{proposition}

The following assumption will sometimes be applied in addition to Assumption~\ref{ass:standing}.
\begin{assumption}\label{ass:dual-pricing}
For every nonempty coalition $A\subseteq[n]$ and every $Y\in\X$, write $\Pi_A(Y):=\argmax_{\pi\in\mathcal Z}\{\E[\pi Y]-\alpha_A(\pi)\}$ for the dual-optimizer set. Assume $\Pi_A(X_A)\neq\varnothing$ for every nonempty coalition $A\subseteq[n]$, and there exists a $\pi_A^\star\in\Pi_A(X_A)$ that is $\sigma(X_A)$-measurable.
\end{assumption}

\begin{theorem}\label{thm:dual-AD-support}
Under Assumption~\ref{ass:dual-pricing}, $\dsquare_{i\in [n]}\rho_i (Y)\ge\E[\pi^\star Y]-\alpha(\pi^\star)$ for every $Y\in\X$, with equality at $Y=X$. Moreover, for every $\bm X\in\Opt(X)$ and every $i\in[n]$, $\rho_i(X_i)-\E[\pi^\star X_i]=-\alpha_i(\pi^\star)=\inf_{Y\in\X}\{\rho_i(Y)-\E[\pi^\star Y]\}$, so that $X_i\in\argmin_{Y\in\X}\{\rho_i(Y)-\E[\pi^\star Y]\}$.
\end{theorem}

The condition $X_i\in\argmin_{Y\in\X}\{\rho_i(Y)-\E[\pi^\star Y]\}$ states that the allocation solves the individual net-cost problem $\inf_{Y\in\X}\{\rho_i(Y)-\E[\pi^\star Y]\}$ against this common price, which implies that Theorem~\ref{thm:dual-AD-support} pins down the Arrow--Debreu pricing measure from a risk-trading equilibrium defined in the sense of \cite{embrechts2018quantilebaseda}. However, one must be careful with this interpretation, as a given pure risk allocation $\bm X_0\in\Opt(X)$ need not be affordable for some agent. Indeed, the set $\Opt(X)$ can contain many pure risk allocations (from different transfer-equivalence classes in $\Opt(X)_{/\sim}$), and not every such allocation needs to be affordable for a given set of budget constraints. This problem cannot be tackled without explicitly modelling such constraints, and so one cannot claim that a particular tuple $(\bm X, \pi^\star)$ necessarily constitutes an equilibrium. 

Theorem~\ref{thm:dual-AD-support} can still be understood as a full decentralization result, where the corresponding allocation would be obtained in an economy where the final positions can be afforded by all agents participating in the pool. More formally, there always exists an economy, potentially fictitious, such that $(\bm X, \pi^\star)$ constitutes an equilibrium.

\begin{theorem}\label{thm:dual-core-allocation}
Under Assumption~\ref{ass:dual-pricing}, fix a nonempty coalition $B\subseteq[n]$, and define $C_i^\star(B):=\E[\pi_B^\star\zeta_i]-\alpha_i(\pi_B^\star)$ and $w_i^\star(B):=\rho_i(\zeta_i)-C_i^\star(B)$ for $i\in B$. Then:
\begin{enumerate}[label=\textup{(\roman*)},ref=\roman*]
\item\label{it:dual-cost-core-selector} $\bm C^\star(B)=(C_i^\star(B))_{i\in B}$ belongs to the cost core of the restricted coalition problem, in that $\sum_{i\in B}C_i^\star(B)=C(B)$, and $\sum_{i\in A}C_i^\star(B)\le C(A)$ for every $\varnothing\ne A\subseteq B$.
\item\label{it:dual-gain-core-allocation} Equivalently, $\bm w^\star(B)=(w_i^\star(B))_{i\in B}\in\Core(\nu^B)$.
\item\label{it:dual-transfer-implementation} For every $\bm X^B\in\Opt_B(X_B)$, the deterministic transfers $c_i^\star(B):=C_i^\star(B)-\rho_i(X_i^B)$, $i\in B$, satisfy $c_i^\star(B)=\E[\pi_B^\star(\zeta_i-X_i^B)]$, $\sum_{i\in B}c_i^\star(B)=0$, and $\rho_i(X_i^B+c_i^\star(B))=C_i^\star(B)$ for $i\in B$.
\end{enumerate}
\end{theorem}

Theorem~\ref{thm:dual-core-allocation} is a selection result, not a new core-existence statement. Total balancedness of $\nu$ is already established in Theorem~\ref{thm:convex-implies-core}, and the additional content is that the dual optimizer $\pi_B^\star$ delivers an explicit core vector and an implementing deterministic transfer for every coalition $B\subseteq[n]$.

\begin{example}[Coherent]\label{ex:coherent-dual}
If, for each $i\in[n]$, $\rho_i$ is a coherent risk measure with dual envelope $\mathcal D_i$, then $\alpha_i=\delta_{\mathcal D_i}$ is the indicator penalty.  Hence, for each nonempty $A\subseteq[n]$, the aggregate penalty is the indicator of $\mathcal D_A:=\cap_{i\in A}\mathcal D_i$, and
$$ \dsquare_{i\in A}\rho_i (Y)=\sup_{\pi\in\mathcal D_A}\E[\pi Y], \quad Y\in\X,
    \qquad C_i^\star(A)=\E[\pi_A^\star\zeta_i], \quad i\in A.$$
This is the dual-based stable allocation rule of \citet{chenhuwang2017stable} in the present notation.
\end{example}

In Appendix~\ref{app:distortion-background}, we specialize Example~\ref{ex:coherent-dual} to distortion risk measures via the lower envelope of the distortions.

\begin{example}[Mean--variance]\label{ex:mean-variance-dual}
For each $i\in[n]$ and loss $Y\in L^2$, let $\rho_i(Y)=\E[Y]+\gamma_i\Var(Y)$, where $\gamma_i>0$.  This functional is convex and cash additive. Working in $L^2$ with the pairing $\langle Y,\pi\rangle=\E[\pi Y]$, the conjugate is finite only on $\{\pi\in L^2:\E[\pi]=1\}$, and the dual representation is
$$ \rho_i(Y) = \sup_{\pi:\ \E[\pi]=1} \left\{\E[\pi Y]-\frac{1}{4\gamma_i}\E[(\pi-1)^2]\right\},
    \qquad i\in[n],\ Y\in L^2.$$
For a nonempty coalition $A\subseteq[n]$, the inf-convolution is again mean--variance: $\dsquare_{i\in A}\rho_i (Y)=\E[Y]+\gamma_A\Var(Y)$, where $\gamma_A:=(\sum_{i\in A}1/\gamma_i)^{-1}$.  The aggregate penalty is $\alpha_A(\pi)=(4\gamma_A)^{-1}\E[(\pi-1)^2]$ for $\E[\pi]=1$, and the unconstrained dual optimizer at $X_A$ is $\pi_A^\star=1+2\gamma_A(X_A-\E[X_A])$.
Consequently,
$$ C_i^\star(A) = \E[\zeta_i] +2\gamma_A\Cov(\zeta_i,X_A) -\frac{\gamma_A^2}{\gamma_i}\Var(X_A), \qquad i\in A.$$ 
Note that mean-variance is not monotone in general, meaning that $\pi_A^\star$ can be signed. This implies that $\pi_A^\star$ is not a probability measure but can still be used as a pricing measure; see also the discussion after Theorem~\ref{thm:dual-AD-support}.
\end{example}

\begin{example}[Entropic]\label{ex:entropic-dual}
For each $i\in[n]$ and loss $Y\in\X$ with $\E[\e^{Y/\iota}]<\infty$, let $\rho_i(Y)=\iota_i\log\E[\e^{Y/\iota_i}]$, where $\iota_i>0$. This functional is convex, cash additive, monotone, and law invariant. Its dual representation is
$$ \rho_i(Y) = \sup_{\pi\ge0,\ \E[\pi]=1} \left\{\E[\pi Y]-\iota_i\E[\pi\log\pi]\right\}, \qquad i\in[n].$$
For a nonempty coalition $A\subseteq[n]$, the inf-convolution is again entropic: $\dsquare_{i\in A}\rho_i (Y)=\iota_A\log\E[\e^{Y/\iota_A}]$, where $\iota_A:=\sum_{i\in A}\iota_i$. The aggregate penalty is $\alpha_A(\pi)=\iota_A\E[\pi\log\pi]$ on $\{\pi\ge0,\ \E[\pi]=1\}$, and the dual optimizer at $X_A$ is $\pi_A^\star=\e^{X_A/\iota_A}/\E[\e^{X_A/\iota_A}]$. With $\d\Q_A^\star/\d\P=\pi_A^\star$,
$$ C_i^\star(A) = \E^{\Q_A^\star}[\zeta_i] -\frac{\iota_i}{\iota_A}\E^{\Q_A^\star}[X_A]    +\iota_i\log\E[\e^{X_A/\iota_A}], \qquad i\in A.$$
\end{example}

\section{Entry monotonicity}\label{sec:stable-expansion-dual}

Stable expansion of risk-sharing pools requires transfer rules that guarantee adding new agents not only leads to a Pareto improvement but also yields a coalitionally stable allocation. This is the fundamental idea underlying \textit{Population Monotonic Allocation Schemes} (PMAS), a refinement of the notion of core due to \cite{sprumont1990population}.  

Since each incumbent's standalone cost $\rho_i(\zeta_i)$ is fixed, entry monotonicity can be checked equivalently on evaluated costs: incumbents must be assigned weakly lower costs in larger coalitions. We therefore formulate the PMAS problem using coalition costs $C(A)=\dsquare_{i\in A}\rho_i(X_A)$ and coalition-wise cost shares for $A \subseteq [n]$. For clarity, the first three subsections emphasize formalism, and we defer applications to the last subsection.

\subsection{PMAS and cost-based entry monotonicity}\label{ss:pmas-cost-allocation-dual}

For a coalition-wise rule $R$, write $C_i^R(A)$ for the evaluated cost assigned to incumbent $i\in A$ when coalition $A$ forms, and write $w_i^R(A):=\rho_i(\zeta_i)-C_i^R(A)$ for the corresponding welfare-gain share. For an efficient transfer rule, we have
$$ \sum_{i\in A}C_i^R(A)=C(A), \qquad A\subseteq[n],$$
or equivalently $\sum_{i\in A}w_i^R(A)=\nu(A)$.

\begin{definition}[Entry monotonicity]\label{def:entry}
A coalition-wise rule $R$ is \emph{entry-monotone} if $w_i^R(B)\ge w_i^R(A)$ whenever $\varnothing\ne A\subseteq B\subseteq[n]$ and $i\in A$.  Equivalently, incumbent evaluated costs are antitone under inclusion: $C_i^R(B)\le C_i^R(A)$ whenever $\varnothing\ne A\subseteq B\subseteq[n]$ and $i\in A$.
\end{definition}

\begin{definition}[PMAS]\label{def:PMAS-transfer}
A coalition-wise rule $R$ is a \emph{Population Monotonic Allocation Scheme} (PMAS) for the gain game $\nu$ if, for every nonempty coalition $B\subseteq[n]$,
\begin{enumerate}[label=\textup{(\roman*)},ref=\roman*]
\item $\bm w^R(B)\in\Core(\nu^B)$, and
\item $w_i^R(B)\ge w_i^R(A)$ whenever $\varnothing\ne A\subseteq B\subseteq[n]$ and $i\in A$.
\end{enumerate}
\end{definition}

It is clear that the game needs to be totally balanced for a PMAS to exist. \cite{sprumont1990population} originally provided the next two results. The first records the relation between entry monotonicity and core stability for efficient coalition-wise rules.

\begin{proposition}\label{prop:entry-not-core}
Suppose a coalition-wise rule $R$ is efficient on every coalition.  If $R$ is entry-monotone, then $\bm w^R(B)\in\Core(\nu^B)$ for every nonempty coalition $B\subseteq[n]$.
\end{proposition}

The author also shows that convex games always guarantee the existence of PMAS, where convexity (supermodularity) means that 
$$ \nu(A)+\nu(B)\le \nu(A\cup B)+\nu(A\cap B), \qquad A,B\subseteq[n].$$
\begin{proposition}\label{prop:convex-PMAS}
If the gain game $\nu$ is convex, then it admits a PMAS. In particular, the coalition-wise Shapley scheme is a PMAS.
\end{proposition}

In practice, risk-sharing games are rarely convex. We are therefore interested in providing alternative sufficient conditions for the existence of PMAS. The next proposition provides a cost criterion that will be used in the rest of the section.
\begin{proposition}\label{prop:dual-premium-pmas}
For each nonempty coalition $B\subseteq[n]$, let $\bm C^R(B)=(C_i^R(B))_{i\in B}$ be a cost-core selector for the restricted cost game on $B$, so that
$$ \sum_{i\in B}C_i^R(B)=C(B), \qquad \sum_{i\in A}C_i^R(B)\le C(A), \quad A\subseteq B.$$
If $C_i^R(B)\le C_i^R(A)$ whenever $\varnothing\ne A\subseteq B\subseteq[n]$ and $i\in A$, then the welfare gains $w_i^R(A)=\rho_i(\zeta_i)-C_i^R(A)$ is a PMAS of $\nu$.
\end{proposition}

The rest of this section focuses on two surplus allocation rules that share the desirable property of being implementable at pool formation without full knowledge of potential participants from pool extensions. Generally, a PMAS is a selection on the coalition lattice, and such a selection typically requires full knowledge of the coalition lattice. However, if one can find ``procedural" rules in the sense that they can be directly applied at entry without knowledge of the potential other pool configurations, then this procedure is essentially a myopic policy like one finds in one-step-ahead decision processes (dynamic programming).

The first allocation rule we analyze is the proportional-cost allocation rule. It requires fixing a set of weights to the agent, and attributing surplus in proportion to these weights. When these weights are chosen procedurally, then the allocation is myopic. Examples include uniform weights or the expected risk of Theorem~\ref{thm:mean-proportional-PMAS}. The second surplus allocation rule we analyze is the Arrow--Debreu pricing; clearly, it is purely procedural and thus myopic. 

\subsection{Proportional-cost allocations}\label{ss:mean-proportional-dual}

We now apply the PMAS criterion to concrete transfer rules. The first family splits each coalition's cost in fixed proportions, providing a simple benchmark that is useful in practice.

\begin{definition}[Proportional-cost allocation]
   Fix weights ${\bm \xi} = (\xi_1, \dots, \xi_n)$ such that $\xi_i > 0$ for $i \in [n]$ and set $\xi_A:=\sum_{i\in A}\xi_i$.  The proportional-cost allocation assigns
\begin{equation}\label{eq:PC-cost-share}
    C_i^{\bm \xi}(A):=\xi_i\frac{C(A)}{\xi_A},
    \qquad i\in A.
\end{equation} 
\end{definition}
This definition is purely cooperative and depends only on coalition costs and the chosen weights.

\begin{theorem}\label{thm:PC-PMAS}
The coalition-wise proportional-cost allocation \eqref{eq:PC-cost-share} is a PMAS if and only if normalized coalition costs are antitone:
\begin{equation}\label{eq:PC-PMAS-condition}
    \frac{C(B)}{\xi_B}\le \frac{C(A)}{\xi_A},
    \qquad \varnothing\ne A\subseteq B\subseteq[n].
\end{equation}
\end{theorem}

PMAS conditions are usually difficult to check directly, and only a small number of tractable constructions are known.  Proposition~\ref{prop:dual-representation} turns the proportional-cost condition \eqref{eq:PC-PMAS-condition} into a comparison of normalized coalition evaluators, which gives a sufficient condition that can be verified in concrete risk-sharing models.  For each nonempty coalition $A\subseteq[n]$, define $R_A^\xi:\X\to(-\infty,\infty]$ by $$R_A^{\bm \xi}(Y):=\dsquare_{i\in A}\rho_i(\xi_A Y) / \xi_A.$$
Under the assumptions of Proposition~\ref{prop:dual-representation}, we have
$$ R_A^{\bm \xi}(Y) = \sup_{\pi\in\mathcal Z} \left\{\E[\pi Y]-\frac{\alpha_A(\pi)}{\xi_A}\right\},$$ 
and thus $C(A)/\xi_A=R_A^\xi(X_A/\xi_A)$.

Because $R_A^{\bm \xi}$ is a positive multiple of $\dsquare_{i\in A}\rho_i$, which is convex, lower semicontinuous, and law invariant (Proposition~\ref{prop:infconv-attainment}), it is automatically \emph{convex-order consistent}: $Y\cx Z$ implies $R_A^{\bm \xi}(Y)\le R_A^{\bm \xi}(Z)$, since law-invariant convex functionals are consistent with the convex order \citep[Theorem~4.3]{bauerle2006stochastic}. The next theorem uses this property to compare normalized coalition costs. 

\begin{theorem}\label{thm:dual-cx-proportional-PMAS}
Suppose that for every $\varnothing\ne A\subseteq B\subseteq[n]$, the normalized losses satisfy $X_B/\xi_B\cx X_A/\xi_A$ and $R_B^{\bm \xi}(X_A/\xi_A)\le R_A^{\bm \xi}(X_A/\xi_A)$.
Then, normalized coalition costs are antitone, and the proportional-cost allocation \eqref{eq:PC-cost-share} is a PMAS.
\end{theorem}

A simple sufficient condition for $R_B^{\bm \xi}\le R_A^{\bm \xi}$, $A\subseteq B$, is the normalized penalty order
$$ \frac{\alpha_B(\pi)}{\xi_B} \ge \frac{\alpha_A(\pi)}{\xi_A}, \qquad \pi\in\mathcal Z.$$
This specializes cleanly in the three running examples.
\begin{itemize}
\item \emph{Coherent risk measures} (Example~\ref{ex:coherent-dual}): $\alpha_A=\delta_{\mathcal D_A}$ with $\mathcal D_A=\cap_{i\in A}\mathcal D_i$, so $\mathcal D_B\subseteq\mathcal D_A$ whenever $A\subseteq B$ and the normalized penalty order holds automatically for any weights $\xi$.
\item \emph{Mean--variance} (Example~\ref{ex:mean-variance-dual}): the order reduces to $\xi_B\gamma_B\le\xi_A\gamma_A$, which holds across the coalition lattice exactly when $\xi_i\gamma_i$ is constant in $i\in[n]$, that is, $\gamma_i=c/\xi_i$ for some $c>0$.
\item \emph{Entropic} (Example~\ref{ex:entropic-dual}): the order reduces to $\iota_B/\xi_B\ge\iota_A/\xi_A$, which holds across the coalition lattice exactly when $\iota_i/\xi_i$ is constant in $i\in[n]$, that is, $\iota_i=c\xi_i$ for some $c>0$.
\end{itemize}
This controls the order of evaluators, not the arguments being evaluated.  The remaining primitive condition is then the convex-order comparison of the normalized aggregate risks, $X_B/\xi_B\cx X_A/\xi_A$.

Convex-order comparisons require equal means, so the relevant proportional weights are effectively mean weights, up to a common scale. Indeed, for $B=A\cup\{i\}$ with $i\notin A$, the equal-means condition $\E[X_B/\xi_B]=\E[X_A/\xi_A]=c$ for some $c\in\R$ gives $\mu_B/\xi_B=\mu_A/\xi_A=c$, hence $\xi_i=\xi_B-\xi_A=\mu_i/c$, so $\xi_i$ is proportional to $\mu_i$. Multiplying all weights by the same constant does not change the allocation, so that we may take $\xi_i=\mu_i:=\E[\zeta_i]$ for $i\in[n]$.

Assume now that $\mu_i>0$ for every $i\in[n]$.  Write $\mu_A:=\sum_{i\in A}\mu_i$, $\bar X_A:=X_A/\mu_A$, and $\bar C(A):=C(A)/\mu_A$.
The mean-proportional allocation is the proportional-cost allocation with $\xi_i=\mu_i$:
\begin{equation}\label{eq:proportional-selector-main}
    C_i^{\bm \mu}(A):=\mu_i\frac{C(A)}{\mu_A},
    \qquad i\in A.
\end{equation}

\begin{theorem}\label{thm:mean-proportional-PMAS}
Suppose that, for every $\varnothing\ne A\subseteq B\subseteq[n]$,
\begin{equation}\label{eq:normalized-cx-main}
    \bar X_B\cx \bar X_A,
\end{equation}
and that $R_B^{\bm \mu}(\bar X_A)\le R_A^{\bm \mu}(\bar X_A)$.
Then $\bar C(B)\le\bar C(A)$ for all $\varnothing\ne A\subseteq B\subseteq[n]$, and the mean-proportional allocation \eqref{eq:proportional-selector-main} is a PMAS.
\end{theorem}

Next, we turn to allocations spanned by a pricing measure.
\subsection{Dual-pricing allocations and conditional-mean criteria}\label{ss:dual-pricing-selector}

The proportional and mean-proportional allocations are cooperative cost-sharing rules.  We now return to pricing. Under Assumption~\ref{ass:dual-pricing}, the associated dual-pricing cost share is
\begin{equation}\label{eq:dual-pricing-cost-share}
    C_i^\star(A) := \E[\pi_A^\star\zeta_i]-\alpha_i(\pi_A^\star), \qquad i\in A.
\end{equation}

\begin{proposition}\label{prop:dual-pricing-PMAS}
Under Assumption~\ref{ass:dual-pricing}, $\bm C^\star(A)$ is a cost-core selector for every nonempty $A\subseteq[n]$. If, in addition, $C_i^\star(B)\le C_i^\star(A)$ whenever $\varnothing\ne A\subseteq B\subseteq[n]$ and $i\in A$, then the welfare gains $w_i^\star(A)=\rho_i(\zeta_i)-C_i^\star(A)$, $i\in A$, is a PMAS.
\end{proposition}

In the coherent-risk-measure setting of Example~\ref{ex:coherent-dual}, \eqref{eq:dual-pricing-cost-share} becomes $C_i^\star(A)=\E[\pi_A^\star\zeta_i]$ for $i\in A$, with $\pi_A^\star\in\argmax_{\pi\in\cap_{j\in A}\mathcal D_j}\E[\pi X_A]$.  For coherent distortion risk measures, Appendix~\ref{app:distortion-background} identifies this rule with lower-envelope pricing $\Q_A^\star$ so that $C_i^\star(A)=\E^{\Q_A^\star}[\zeta_i]$.

The next results give sufficient conditions for this antitonicity condition. For every nonempty coalition $A\subseteq[n]$ and every $i\in A$, define the conditional-mean random variable
$m_{i,A}(X_A):=\E[\zeta_i\mid X_A]$.  For a selected dual optimizer, define $\Gamma_{i,A}^\star:\X\to\R$ by $$\Gamma_{i,A}^\star(Y):=\E[\pi_A^\star Y]-\alpha_i(\pi_A^\star).$$
The evaluator $\Gamma_{i,A}^\star$ is affine in $Y$, but its intercept depends on the individual penalty of agent $i\in A$ at the coalition optimizer $\pi_A^\star$.

\begin{theorem}\label{thm:cmrs-pricing-PMAS}
Suppose Assumption~\ref{ass:dual-pricing} holds and that the selected dual prices are cross-coalition compatible in the sense that $\Gamma_{i,B}^\star(m_{i,B}(X_B))\le\Gamma_{i,A}^\star(m_{i,A}(X_A))$ for every $\varnothing\ne A\subseteq B\subseteq[n]$ and every $i\in A$. Then the dual-pricing welfare gains $w_i^\star(A)=\rho_i(\zeta_i)-C_i^\star(A)$ are a PMAS.
\end{theorem}

The compatibility condition is a price condition on the relevant conditional-mean variables. One sufficient way to verify it is
$$ \E[\pi_B^\star m_{i,B}(X_B)] \le \E[\pi_A^\star m_{i,A}(X_A)] \quad\text{and}\quad \alpha_i(\pi_B^\star)\ge \alpha_i(\pi_A^\star).$$
In coherent models, the penalty inequality is automatic on the feasible dual envelopes.

The next identity gives a density-level way to inspect the conditional-mean components.

\begin{proposition}\label{prop:cmrs-density-identity}
Fix a nonempty coalition $A\subseteq[n]$ and $i\in A$.  Assume $\zeta_i$ is integrable and define the finite signed allocation measure
$$ \eta_{i,A}(D):=\E[\zeta_i\id_{\{X_A\in D\}}], \qquad D\in\mathcal B(\R).$$
Suppose $X_A$ has density $f_A>0$ and $\eta_{i,A}$ has density $\varsigma_{i,A}$. Then $m_{i,A}(s)=\varsigma_{i,A}(s)/f_A(s)$ for $f_A$-a.e.\ $s\in\R$. Consequently, $m_{i,A}$ has a nondecreasing version whenever $\varsigma_{i,A}/f_A$ is nondecreasing.
\end{proposition}

Conditional-mean independence gives a primitive condition for convex-order contraction.

\begin{proposition}\label{prop:cmi-enlargement-cmrs}
Fix nonempty coalitions $A\subseteq B\subseteq[n]$ and $i\in A$. Write $Y:=X_{B\setminus A}$. If
$\E[\zeta_i\mid X_A,Y]=\E[\zeta_i\mid X_A]$, then $m_{i,B}(X_B)\cx m_{i,A}(X_A)$.
\end{proposition}

For the dual-pricing rule, PMAS verification reduces to cross-coalition compatibility of the priced conditional-mean components. The next corollary records a useful convex-order contraction of conditional means.

\begin{corollary}\label{cor:independent-enlargement-cmrs}
Suppose $\zeta_1,\dots,\zeta_n$ are independent and integrable. Then $m_{i,B}(X_B)\cx m_{i,A}(X_A)$ for every $\varnothing\ne A\subseteq B\subseteq[n]$ and every $i\in A$.
\end{corollary}

For general convex risk measures, proportional conditional means alone do not ensure that dual pricing yields mean-proportional cost sharing, because the individual penalty terms need not split in mean proportion. The missing condition is a proportional split of the selected dual penalty. The next theorem states the resulting identification. 

\begin{theorem}\label{thm:proportional-cost-dual-pricing-PMAS}
Assume that $\mu_i>0$ for every $i\in[n]$, and that Assumption~\ref{ass:dual-pricing} holds. If, for every nonempty coalition $A\subseteq[n]$ and every $i\in A$,
\begin{equation}\label{eq:proportional-cmrs-main}
    \E[\zeta_i\mid X_A]=\frac{\mu_i}{\mu_A}X_A
\end{equation}
and
$$\alpha_i(\pi_A^\star)=(\mu_i/\mu_A)\alpha_A(\pi_A^\star),$$
then the dual pricing and the mean-proportional allocation coincide and are a PMAS.
\end{theorem}

The penalty-splitting condition is automatic in coherent models.

\begin{example}[Coherent specialization]\label{ex:coherent-specialization}
Suppose that, for each $i\in[n]$, $\rho_i$ is coherent with dual envelope $\mathcal D_i$.  For every nonempty coalition $A\subseteq[n]$, the individual penalties are indicators and vanish on $\cap_{j\in A}\mathcal D_j$, so the dual-pricing allocation in \eqref{eq:dual-pricing-cost-share} is
$$ C_i^\star(A)=\E[\pi_A^\star\zeta_i], \qquad \pi_A^\star\in\argmax_{\pi\in\cap_{j\in A}\mathcal D_j}\E[\pi X_A], \quad i\in A.$$
Theorem~\ref{thm:proportional-cost-dual-pricing-PMAS} identifies this allocation with the mean-proportional allocation. Appendix~\ref{app:distortion-background} gives the PMAS obtained from dual pricing to distortion risk measures.
\end{example}

Following \citet{blierwong2026laplace}, the condition \eqref{eq:proportional-cmrs-main} can be verified from Laplace transforms rather than conditional densities. Fix a nonempty coalition $A\subseteq[n]$ and $i\in A$.  Assume $\zeta_i\in L^1$, $X_A\in L^1$, and $\mu_A\ne0$. Define
$$ \mathcal L_A(t):=\E[\e^{-tX_A}], \qquad \mathcal L_{i,A}^{\ast}(t):=\E[\zeta_i\e^{-tX_A}],$$
and define the finite signed measures $\eta(D):=\E[\zeta_i\id_{\{X_A\in D\}}]$ and $\beta(D):=\E[X_A\id_{\{X_A\in D\}}]$ for Borel sets $D\in\mathcal B (\R)$.  Assume the Laplace--Stieltjes transforms of $\eta$ and $\beta$ are finite on an interval $t\in(0,\varepsilon)$ for some $\varepsilon>0$, that these transforms uniquely determine the finite signed measures, and that differentiation under the expectation gives $-\mathcal L_A'(t)=\E[X_A\e^{-tX_A}]$ on this interval.

\begin{proposition}\label{prop:laplace-proportionality}
If $\mathcal L_{i,A}^{\ast}(t)=(\mu_i/\mu_A)(-\mathcal L_A'(t))$ for $t\in(0,\varepsilon)$ for some $\varepsilon > 0$, then \eqref{eq:proportional-cmrs-main} holds for this coalition $A$ and agent $i$.
\end{proposition}

Theorem~\ref{thm:proportional-cost-dual-pricing-PMAS} explains why proportional-regression models are useful: the mean-proportional PMAS becomes a pricing PMAS. Outside these cases, mean-proportional allocations remain a cooperative solution rather than an equilibrium concept.

\subsection{PMAS special cases}\label{ss:pmas-corollaries}

The preceding results give abstract conditions for PMAS. We now record special cases where those conditions can be verified, with an emphasis on dependence structures that arise in institutional risk-sharing pools. The first results use the dual-pricing allocation from Section~\ref{ss:dual-pricing-selector}. In other words, these conditions allow for myopic extensions of the pool. If some original pool satisfies the conditions, any new agent that also satisfies them can join, but the original pool members need not know the exact information about the new members at the pool's inception; they only need to agree on a surplus allocation rule. 

We begin with the case of independent endowments. Independence describes many routine insurance and lending settings, including geographically diversified motor or property portfolios, multi-line insurance books, and loan portfolios whose cross-exposures have been hedged or sold down. In this case, Corollary~\ref{cor:independent-enlargement-cmrs} supplies a useful convex-order contraction of conditional means.

\begin{corollary}\label{cor:independent-logconcave-PMAS}
Suppose the endowments $\zeta_1,\dots,\zeta_n$ are independent and integrable and that Assumption~\ref{ass:dual-pricing} holds. If the selected dual prices satisfy the cross-coalition compatibility condition in Theorem~\ref{thm:cmrs-pricing-PMAS}, then the dual-pricing allocation is a PMAS.
\end{corollary}

The next case turns the same idea into a closed-form regression condition. Elliptical distributions extend the Gaussian framework to symmetric heavy tails and are widely used in quantitative risk management, including the multivariate-$t$ models that underlie market-risk capital aggregation, credit risk modelling and when determining the solvency capital requirement in Solvency II \citep{mcneil2015quantitative}. In the elliptical family, the conditional mean $\E[\zeta_i\mid X_A]$ is affine in $X_A$ for every nonempty $A\subseteq[n]$ and every $i\in A$, so the dual-pricing comparison reduces to a scalar comparison of priced aggregate deviations and individual penalties.

\begin{corollary}\label{cor:elliptical-PMAS}
Suppose Assumption~\ref{ass:dual-pricing} and that $\bm{\zeta}=(\zeta_1,\dots,\zeta_n)$ is elliptically distributed with mean vector $\bm{\mu}$, positive-definite covariance matrix $\Sigma$, and finite second moments. For every nonempty coalition $A\subseteq[n]$, set $\mu_A:=\E[X_A]$, $\sigma_A^2:=\Var(X_A)$, and $\beta_{i,A}:=\Cov(\zeta_i,X_A)/\Var(X_A)$ for $i\in A$. Then
$$ C_i^\star(A) = \mu_i + \beta_{i,A}\left(\E[\pi_A^\star X_A]-\mu_A\right) - \alpha_i(\pi_A^\star),
\qquad i\in A.$$
If, for every $\varnothing\ne A\subseteq B\subseteq[n]$ and every $i\in A$,
$$ \beta_{i,B}\left(\E[\pi_B^\star X_B]-\mu_B\right)-\alpha_i(\pi_B^\star) \le    \beta_{i,A}\left(\E[\pi_A^\star X_A]-\mu_A\right)-\alpha_i(\pi_A^\star),$$
then the dual-pricing allocation is a PMAS.
\end{corollary}

The next corollary is useful for risk sharing among homogeneous pools. When $(\zeta_1,\dots,\zeta_n)$ is an exchangeable random vector, the conditional means are equal for every participant, which means that the equal-weight proportional-cost share coincides with each participant's dual-pricing share.

\begin{corollary}\label{cor:exchangeable-PMAS-common}
Suppose Assumption~\ref{ass:dual-pricing}, that $(\zeta_1,\dots,\zeta_n)$ is exchangeable and integrable, and that $\alpha_i(\pi_A^\star)=|A|^{-1}\alpha_A(\pi_A^\star)$ for every nonempty $A\subseteq[n]$ and every $i\in A$. Then $\E[\zeta_i\mid X_A]=X_A/|A|$ for every nonempty $A\subseteq[n]$ and every $i\in A$, and the normalized convex-order condition \eqref{eq:normalized-cx-main} holds. The dual-pricing allocation coincides with the mean-proportional allocation, and it is a PMAS.
\end{corollary}

The penalty-split hypothesis is automatic in the two leading cases: when the risk measures are coherent, the individual penalties vanish at any feasible dual optimizer; and when all agents share the same convex risk measure, the individual penalties are equal, so each equals $|A|^{-1}\alpha_A(\pi_A^\star)$.

Consider the situation in which $(\zeta_1,\dots,\zeta_n)$ is a vector of comonotonic risks, so there is no diversification benefit. This worst-case dependence scenario is relevant for catastrophe, pandemic, and other systemic-risk insurance pools, where individual losses move together, and stable expansion must rely on the entry of additional risk bearers, such as reinsurers, rather than on diversification. Coalition expansion may still lower the distortion used to evaluate shared losses, and this risk-capacity effect alone produces entry monotonicity. As the set of participants grows, the coalition's willingness to bear each layer of risk increases.



\begin{corollary}\label{cor:comonotonic-PMAS}
Suppose Assumption~\ref{ass:dual-pricing} holds, that $(\zeta_1,\dots,\zeta_n)$ are comonotonic, and suppose $\rho_i(0)=0$ for every $i\in[n]$. Assume that each coalition inf-convolution is comonotonic additive on the cone generated by the endowments, in the sense that, for every pair of nonempty coalitions $A\subseteq B\subseteq[n]$,
\begin{equation}\label{eq:comonotonic-additive-inf-convolution}
    \dsquare_{j\in B}\rho_j(X_A)=\sum_{i\in A}\dsquare_{j\in B}\rho_j(\zeta_i).
\end{equation}
For each nonempty coalition $A\subseteq[n]$ and every $i\in A$, define $w_i(A):=\rho_i(\zeta_i)-\dsquare_{j\in A}\rho_j(\zeta_i)$. Then $(\bm{w}(A))_{A\subseteq[n]}$ is a PMAS.
\end{corollary}

Note that comonotonic additivity of the inf-convolution in \eqref{eq:comonotonic-additive-inf-convolution} holds for distortion risk measures from Appendix~\ref{app:distortion-background}, which includes the most relevant risk measures in practice.

We collect further examples of conditions for PMAS in Appendix~\ref{app:additional-pmas-examples}. Among others, we provide PMAS conditions for risk-sharing with mean--variance and entropic risk measures; proportional conditional-mean examples based on convolution semigroups; Dirichlet--Liouville radial risks; conditional-convolution mixtures; and an example showing that a mean-proportional cost can be PMAS without coinciding with Arrow--Debreu pricing as in Theorem \ref{thm:proportional-cost-dual-pricing-PMAS}.

\section{Conclusion}\label{sec:discussion}

We analyze how efficient risk sharing under cash-additive evaluations can be implemented in practice. Computing efficient allocations through inf-convolution is only half the problem: without a rule that determines deterministic transfers, the resulting allocation is not implementable, and the exercise has little practical content. Cash additivity makes the distributional step explicit and casts the problem as a transferable-utility game, allowing us to draw directly on the cooperative game literature. We immediately obtain (Theorem~\ref{thm:transfer-IR-geometry}) that individually rational transfers form a polyhedron in each efficient transfer class, which is homeomorphic to the welfare simplex. 

We then show that the gain game is totally balanced whenever the agents are risk averse (Theorem~\ref{thm:convex-implies-core}), so every coalition has a nonempty core and stable transfers always exist. Total balancedness is a necessary condition for the existence of PMAS, and the aforementioned result motivates our analysis of coalition expansion. 

Since risk-sharing games are rarely convex, we aim to provide novel sufficient conditions to guarantee the existence of a PMAS. We focus on two allocation rules, the proportional-cost allocation rule and the (dual-)pricing allocation rule. Both rules have the advantage of being myopic, in that they do not require full knowledge of all potential future participants at inception. The guiding principle is that proportional-cost allocations are PMAS when normalized coalition risks decrease in the convex order (Theorem~\ref{thm:dual-cx-proportional-PMAS}), and that pricing allocations are PMAS under cross-coalition compatibility of the priced conditional means (Theorem~\ref{thm:cmrs-pricing-PMAS}). These novel PMAS conditions are often satisfied in many real-world applications, such as independent risks (Corollary~\ref{cor:independent-logconcave-PMAS}), elliptical risks (Corollary~\ref{cor:elliptical-PMAS}), and comonotonic risks (Corollary~\ref{cor:comonotonic-PMAS}).

\section*{Acknowledgements}

The authors thank Ruodu Wang for his very helpful discussions and comments. CBW acknowledges financial support from the Natural Sciences and Engineering Research Council of Canada (RGPIN-2025-06879). 


\bibliographystyle{apalike}
\bibliography{ref}

\appendix

\section{Proofs of the main text}\label{app:main-text-proofs}

\begin{proof}[Proof of Proposition~\ref{prop:infconv-attainment}]
Fix a nonempty coalition $A\subseteq[n]$. Applying the loss-convention form of \citet[Theorem~2.5]{filipovic2008optimal} to the restricted family $(\rho_i)_{i\in A}$ gives that $\dsquare_{i\in A}\rho_i$ is proper, cash additive, convex, lower semicontinuous, law invariant, and exact. For every $Y\in\X$, Proposition~\ref{prop:comonotonic-optimality} gives that an optimal allocation of $Y$ is comonotonic. Taking $Y=X_A$ gives a comonotonic element of $\Opt_A(X_A)$.

It remains only to prove that $C(A)$ is finite. Since $\zeta_i\in\dom\rho_i$ for every $i\in A$, the no-sharing allocation obtained by assigning $\zeta_i$ to each $i\in A$ is feasible, so
$C(A)\le \sum_{i\in A}\rho_i(\zeta_i)<\infty.$
On the other hand, lower semicontinuity, convexity, law invariance, and cash additivity imply the standard lower bound
$\rho_i(Y)\ge \E[Y]+\rho_i(0)$ for $Y\in\X.$
Hence, for any feasible allocation $\bm{Y}\in\A_{|A|}(X_A)$, we have
$ \sum_{i\in A}\rho_i(Y_i) \ge \E[X_A]+\sum_{i\in A}\rho_i(0)>-\infty$
and $C(A)\in\R$.
\end{proof}

\begin{proof}[Proof of Lemma~\ref{lemma:gain-monotone}]
For $A=\{i\}$, feasibility forces $Y_i=X_{\{i\}}=\zeta_i$, so $C(\{i\})=\rho_i(\zeta_i)$ and $\nu(\{i\})=0$, showing (i). For (ii), the empty coalition has zero gain by definition; if $A$ is nonempty, the no-sharing allocation $Y_i=\zeta_i$ for $i\in A$ is feasible, hence $C(A)\le \sum_{i\in A}\rho_i(\zeta_i)$ and $\nu(A)\ge 0$. For (iii), fix $A\subseteq B\subseteq[n]$. If $A=\varnothing$, the claim follows from (ii). Otherwise choose $\bm X^A\in\Opt_A(X_A)$ and extend it to a feasible allocation for $B$ by setting $X_j^A:=\zeta_j$ for $j\in B\setminus A$. This yields
$$C(B)\le \sum_{i\in A}\rho_i(X_i^A)+\sum_{j\in B\setminus A}\rho_j(\zeta_j)=C(A)+\sum_{j\in B\setminus A}\rho_j(\zeta_j),$$
and therefore $\nu(B)=\sum_{k\in B}\rho_k(\zeta_k)-C(B)\ge \nu(A)$.
Finally, (iv) holds because $\nu([n])=\sum_{i\in[n]}\rho_i(\zeta_i)-C([n])$ and, under the hypotheses of Proposition~\ref{prop:infconv-attainment}, each $\rho_i(\zeta_i)$, $i\in[n]$, and $C([n])$ are finite.
\end{proof}

\begin{proof}[Proof of Theorem~\ref{thm:transfer-IR-geometry}]
To prove~\eqref{it:HIR}, let $\bm{c}\in H$. Individual rationality for $\bm{X}_0+\bm{c}$ requires $\rho_i(X_{0,i}+c_i)\le \rho_i(\zeta_i)$ for all $i\in[n]$. By cash additivity, this is equivalent to $c_i\le b_i$, which gives \eqref{eq:HIR}. The identity $w_i=b_i-c_i$ follows by substitution; since $\bm{X}_0$ is efficient, $\sum_{i\in [n]} b_i=\sum_{i\in [n]} \rho_i(\zeta_i)-C([n])=W$, so $\sum_{i\in [n]} w_i=W$ when $\sum_{i\in [n]} c_i=0$. Nonemptiness follows from $W\ge 0$ (Lemma~\ref{lemma:gain-monotone}\,(ii)): choose any $\bm{w}\in\Delta_W$ and set $c_i:=b_i-w_i$ for $i\in[n]$. To prove~\eqref{it:simplex}, the bijection is immediate from~\eqref{it:HIR}, and the extreme points of the standard simplex $\Delta_W$ are its vertices $W\bm{e}^{(k)}$.
\end{proof}

\begin{proof}[Proof of Corollary~\ref{prop:HIR-vertices}]
The affine bijection $T(\bm{w}):=\bm{b}-\bm{w}$ maps $\Delta_W$ onto $H_{\IR}(\bm{X}_0)$ by Theorem~\ref{thm:transfer-IR-geometry}\,(\ref{it:simplex}). The vertices $W\bm{e}^{(k)}$ of $\Delta_W$ map to $\bm{b}-W\bm{e}^{(k)}=\bm{c}^{(k)}$, and affine bijections preserve extreme points. The allocation-space identity follows by translation by $\bm{X}_0$.
\end{proof}

\begin{proof}[Proof of Proposition~\ref{prop:core-implement}]
To prove~\eqref{it:core-nonneg}, note that the core constraint for the singleton coalition $\{i\}$ reads $w_i\ge \nu(\{i\})=0$, since the gain game is zero-normalized. To prove~\eqref{it:core-transfer}, observe that $\sum_{i\in [n]} w_i=\nu([n])=W=\sum_{i\in [n]} b_i$ gives $\sum_{i\in [n]} c_i=0$, so $\bm{c}\in H$. Since $w_i\ge 0$ for every $i\in[n]$, we have $c_i=b_i-w_i\le b_i$, hence $\bm{c}\in H_{\IR}(\bm{X}_0)$ by \eqref{eq:HIR}. By cash additivity and $\sum_{i \in [n]} c_i=0$, $\sum_{i\in [n]}\rho_i(X_{0,i}+c_i)=\sum_{i \in [n]}\rho_i(X_{0,i})=C([n])$, so $\bm{X}_0+\bm{c}$ remains optimal. Individual rationality holds because $\rho_i(X_{0,i}+c_i)=\rho_i(\zeta_i)-w_i\le \rho_i(\zeta_i)$ for every $i\in[n]$. Finally, $b_i(\bm{X})=b_i-c_i=w_i$ for every $i\in[n]$. 
\end{proof}

\begin{proof}[Proof of Theorem~\ref{thm:convex-implies-core}]
Fix $m\in\N$, balanced weights $\lambda_1,\dots,\lambda_m\ge 0$, and coalitions $A_1,\dots,A_m\subseteq [n]$ satisfying \eqref{eq:balanced-weights}. Terms with $A_j=\varnothing$ may be discarded, since they do not affect \eqref{eq:balanced-weights} and $C(\varnothing)=0$. For each remaining $j$, choose $\bm X^{A_j}\in\Opt_{A_j}(X_{A_j})$, and set $X_i^{A_j}:=0$ for $i\notin A_j$. Define $Z_i:=\sum_{j:\, i\in A_j} \lambda_j X_i^{A_j}$ for each $i\in [n]$. Because $\sum_{j:\, i\in A_j}\lambda_j=1$ by \eqref{eq:balanced-weights}, each $Z_i$ is a convex combination; interchanging the sums over $i$ and $j$ confirms $\sum_{i\in [n]} Z_i = \sum_{j\in [m]}\lambda_j X_{A_j}=X$, so $\bm{Z}\in\A_n(X)$. Convexity of each $\rho_i$ gives $\rho_i(Z_i)\le \sum_{j:\,i\in A_j}\lambda_j\rho_i(X_i^{A_j})$; summing over $i\in[n]$ and interchanging the order of summation yields
$$C([n]) \le \sum_{i\in [n]}\rho_i(Z_i) \le \sum_{j\in [m]}\lambda_j \sum_{i\in A_j}\rho_i(X_i^{A_j}) = \sum_{j\in [m]}\lambda_j C(A_j).$$
This is \eqref{eq:balancedness-cost}. For any nonempty $B\subseteq [n]$, the restricted game $C^B$ inherits the same hypotheses, so the identical argument gives balancedness of $C^B$; hence $C$ is totally balanced.
\end{proof}

\begin{proof}[Proof of Proposition~\ref{prop:dual-representation}]
Fix a nonempty coalition $A\subseteq[n]$. For a two-agent coalition,
\citet[Theorem~3.6]{barrieu2005infconvolutiona} identifies the minimal
penalty of the closed inf-convolution with the sum of the individual minimal
penalties, after translating their gain convention into the present loss
convention. Iterating this binary identity over the finite coalition $A$
shows that the minimal penalty associated with $\dsquare_{i\in A}\rho_i$ is
$\sum_{i\in A}\alpha_i=\alpha_A .$
Since $\dsquare_{i\in A}\rho_i$ is proper, convex, and
$\sigma(\X,\mathcal Z)$-lower semicontinuous, the Fenchel--Moreau Theorem in the dual pair $(\X,\mathcal Z)$ gives its
representation by its minimal penalty. Substituting the preceding penalty yields the dual formula.
\end{proof}

\begin{proof}[Proof of Theorem~\ref{thm:dual-AD-support}]
The support inequality follows from Proposition~\ref{prop:dual-representation} because $\pi^\star$ is feasible in the supremum defining $\dsquare_{i\in [n]}\rho_i (Y)$. Equality at $X$ is the dual optimality of $\pi^\star$ for $X$ under Assumption~\ref{ass:dual-pricing}.

Let $\bm X\in\Opt(X)$. For each $i\in[n]$, \eqref{eq:individual-dual-representation} gives $\rho_i(X_i)\ge\E[\pi^\star X_i]-\alpha_i(\pi^\star)$, and summing over $i\in[n]$ yields $\sum_{i\in [n]}\rho_i(X_i)\ge\E[\pi^\star X]-\sum_{i=1}^n\alpha_i(\pi^\star)=\dsquare_{i\in [n]}\rho_i (X)$. The left-hand side equals $\dsquare_{i\in [n]}\rho_i (X)$ because $\bm X$ is primal optimal, so equality holds in the sum and therefore in every individual Fenchel inequality. By \eqref{eq:individual-dual-representation}, $\rho_i(Y)-\E[\pi^\star Y]\ge-\alpha_i(\pi^\star)$ for every $Y\in\X$ and every $i\in[n]$, so $\inf_{Y\in\X}\{\rho_i(Y)-\E[\pi^\star Y]\}=-\alpha_i(\pi^\star)$, attained at $X_i$.
\end{proof}

\begin{proof}[Proof of Theorem~\ref{thm:dual-core-allocation}]
Fix nonempty $B\subseteq[n]$. Dual optimality of $\pi_B^\star$ at $X_B$ gives $\sum_{i\in B}C_i^\star(B)=\E[\pi_B^\star X_B]-\sum_{i\in B}\alpha_i(\pi_B^\star)=C(B)$. For $\varnothing\ne A\subseteq B$, the same $\pi_B^\star$ is feasible in the coalition-$A$ dual problem because $\alpha_i(\pi_B^\star)<\infty$ for $i\in B$, hence $\sum_{i\in A}C_i^\star(B)=\E[\pi_B^\star X_A]-\sum_{i\in A}\alpha_i(\pi_B^\star)\le\sup_{\pi\in\mathcal Z}\{\E[\pi X_A]-\sum_{i\in A}\alpha_i(\pi)\}=C(A)$. Subtracting the autarky baselines $\sum_{i\in A}\rho_i(\zeta_i)$ gives \textup{(\ref{it:dual-gain-core-allocation})}.

For \textup{(\ref{it:dual-transfer-implementation})}, fix $\bm X^B\in\Opt_B(X_B)$. The Fenchel-equality argument of Theorem~\ref{thm:dual-AD-support}, applied to the coalition-$B$ problem with optimizer $\pi_B^\star$, gives $\rho_i(X_i^B)=\E[\pi_B^\star X_i^B]-\alpha_i(\pi_B^\star)$ for $i\in B$, so $c_i^\star(B)=C_i^\star(B)-\rho_i(X_i^B)=\E[\pi_B^\star(\zeta_i-X_i^B)]$. Summing over $i\in B$ gives zero because $\sum_{i\in B}\zeta_i=\sum_{i\in B}X_i^B=X_B$, and cash additivity gives $\rho_i(X_i^B+c_i^\star(B))=\rho_i(X_i^B)+c_i^\star(B)=C_i^\star(B)$.
\end{proof}

\begin{proof}[Proof of Proposition~\ref{prop:entry-not-core}]
Fix $\varnothing\ne A\subseteq B\subseteq[n]$.  Entry monotonicity gives $w_i^R(B)\ge w_i^R(A)$ for every $i\in A$, hence
$$ \sum_{i\in A}w_i^R(B) \ge \sum_{i\in A}w_i^R(A) = \nu(A),$$
where the equality is efficiency on coalition $A$.  Efficiency on $B$ gives $\sum_{i\in B}w_i^R(B)=\nu(B)$.  Thus $\bm w^R(B)$ belongs to the gain core of the restricted game on $B$.
\end{proof}

\begin{proof}[Proof of Proposition~\ref{prop:dual-premium-pmas}]
For every nonempty $B\subseteq[n]$, efficiency follows from the cost-core identities.  For $A\subseteq B$,
$$ \sum_{i\in A}w_i^R(B) = \sum_{i\in A}\rho_i(\zeta_i)-\sum_{i\in A}C_i^R(B) 
\ge \sum_{i\in A}\rho_i(\zeta_i)-C(A) = \nu(A),$$
so $\bm w^R(B)\in\Core(\nu^B)$.  Finally,
$w_i^R(B)-w_i^R(A)=C_i^R(A)-C_i^R(B)\ge0$, which is entry monotonicity.
\end{proof}

\begin{proof}[Proof of Theorem~\ref{thm:PC-PMAS}]
If \eqref{eq:PC-PMAS-condition} holds, then for $A\subseteq B$,
$$ \sum_{i\in A}C_i^{\bm \xi}(B) = \xi_A\frac{C(B)}{\xi_B} \le \xi_A\frac{C(A)}{\xi_A} = C(A),$$
so the allocation rule is cost-core selecting.  Notice that $C_i^{\bm \xi}(B)\le C_i^{\bm \xi}(A)$ for every $i\in A$, and Proposition~\ref{prop:dual-premium-pmas} gives the PMAS property.

Conversely, if the proportional-cost allocation is a PMAS, then entry monotonicity gives $C_i^{\bm \xi}(B)\le C_i^{\bm \xi}(A)$ for any $i\in A\subseteq B$.  Since $\xi_i>0$, this is \eqref{eq:PC-PMAS-condition}.
\end{proof}

\begin{proof}[Proof of Theorem~\ref{thm:dual-cx-proportional-PMAS}]
For $A\subseteq B$,
$$ \frac{C(B)}{\xi_B} = R_B^{\bm \xi}\left(\frac{X_B}{\xi_B}\right) \le R_B^{\bm \xi}\left(\frac{X_A}{\xi_A}\right) \le R_A^{\bm \xi}\left(\frac{X_A}{\xi_A}\right) = \frac{C(A)}{\xi_A}.$$
The first inequality is convex-order consistency; the second is the evaluator order.  Theorem~\ref{thm:PC-PMAS} gives the PMAS conclusion.
\end{proof}

\begin{proof}[Proof of Theorem~\ref{thm:mean-proportional-PMAS}]
This is Theorem~\ref{thm:dual-cx-proportional-PMAS} with $\xi_i=\mu_i$.  Explicitly,
$$\bar C(B) = R_B^{\bm \mu}(\bar X_B) \le R_B^{\bm \mu}(\bar X_A) \le R_A^{\bm \mu}(\bar X_A) = \bar C(A).$$
Theorem~\ref{thm:PC-PMAS} then gives the PMAS property.
\end{proof}

\begin{proof}[Proof of Proposition~\ref{prop:dual-pricing-PMAS}]
The cost-core selector property is Theorem~\ref{thm:dual-core-allocation}\,\textup{(\ref{it:dual-cost-core-selector})} applied to each nonempty $A\subseteq[n]$. Under the antitonicity condition, Proposition~\ref{prop:dual-premium-pmas} gives the PMAS conclusion.
\end{proof}

\begin{proof}[Proof of Theorem~\ref{thm:cmrs-pricing-PMAS}]
By $\sigma(X_A)$-measurability and the tower property of conditional expectations, we have
$$ C_i^\star(A) =\E[\pi_A^\star\zeta_i]-\alpha_i(\pi_A^\star) = \E[\pi_A^\star m_{i,A}(X_A)]-\alpha_i(\pi_A^\star) = \Gamma_{i,A}^\star\left(m_{i,A}(X_A)\right).$$
The compatibility condition gives $C_i^\star(B)\le C_i^\star(A)$ for every incumbent $i\in A\subseteq B$. Proposition~\ref{prop:dual-pricing-PMAS} gives the PMAS conclusion.
\end{proof}

\begin{proof}[Proof of Proposition~\ref{prop:cmrs-density-identity}]
For every bounded Borel function $\varphi:\R\to\R$ it holds
$$ \int \varphi(s)\varsigma_{i,A}(s)\,\d s = \E[\zeta_i \varphi(X_A)] = \E[m_{i,A}(X_A)\varphi(X_A)] = \int \varphi(s)m_{i,A}(s)f_A(s)\,\d s.$$
It follows that $\varsigma_{i,A}=m_{i,A}f_A$ Lebesgue-a.e., which gives the density identity.  The monotonicity statement follows immediately.
\end{proof}

\begin{proof}[Proof of Proposition~\ref{prop:cmi-enlargement-cmrs}]
For every bounded Borel function $\varphi:\R\to\R$ it holds
$$ \E[\varphi(X_B)m_{i,A}(X_A)] =\E[\varphi(X_A+Y)\E[\zeta_i\mid X_A]] =  \E[\varphi(X_A+Y)\E[\zeta_i\mid X_A,Y]] = \E[\varphi(X_B)\zeta_i].$$
Thus
$$ \E[m_{i,A}(X_A)\mid X_B] = \E[\zeta_i\mid X_B] = m_{i,B}(X_B).$$
Jensen's inequality gives $m_{i,B}(X_B)\cx m_{i,A}(X_A)$.
\end{proof}

\begin{proof}[Proof of Corollary~\ref{cor:independent-enlargement-cmrs}]
For $A\subseteq B\subseteq[n]$ and $i\in A$, $X_{B\setminus A}$ is independent of $(\zeta_i,X_A)$, so the conditional-mean-independence condition in Proposition~\ref{prop:cmi-enlargement-cmrs} holds, which implies the stated result.
\end{proof}

\begin{proof}[Proof of Theorem~\ref{thm:proportional-cost-dual-pricing-PMAS}]
Fix a nonempty coalition $A\subseteq[n]$. For every $i\in A$, the tower property, \eqref{eq:proportional-cmrs-main}, and $\sigma(X_A)$-measurability of $\pi_A^\star$ give
$$ \E[\pi_A^\star\zeta_i] = \E\left[\pi_A^\star\E[\zeta_i\mid X_A]\right] = \frac{\mu_i}{\mu_A}\E[\pi_A^\star X_A].$$
Using the proportional penalty split,
$$ C_i^\star(A) = \frac{\mu_i}{\mu_A}\E[\pi_A^\star X_A] - \frac{\mu_i}{\mu_A}\alpha_A(\pi_A^\star)
    = \frac{\mu_i}{\mu_A}\left(\E[\pi_A^\star X_A]-\alpha_A(\pi_A^\star)\right) = \frac{\mu_i}{\mu_A}C(A),$$
where the last equality is the dual optimality of $\pi_A^\star$. Thus, the dual-pricing allocation coincides with the mean-proportional allocation \eqref{eq:proportional-selector-main}.

For the PMAS property, $\bm C^\star(B)$ is a cost-core selector for every nonempty $B\subseteq[n]$ (Proposition~\ref{prop:dual-pricing-PMAS}), so $\sum_{i\in A}C_i^\star(B)\le C(A)$ whenever $\varnothing\ne A\subseteq B\subseteq[n]$. Substituting the proportional form $C_i^\star(B)=(\mu_i/\mu_B)C(B)$ gives $(\mu_A/\mu_B)C(B)\le C(A)$, that is, $C(B)/\mu_B\le C(A)/\mu_A$. This is the antitonicity condition \eqref{eq:PC-PMAS-condition} with $\xi_i=\mu_i$, so Theorem~\ref{thm:PC-PMAS} shows that the mean-proportional allocation is a PMAS.
\end{proof}

\begin{proof}[Proof of Proposition~\ref{prop:laplace-proportionality}]
The transforms $\mathcal L_{i,A}^{\ast}$ and $-\mathcal L_A'$ are the Laplace--Stieltjes transforms of $\eta$ and $\beta$, respectively. By the uniqueness of Laplace--Stieltjes transforms for finite measures, the assumed proportionality implies
$$ \E[\zeta_i\id_{\{X_A\in D\}}] = \frac{\mu_i}{\mu_A}\E[X_A\id_{\{X_A\in D\}}]$$
for all Borel sets $D\subseteq\R$, which is equivalent to \eqref{eq:proportional-cmrs-main} for this coalition $A$ and agent $i$.
\end{proof}


\begin{proof}[Proof of Corollary~\ref{cor:independent-logconcave-PMAS}]
The conclusion follows from Theorem~\ref{thm:cmrs-pricing-PMAS}, since Assumption~\ref{ass:dual-pricing} and cross-coalition compatibility are exactly its hypotheses.
\end{proof}

\begin{proof}[Proof of Corollary~\ref{cor:elliptical-PMAS}]
By the linear regression property of elliptically distributed vectors \citep{cambanis1981theory},
$$ \E[\zeta_i\mid X_A] = \mu_i+\beta_{i,A}(X_A-\mu_A).$$
Since finite dual penalties under cash additivity imply $\E[\pi_A^\star]=1$,
$$ \E[\pi_A^\star\zeta_i] = \mu_i+\beta_{i,A}\left(\E[\pi_A^\star X_A]-\mu_A\right).$$
Subtracting $\alpha_i(\pi_A^\star)$ gives the displayed cost-share formula. The cross-coalition inequality in the statement gives $C_i^\star(B)\le C_i^\star(A)$ whenever $\varnothing\ne A\subseteq B\subseteq[n]$ and $i\in A$. Proposition~\ref{prop:dual-pricing-PMAS} gives the PMAS conclusion.
\end{proof}

\begin{proof}[Proof of Corollary~\ref{cor:exchangeable-PMAS-common}]
Exchangeability gives equal conditional means inside each coalition:
$\E[\zeta_i\mid X_A]=X_A/|A|$ for every nonempty $A\subseteq[n]$
and every $i\in A$.
By the $\sigma(X_A)$-measurability of $\pi_A^\star$ and the tower property,
$\E[\pi_A^\star\zeta_i]=|A|^{-1}\E[\pi_A^\star X_A]$, and the penalty split gives
$$C_i^\star(A)=\E[\pi_A^\star\zeta_i]-\alpha_i(\pi_A^\star)=|A|^{-1}\bigl(\E[\pi_A^\star X_A]-\alpha_A(\pi_A^\star)\bigr)=|A|^{-1}C(A),$$
where the last equality is dual optimality of the selected optimizer for
coalition $A$. Thus, the dual-pricing allocation coincides with the equal-split rule, which equals the mean-proportional allocation because exchangeability gives equal means.

For the PMAS property, $\bm C^\star(B)$ is a cost-core selector for every nonempty $B\subseteq[n]$ (Proposition~\ref{prop:dual-pricing-PMAS}), so $\sum_{i\in A}C_i^\star(B)\le C(A)$ whenever $\varnothing\ne A\subseteq B\subseteq[n]$. Substituting $C_i^\star(B)=C(B)/|B|$ gives $|A|\,C(B)/|B|\le C(A)$, that is, $C(B)/|B|\le C(A)/|A|$, the antitonicity condition \eqref{eq:PC-PMAS-condition} with $\xi_i=1$; Theorem~\ref{thm:PC-PMAS} gives the PMAS property.

Finally, for $\varnothing\ne A\subseteq B\subseteq[n]$, let $J$ be a uniformly chosen $|A|$-subset of $B$, independent of the risks. Then $|A|^{-1}\sum_{j\in J}\zeta_j$ has the same distribution as $X_A/|A|$, and
$$\E\left[ |A|^{-1}\sum_{j\in J}\zeta_j \,\middle|\, (\zeta_j)_{j\in B}\right] = \frac{X_B}{|B|}.$$
Jensen's inequality gives the normalized convex-order condition \eqref{eq:normalized-cx-main}.
\end{proof}

\begin{proof}[Proof of Corollary~\ref{cor:comonotonic-PMAS}]
For every nonempty coalition $B\subseteq[n]$, take the associated cost share of agent $i\in B$ to be $\dsquare_{j\in B}\rho_j(\zeta_i)$. The displayed comonotonic additivity assumption yields efficiency when $A=B$. If $A\subseteq B$, coalition monotonicity of normalized inf-convolutions gives $\dsquare_{j\in B}\rho_j(X_A)\le \dsquare_{j\in A}\rho_j(X_A)$, since the larger coalition can use an optimal $A$-allocation and assign zero to the agents in $B\setminus A$. Hence
$$\sum_{i\in A}\dsquare_{j\in B}\rho_j(\zeta_i) = \dsquare_{j\in B}\rho_j(X_A) \le \dsquare_{j\in A}\rho_j(X_A) = C(A).$$
Thus, the selector is cost-core selecting. If $A\subseteq B$ and $i\in A$, the same coalition monotonicity gives $\dsquare_{j\in B}\rho_j(\zeta_i)\le\dsquare_{j\in A}\rho_j(\zeta_i)$. Proposition~\ref{prop:dual-premium-pmas} applies.
\end{proof}

\section{Distortion risk measures}\label{app:distortion-background}

The main stability results from the main text require only convex risk measures, but the case of distortion risk measures is useful because it yields explicit formulas for the inf-convolution and dual-price objects. This appendix specializes some of the results from Section \ref{sec:main-results} to distortion risk measures. 

For this Appendix, we make the following assumption.
\begin{assumption}\label{ass:distortion}
In addition to Assumption \ref{ass:standing} and \ref{ass:dual-pricing}, assume that for each $i\in[n]$, agent $i$ evaluates losses by a concave distortion risk measure as in \eqref{eq:choquet-integral}. That is, $\rho_i=\rho_{h_i}$ for an increasing, concave distortion $h_i:[0,1]\to[0,1]$ with $h_i(0)=0$ and $h_i(1)=1$.
\end{assumption}

Under Assumption~\ref{ass:distortion}, each risk measure $\rho_{h_i}, i \in [n]$ is a coherent distortion risk measure. We record two standard representations used in the examples. For integer-valued $Y$ supported on $\{0,\ldots,K\}$ for some $K\in\N$, the tail integral in \eqref{eq:choquet-integral} reduces to $\rho_h(Y)=\sum_{\ell=1}^K h(\P(Y\ge \ell))$. For a bounded random variable $Y$, the same functional can also be written in quantile form as $\rho_h(Y)=\int_0^1 \VaR_{1-u}(Y)\,\d h(u)$; see, e.g., \cite{dhaene2012remarks} for details.

We use the following tail-integral identity when passing from lower-envelope distortions to explicit pricing kernels. It is a special case of Lemma 3 of \citet{lauzier2025risk}.
\begin{lemma}\label{lemma:lauzier-lemma3}
Let $h\in\mathcal{H}$, let $X\in\X$ be bounded from below, and let $f:\R\to\R$ be increasing and Lipschitz with right-derivative $g$. Then
\begin{equation*}
\rho_h\big(f(X)\big)
=f(0)+\int_{0}^{\infty} g(x)\,h(\P(X>x))\,\d x
+\int_{-\infty}^{0} g(x)\,\big(h(\P(X>x))-h(1)\big)\,\d x.
\end{equation*}
\end{lemma}

We first recall the lower-envelope representation for the inf-convolution of distortion risk measures; see, e.g., \citet{lauzier2025risk} for a more general result.
It identifies the pool's aggregate evaluation as a single lower-envelope distortion functional and provides an explicit layer allocation that assigns each marginal increment to the agent that prices it most favourably.

\begin{theorem}[\citealp{lauzier2025risk}, Theorem~3]\label{thm:lauzier-thm3}
Let $h_1,\dots,h_n\in\mathcal{H}$ and define the lower envelope
\begin{equation}\label{eq:lower-envelope}
h_\wedge(t):=\min\{h_1(t),\dots,h_n(t)\},\qquad t\in[0,1].
\end{equation}
Then $\dboxplus_{i=1}^n \rho_{h_i}=\rho_{h_\wedge}$. Fixing a loss $X\in\X$, a sum-optimal allocation $(X_1,\dots,X_n)\in\A_n^+(X)$ is given by
\begin{equation}\label{eq:lauzier-fi}
X_i=f_i(X),\qquad f_i(x)=\int_{0}^{x} g_i(t)\,\d t,
\qquad i\in[n],\ x\in\R,
\end{equation}
where for $x \le 0$ the integral is understood as $-\int_x^0 g_i(t)\,\d t$, and where $M_x:=\{j\in [n]: h_j(\P(X>x))=h_\wedge(\P(X>x))\}$ and
\begin{equation}\label{eq:lauzier-gi}
g_i(x)=\frac{1}{|M_x|}\,\id_{\{i\in M_x\}},\qquad i\in[n],\ x\in\R.
\end{equation}
Moreover, the sum-optimal allocation in $\A_n^+(X)$ is unique up to constant shifts almost surely if and only if $|M_x|=1$ for $\mu_X$-almost every $x$, where $\mu_X$ is the law of $X$.
\end{theorem}

Figure~\ref{fig:lower-envelope} illustrates the layer allocation in Theorem~\ref{thm:lauzier-thm3} for three agents: the bold lower-envelope curve identifies which agent absorbs each layer of the aggregate loss.

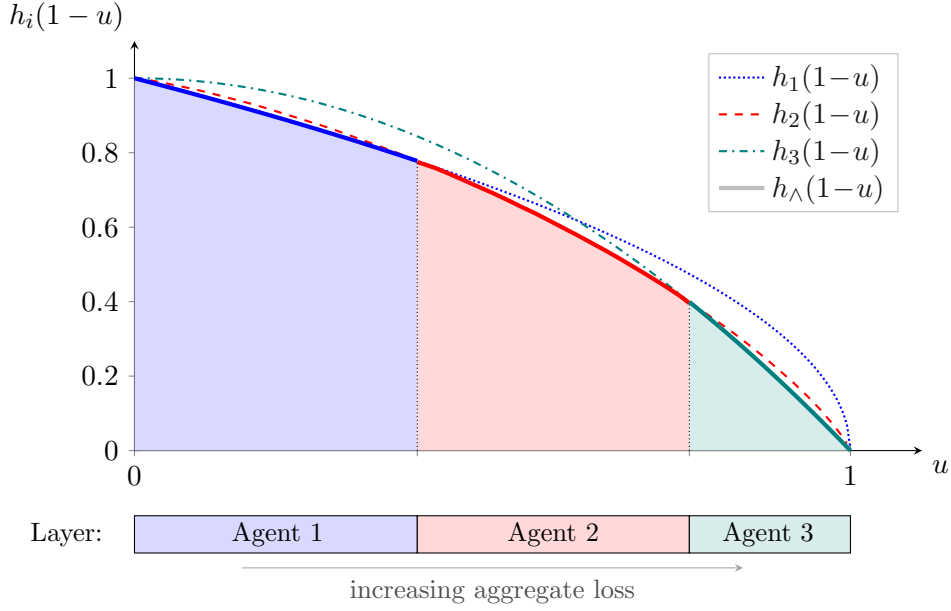
\begin{figure}[ht]
\centering
\begin{tikzpicture}
\begin{axis}[
  name=distplot,
  width=12cm, height=7cm,
  xmin=0, xmax=1.1,
  ymin=0, ymax=1.1,
  xlabel={$u$},
  ylabel={$h_i(1-u)$},
  axis lines=left,
  every axis x label/.style={at={(ticklabel* cs:1.0)}, anchor=north west},
  every axis y label/.style={at={(ticklabel* cs:1.0)}, anchor=south east},
  xtick={0, 1.0},
  extra x ticks={0.395, 0.775},
  extra x tick labels={{}, {}},
  extra x tick style={tick label style={font=\normalsize}},
  ytick={0, 0.2, 0.4, 0.6, 0.8, 1.0},
  tick label style={font=\normalsize},
  label style={font=\normalsize},
  legend style={
    font=\normalsize,
    at={(0.97,0.97)},
    anchor=north east,
    draw=gray!50,
    fill=white,
    fill opacity=0.92,
    cells={anchor=west},
  },
  clip=false,
  samples=200,
  domain=0:1,
]


\addplot[name path=env1, draw=none, domain=0:0.395, samples=60, forget plot]
  {sqrt(1-x)};
\addplot[name path=zero1, draw=none, domain=0:0.395, forget plot] {0};
\addplot[blue!15, forget plot] fill between[of=env1 and zero1];

\addplot[name path=env2, draw=none, forget plot]
  table[x=u, y=g2, col sep=comma] {
    u,     g2
    0.395, 0.775
    0.42,  0.759
    0.44,  0.743
    0.46,  0.726
    0.48,  0.709
    0.50,  0.692
    0.52,  0.674
    0.54,  0.655
    0.56,  0.637
    0.58,  0.617
    0.60,  0.597
    0.62,  0.577
    0.64,  0.556
    0.66,  0.535
    0.68,  0.513
    0.70,  0.490
    0.72,  0.467
    0.74,  0.443
    0.76,  0.418
    0.775, 0.396
  };
\addplot[name path=zero2, draw=none, domain=0.395:0.775, forget plot]
  {0};
\addplot[red!15, forget plot] fill between[of=env2 and zero2];

\addplot[name path=env3, draw=none, domain=0.775:1, samples=60, forget plot]
  {1-x^2};
\addplot[name path=zero3, draw=none, domain=0.775:1, forget plot] {0};
\addplot[teal!15, forget plot] fill between[of=env3 and zero3];


\addplot[thick, blue, densely dotted]
  {sqrt(1-x)};
\addlegendentry{$h_1(1\!-\!u)$}

\addplot[thick, red, dashed]
  table[x=u, y=g2, col sep=comma] {
    u,     g2
    0.000, 1.000
    0.005, 0.999
    0.020, 0.995
    0.040, 0.988
    0.060, 0.980
    0.080, 0.972
    0.100, 0.963
    0.120, 0.953
    0.140, 0.943
    0.160, 0.933
    0.180, 0.922
    0.200, 0.910
    0.220, 0.898
    0.240, 0.886
    0.260, 0.874
    0.280, 0.861
    0.300, 0.847
    0.320, 0.833
    0.340, 0.819
    0.360, 0.805
    0.380, 0.790
    0.400, 0.774
    0.420, 0.759
    0.440, 0.743
    0.460, 0.726
    0.480, 0.709
    0.500, 0.692
    0.520, 0.674
    0.540, 0.655
    0.560, 0.637
    0.580, 0.617
    0.600, 0.597
    0.620, 0.577
    0.640, 0.556
    0.660, 0.535
    0.680, 0.513
    0.700, 0.490
    0.720, 0.467
    0.740, 0.443
    0.760, 0.418
    0.780, 0.393
    0.800, 0.366
    0.820, 0.339
    0.840, 0.311
    0.860, 0.281
    0.880, 0.250
    0.900, 0.217
    0.920, 0.183
    0.940, 0.146
    0.960, 0.106
    0.980, 0.060
    0.990, 0.034
    0.995, 0.019
  };
\addlegendentry{$h_2(1\!-\!u)$}

\addplot[thick, teal, dashdotted]
  {1-x^2};
\addlegendentry{$h_3(1\!-\!u)$}


\addplot[line width=1.6pt, blue, solid, domain=0:0.395,
  samples=60, forget plot]
  {sqrt(1-x)};

\addplot[line width=1.6pt, red, solid, forget plot]
  table[x=u, y=g2, col sep=comma] {
    u,     g2
    0.395, 0.775
    0.42,  0.759
    0.44,  0.743
    0.46,  0.726
    0.48,  0.709
    0.50,  0.692
    0.52,  0.674
    0.54,  0.655
    0.56,  0.637
    0.58,  0.617
    0.60,  0.597
    0.62,  0.577
    0.64,  0.556
    0.66,  0.535
    0.68,  0.513
    0.70,  0.490
    0.72,  0.467
    0.74,  0.443
    0.76,  0.418
    0.775, 0.396
  };

\addplot[line width=1.6pt, teal, solid, domain=0.775:1,
  samples=60, forget plot]
  {1-x^2};

\addlegendimage{line width=1.6pt, solid}
\addlegendentry{$h_\wedge(1\!-\!u)$}

\addplot[thin, densely dotted, black, forget plot]
  coordinates {(0.395, 0) (0.395, 0.778)};
\addplot[thin, densely dotted, black, forget plot]
  coordinates {(0.775, 0) (0.775, 0.396)};

\draw[fill=blue!15, draw=black, thin]
  (axis cs:0,-0.275) rectangle (axis cs:0.395,-0.175);
\draw[fill=red!15, draw=black, thin]
  (axis cs:0.395,-0.275) rectangle (axis cs:0.775,-0.175);
\draw[fill=teal!15, draw=black, thin]
  (axis cs:0.775,-0.275) rectangle (axis cs:1,-0.175);

\node[font=\small, anchor=center]
  at (axis cs:0.200, -0.225) {Agent 1};
\node[font=\small, anchor=center]
  at (axis cs:0.585, -0.225) {Agent 2};
\node[font=\small, anchor=center]
  at (axis cs:0.885, -0.225) {Agent 3};

\node[font=\small, anchor=east, inner sep=1pt]
  at (axis cs:-0.04, -0.225) {Layer:};

\draw[-stealth, thin, gray!70]
  (axis cs:0.15, -0.315) -- (axis cs:0.85, -0.315);
\node[font=\small, gray!60!black, anchor=north]
  at (axis cs:0.50, -0.318) {increasing aggregate loss};

\end{axis}
\end{tikzpicture}
\caption{Concave distortions, plotted as $h_i(1-u)$ for $i \in [3]$ with their lower envelope $h_\wedge(1-u)=\min_{i\in[3]} h_i(1-u)$. The bold curve gives the layer allocation of Theorem~\ref{thm:lauzier-thm3}: Agent~1 absorbs the attritional layers, Agent~2 the middle layers, and Agent~3 the tail layers.}
\label{fig:lower-envelope}
\end{figure}

In the context of distortion risk measures, the canonical pure risk allocation can be identified directly from the distortions. Once $h_\wedge$ is known, every marginal layer of the aggregate loss is assigned to the agent or agents attaining the lower envelope at that tail probability.

The same lower envelope is also the dual object behind the pricing rule in Section~\ref{ss:dual-infconv-pricing}. For a concave distortion $h$, define its dual envelope by
\begin{equation}\label{eq:distortion-dual-envelope}
    \mathcal D_h := \left\{\pi\in L^1_+:\ \E[\pi]=1,\ 
    \E[\pi\id_E]\le h(\P(E))\ \text{for all }E\in\mathcal F\right\}.
\end{equation}

For a nonempty coalition $A\subseteq[n]$ with $\rho_i=\rho_{h_i}$ for $i\in A$, the dual coalition envelope
$$ \bigcap_{i\in A}\mathcal D_{h_i} = \left\{\pi\ge0:\ \E[\pi]=1,\ \E[\pi\id_E]\le h_A(\P(E))\ \text{for all }E\in\mathcal F\right\},
$$
where $h_A(t):=\min_{i\in A}h_i(t)$ for $t\in[0,1]$. Since each $h_i$ is concave for $i \in A$, $h_A$ is also concave, and $\dsquare_{i\in A}\rho_i (Y)=\rho_{h_A}(Y)$ and $C(A)=\rho_{h_A}(X_A)$. A dual optimizer $\pi_A^\star=\d\Q_A^\star/\d\P$ therefore gives the lower-envelope pricing selector $C_i^\star(A)=\E^{\Q_A^\star}[\zeta_i]$ for $i\in A$.

\begin{proposition}\label{prop:distortion-dual-lower-envelope}
Let $A\subseteq[n]$ be a nonempty coalition. For each $i\in A$, let $h_i$ be a concave distortion, and let $h_A:=\min_{i\in A}h_i$. Then $\rho_h(Y)=\sup_{\pi\in\mathcal D_h}\E[\pi Y]$ for every concave distortion $h$ and every $Y\in\X$ for which the two sides are finite, and
$\mathcal D_{h_A}=\bigcap_{i\in A}\mathcal D_{h_i}$. Whenever $C(A)=\rho_{h_A}(X_A)$, the coherent dual representation of the coalition inf-convolution is $C(A)=\sup_{\pi\in\mathcal D_{h_A}}\E[\pi X_A]$.
\end{proposition}

\begin{proof}
The first assertion is the dual representation of coherent distortion risk measures. For the envelope identity, the constraints $\E[\pi\id_E]\le h_i(\P(E))$ for every $i\in A$ are equivalent to the single constraint $\E[\pi\id_E]\le \min_{i\in A}h_i(\P(E))=h_A(\P(E))$ for every $E\in\mathcal F$. The final identity is the same coherent dual representation applied to the lower-envelope coalition cost.
\end{proof}

Thus, the lower-envelope distortion determines both sides of the construction: it gives the primal efficient allocation in Theorem~\ref{thm:lauzier-thm3}, and its dual envelope gives the Arrow--Debreu price used by the allocation in Section~\ref{ss:dual-infconv-pricing}. By Proposition~\ref{prop:infconv-attainment}, coalition costs are finite and attained for every coalition under Assumption \ref{ass:standing}. In the distortion specialization, Theorem~\ref{thm:lauzier-thm3} identifies the value and an explicit lower-envelope optimizer.

\begin{theorem}\label{thm:pricing-kernel}
Let $\rho_i=\rho_{h_i}$, $i\in[n]$, be coherent distortion risk measures, and let $h_A:=\min_{i\in A}h_i$ for every nonempty coalition $A\subseteq[n]$. Let $\Q_A^\star$ be an additive dual optimizer, that is, a probability measure with density $\pi_A^\star=\d\Q_A^\star/\d\P\in\mathcal D_{h_A}$ attaining $\E[\pi_A^\star X_A]=\rho_{h_A}(X_A)$; it satisfies $\Q_A^\star(E)\le h_A(\P(E))$ for every $E\in\mathcal F$, with equality on the upper level events $\{X_A>x\}$, $x\in\R$. Then:
\begin{enumerate}[label=\textup{(\roman*)},ref=\roman*]
\item\label{it:dual-domination} $\E^{\Q_A^\star}[Y]\le\rho_{h_i}(Y)$ for every $i\in A$ and every $Y\in\X$.
\item\label{it:coalition-core-price} The cost shares $C_i^\star(A):=\E^{\Q_A^\star}[\zeta_i]$, $i\in A$, form a cost-core selector:
$$ \sum_{i\in A} C_i^\star(A)=C(A), \qquad \sum_{i\in S} C_i^\star(A)\le C(S), \quad \varnothing\ne S\subseteq A.$$
\item\label{it:zero-excess} For every $\bm X^{A}\in\Opt_A(X_A)$ and every $i\in A$, $\rho_{h_i}(X_i^{A})=\E^{\Q_A^\star}[X_i^{A}]$.
\end{enumerate}
\end{theorem}

\begin{proof}[Proof of Theorem~\ref{thm:pricing-kernel}]
By Proposition~\ref{prop:distortion-dual-lower-envelope}, the lower-envelope distortion case is the coherent specialization of the dual framework in Section~\ref{ss:dual-infconv-pricing}, with coalition envelope $\mathcal D_{h_A}=\cap_{i\in A}\mathcal D_{h_i}$. The domination statement follows from $\pi_A^\star\in\mathcal D_{h_i}$ for every $i\in A$. The cost-core selector statement is Proposition~\ref{prop:dual-pricing-PMAS} applied to the case of coherent risk measures.

For the zero-excess identity, let $\bm X^{A}\in\Opt_A(X_A)$. By domination, $\E^{\Q_A^\star}[X_i^{A}]\le\rho_{h_i}(X_i^{A})$ for each $i\in A$. Summing over $A$ gives
$$C(A) = \E^{\Q_A^\star}[X_A] = \sum_{i\in A}\E^{\Q_A^\star}[X_i^{A}] \le \sum_{i\in A}\rho_{h_i}(X_i^{A}) = C(A).$$
Hence every summand satisfies $\rho_{h_i}(X_i^{A})=\E^{\Q_A^\star}[X_i^{A}]$.
\end{proof}

We note that $\pi_A^\star$ is the dual pricing element corresponding to the optimizer $\Q_A^\star$ used by \citet{chenhuwang2017stable} in the coherent distortion setting; in the absolutely continuous case, $\pi_A^\star=\d\Q_A^\star/\d\P$. The condition $\pi_A^\star\in\cap_{i\in A}\mathcal D_{h_i}$ is exactly simultaneous domination by every agent's risk measure in the coalition. The zero-excess property is the corresponding equality at an efficient allocation, and it converts $\Q_A^\star$-expectation transfers into coalitionally stable transfers. This parallels Arrow--Debreu principles, which support efficient allocations via state-price densities \citep{mas1995microeconomic}, here tailored to coherent distortion risk measures.

\begin{theorem}\label{thm:Qstar-core}
Let $\rho_i=\rho_{h_i}$, $i\in[n]$, be coherent distortion risk measures. Let $\Q^\star$ be the lower-envelope pricing measure for the grand coalition, and let $\nu$ be the gain game from Definition~\ref{def:coalitional-game}. The lower-envelope pricing welfare-gain vector $\bm{w}^{\Q^\star}\in\R^n$ defined by $w_i^{\Q^\star}:=\rho_{h_i}(\zeta_i)-\E^{\Q^\star}[\zeta_i]$, $i\in[n]$, belongs to $\Core(\nu)$.
\end{theorem}

\begin{proof}[Proof of Theorem~\ref{thm:Qstar-core}]
This is Theorem~\ref{thm:dual-core-allocation}\,\textup{(\ref{it:dual-gain-core-allocation})} in the case of coherent risk measures, using the envelope identity from Proposition~\ref{prop:distortion-dual-lower-envelope}.
\end{proof}

Theorem~\ref{thm:Qstar-core} is the distortion-risk-measure analogue of the classical result that Walrasian equilibria belong to the core of an exchange economy \citep[Proposition~18.B.1]{mas1995microeconomic}. \citet[Proposition~4.3]{grechuk2013cooperative} proves an analogous inclusion for subdifferential-based pricing in a mean-deviation setting.

Since core membership implies individual rationality, the $\Q^\star$-welfare-gain vector from Theorem~\ref{thm:Qstar-core} produces an implementable transfer rule for every allocation class in $\Opt(X)_{/\sim}$. The $\Q^\star$-pricing rule distributes the full surplus from pooling: each agent is charged the $\Q^\star$-expected value of their endowment and rebated the difference between their autarky evaluation and that price, so the welfare shares aggregate to the grand-coalition gain, $\sum_{i=1}^n w_i^{\Q^\star}=\sum_{i=1}^n \rho_{h_i}(\zeta_i)-C([n])$.


For explicit density formulas, introduce a randomized probability-integral transform of the coalition aggregate.

\begin{assumption}\label{ass:randomized-PIT}
There exists a random variable $V\sim\mathrm{Unif}(0,1)$, defined on the underlying probability space $(\Omega,\mathcal{F},\P)$ and independent of $\sigma(\zeta_1,\dots,\zeta_n)$.
\end{assumption}

For a nonempty coalition $A\subseteq[n]$, define
\begin{equation}\label{eq:randomized-PIT}
U_A \;:=\; 1 - F_{X_A}(X_A{-}) \;-\; V\bigl(F_{X_A}(X_A) - F_{X_A}(X_A{-})\bigr).
\end{equation}
Then $U_A\sim\mathrm{Unif}(0,1)$ by the randomized probability-integral transform \citep[Proposition~1.3]{ruschendorf2013mathematical}. When $h_A$ is absolutely continuous, the density $(h_A)'_+(U_A)$ is the canonical lower-envelope Arrow--Debreu price whenever it attains the dual supremum at $X_A$.

\begin{example}[Expected Shortfall]\label{ex:ES-lower-envelope}
For $i\in[n]$, let agent $i$ use $\ES_{\beta_i}$ with $\beta_i\in[0,1)$, so its distortion is $h_{\beta_i}(u)=\min\{u/(1-\beta_i),1\}$. For a nonempty coalition $A\subseteq[n]$, let $\beta_A:=\min_{i\in A}\beta_i$. Since Expected Shortfall distortions are pointwise ordered by the confidence level, $h_\wedge^A=h_{\beta_A}$, and therefore $C(A)=\ES_{\beta_A}(X_A)$ and $\mathcal D_A=\mathcal D_{h_{\beta_A}}$. The layer allocation assigns all of $X_A$ to the least conservative agents $M_A:=\{i\in A:\beta_i=\beta_A\}$, split equally among them under the tie rule in Theorem~\ref{thm:lauzier-thm3}. If Assumption~\ref{ass:randomized-PIT} holds, the corresponding price density is
$$ \frac{\d\Q_A^\star}{\d\P} = \frac{1}{1-\beta_A}\id_{\{U_A<1-\beta_A\}},$$
where $U_A$ is the randomized probability-integral transform of $X_A$. If $X_A$ is continuous, this is the usual tail measure on the worst $1-\beta_A$ probability mass of $X_A$, and $C_i^\star(A)=\E^{\Q_A^\star}[\zeta_i]$ for $i\in A$.
\end{example}

\section{Other transfer rules}\label{sec:transfer-rules}

Section~\ref{sec:main-results} defines the risk-sharing game and reduces implementation of efficient risk sharing to selecting a welfare-gain split (equivalently, a deterministic transfer) in the welfare simplex, subject to individual-rationality and core constraints. Each such selection defines a transfer rule, a function that maps agent characteristics to balanced-budget transfers on an optimal allocation. This appendix derives the implementation, individual-rationality, and core-stability properties of other canonical transfer rules in the literature, and records counterexamples showing that none of these rules is a PMAS in general.

\subsection{Proportional-cost transfer rules}\label{ss:proportional-cost-transfers}

Section~\ref{ss:mean-proportional-dual} defines the proportional-cost rule and gives a condition under which it is a PMAS. Here, we record its fixed-coalition implementation, IR criterion, and core criterion. If $\bm X^{A}=(X_i)_{i\in A}$ is an efficient allocation for the restricted coalition problem, the deterministic transfer that implements the target evaluated costs in \eqref{eq:PC-cost-share} is
\begin{equation}\label{eq:PC-transfer}
    c_i^{\bm \xi}(A):=C_i^{\bm \xi}(A)-\rho_i(X_i^{A}),
    \qquad i\in A.
\end{equation}
The normalized object in \eqref{eq:PC-cost-share} is the cost per unit of ${\bm \xi}$, not cost per capita unless $\xi_i\equiv1$ for each $i \in [n]$.

The proportional-cost rule equalizes cost per unit of the chosen weight: every member $i\in A$ of a coalition $A\subseteq[n]$ pays $C_i^{\bm \xi}(A)$ and obtains welfare $w_i^\xi(A)=\rho_i(\zeta_i)-C_i^{\bm \xi}(A)$, so costs split proportionally to $\xi$ while welfare gains favour agents whose autarky cost per unit of weight, $\rho_i(\zeta_i)/\xi_i$, most exceeds the coalition rate $C_i^{\bm \xi}(A)/\xi_i$. Choosing $\xi_i\equiv1$ gives an egalitarian per-capita cost split, $\xi_i=\mu_i$ allocates by expected loss, $\xi_i=\rho_i(\zeta_i)$ implements a pay-as-you-stand split under which welfare gains are proportional to autarky cost, and a contractual exposure measure for $\xi_i$ allocates by that exposure.

\begin{theorem}\label{thm:PC-on-class}
Fix a nonempty coalition $A\subseteq[n]$ and let $\bm X^{A}\in\Opt_A(X_A)$. Define $C_i^{\bm \xi}(A)$ and $c_i^{\bm \xi}(A)$ by \eqref{eq:PC-cost-share} and \eqref{eq:PC-transfer}, for $i \in A$. Then:
\begin{enumerate}[label=\textup{(\roman*)},ref=\roman*]
\item $\sum_{i\in A}c_i^{\bm \xi}(A)=0$, so the transfer is budget-balanced.
\item The transferred allocation $\bm X^{A}+\bm c^{\bm \xi}(A)$ is efficient and satisfies $\rho_i(X_i^{A}+c_i^{\bm \xi}(A))=C_i^{\bm \xi}(A)$ for every $i\in A$.
\item $(C_i^{\bm \xi}(A))_{i \in A}$ is the unique allocation in the deterministic-transfer class $\bm X^{A}+H_A$, where $H_A:=\{\bm d\in\R^A:\sum_{i\in A}d_i=0\}$, whose evaluated costs are proportional to $\bm \xi$.
\item It is individually rational for coalition $A$ if and only if $C_i^{\bm \xi}(A)\le \rho_i(\zeta_i)$ for every $i\in A$.
\end{enumerate}
\end{theorem}

\begin{proof}
Since $\bm X^{A}$ is efficient, $\sum_{i\in A}\rho_i(X_i^{A})=C(A)$. Hence
$$ \sum_{i\in A}c_i^{\bm \xi}(A) = \sum_{i\in A}C_i^{\bm \xi}(A)-\sum_{i\in A}\rho_i(X_i^{A})
    = C(A)-C(A)= 0.$$
Cash additivity gives $\rho_i(X_i^{A}+c_i^{\bm \xi}(A))=\rho_i(X_i^{A})+c_i^{\bm \xi}(A)=C_i^{\bm \xi}(A)$. Deterministic balanced transfers preserve the aggregate loss and the objective value, so the transferred allocation remains efficient.

For uniqueness, suppose $\bm d\in H_A$ and the evaluated costs of $\bm X^{A}+\bm d$ are proportional to some constant $\xi>0$, say $\rho_i(X_i^{A}+d_i)=a\xi$ for all $i\in A$. Summing over $A$ and using efficiency gives $a\xi=C(A)$, so $a=C(A)/\xi$ and $d_i=C_i^{\bm \xi}(A)-\rho_i(X_i^{A})=c_i^{\bm \xi}(A)$. Finally, IR is exactly $\rho_i(X_i^{A}+c_i^{\bm \xi}(A))\le\rho_i(\zeta_i)$ for every $i\in A$, which is the stated condition.
\end{proof}



The proportional-cost IR criterion depends only on $C([n])$ and the weights ${\bm \xi}$, not on the representative efficient random allocation within a deterministic-transfer class. This contrasts with other pricing rules, such as actuarial fairness, where IR can depend on the risk premia of the chosen efficient allocation.

\begin{theorem}\label{thm:PC-core}
Fix a nonempty coalition $B\subseteq[n]$. The proportional-cost welfare-gain vector $\bm{w}^{\bm \xi}(B)$ belongs to $\Core(\nu^B)$ if and only if
\begin{equation}\label{eq:PC-core}
    \frac{C(B)}{\xi_B}
    \le
    \frac{C(A)}{\xi_A},
    \qquad \forall \, \varnothing\ne A\subseteq B.
\end{equation}
Equivalently, $\bar C_\xi(B)\le\bar C_\xi(A)$ for every nonempty $A\subseteq B$.
\end{theorem}

\begin{proof}
Efficiency holds because $\sum_{i\in B}w_i^{\bm \xi}(B)=\sum_{i\in B}\rho_i(\zeta_i)-C(B)=\nu(B)$. For a nonempty $A\subseteq B$, the gain-core inequality is
$$ \sum_{i\in A}\rho_i(\zeta_i)-\xi_A\frac{C(B)}{\xi_B} \ge \sum_{i\in A}\rho_i(\zeta_i)-C(A).$$
Cancelling the stand-alone terms gives $\xi_A C(B)/\xi_B\le C(A)$, which is \eqref{eq:PC-core}. This is the same as the cost-core inequality $\sum_{i\in A}C_i^\xi(B)\le C(A)$. The singleton cases recover the IR conditions because $C(\{i\})=\rho_i(\zeta_i)$.
\end{proof}

\subsection{Expected-value pricing transfers}\label{ss:expected-value-transfers}

The simplest way to specify a transfer rule is to price each agent's position under a probability measure $\Q\ll\P$, where $\P$ is the physical measure from Section~\ref{ss:risk-measures}. Given any such $\Q$, the $\Q$-expected-value transfer rule maps agent characteristics to the balanced-budget transfer $c_i^{\Q}:=\E^{\Q}[\zeta_i]-\E^{\Q}[X_i]$. Note that $\Q^*$ is the Arrow--Debreu pricing measure in \ref{thm:pricing-kernel}, which is one possible $\Q$-pricing transfer rule. The next proposition shows that IR reduces to an agent-wise comparison of $\Q$-risk premia.

\begin{proposition}\label{prop:Q-transfer-criterion}
Fix an efficient allocation $\bm{X}\in\Opt(X)$ and let $\Q\ll\P$ be any probability measure such that $\zeta_i,X_i\in L^1(\Q)$ for each $i \in[n]$. Define
\begin{equation}\label{eq:cQ}
c_i^{\Q}:=\E^{\Q}[\zeta_i]-\E^{\Q}[X_i],\qquad i\in [n].
\end{equation}
Then $\bm{c}^{\Q}\in H$, and $\bm{X}+\bm{c}^{\Q}$ is IR if and only if
\begin{equation}\label{eq:Q-premium-condition}
\rho_i(X_i)-\E^{\Q}[X_i]\le \rho_i(\zeta_i)-\E^{\Q}[\zeta_i],\qquad \forall \, i\in [n].
\end{equation}
\end{proposition}

\begin{proof}[Proof of Proposition~\ref{prop:Q-transfer-criterion}]
Since $\sum_{i=1}^n \zeta_i=\sum_{i=1}^n X_i=X$, we have $\sum_{i=1}^n c_i^{\Q}=\E^{\Q}[X]-\E^{\Q}[X]=0$. By cash additivity,
$$\rho_i(X_i+c_i^{\Q})=\rho_i(X_i)+c_i^{\Q} =\rho_i(X_i)+\E^{\Q}[\zeta_i]-\E^{\Q}[X_i].$$
Thus $\rho_i(X_i+c_i^{\Q})\le \rho_i(\zeta_i)$ is equivalent to \eqref{eq:Q-premium-condition}. \end{proof}

The coalitional analogue reduces core membership to an inequality for each sub-coalition.

\begin{proposition}\label{prop:supporting-price-core-criterion}
Fix an efficient allocation $\bm{X}^\star\in \Opt(X)$ and let $\Q\ll \P$ be such that $\zeta_i,X_i^\star\in L^1(\Q)$ for each $i\in[n]$. Define
\[
c_i^{\Q}:=\E^{\Q}[\zeta_i]-\E^{\Q}[X_i^\star],
\qquad
w_i^{\Q}:=\rho_i(\zeta_i)-\rho_i(X_i^\star+c_i^{\Q}),
\qquad i\in[n],
\]
and, for each coalition $A\subseteq[n]$, write $X_A:=\sum_{i\in A}\zeta_i$ and $X_A^\star:=\sum_{i\in A}X_i^\star$.
Then the welfare-gain vector $\bm{w}^{\Q}$ belongs to $\Core(\nu)$ if and only if
\begin{equation}\label{eq:general-supporting-price-core-criterion}
C(A)-\E^{\Q}[X_A]
\;\ge\;
\sum_{i\in A}\rho_i(X_i^\star)-\E^{\Q}[X_A^\star],
\qquad \forall A\subseteq[n].
\end{equation}
Equivalently,
\begin{equation}\label{eq:general-supporting-price-core-criterion-sum}
\sum_{i\in A}\Big(\rho_i(X_i^\star)-\E^{\Q}[X_i^\star]\Big)
\;\le\;
C(A)-\E^{\Q}[X_A],
\qquad \forall A\subseteq[n].
\end{equation}
\end{proposition}

\begin{proof}[Proof of Proposition~\ref{prop:supporting-price-core-criterion}]
Since $\sum_{i=1}^n c_i^\Q=0$ and $\bm{X}^\star$ is efficient, $\sum_{i=1}^n w_i^\Q=\nu([n])$. Cash additivity gives $w_i^{\Q}=\rho_i(\zeta_i)-\rho_i(X_i^\star)-\E^{\Q}[\zeta_i]+\E^{\Q}[X_i^\star]$ for $i\in[n]$. Summing over $i\in A$ and substituting $\nu(A)=\sum_{i\in A}\rho_i(\zeta_i)-C(A)$, the core condition $\sum_{i\in A}w_i^{\Q}\ge\nu(A)$ reduces to \eqref{eq:general-supporting-price-core-criterion} after cancelling $\sum_{i\in A}\rho_i(\zeta_i)$. Linearity of $\E^{\Q}$ gives \eqref{eq:general-supporting-price-core-criterion-sum}.
\end{proof}

For any $\Q\ll\P$, the $\Q$-risk premium $\Pi_i^{\Q}(Y):=\rho_i(Y)-\E^{\Q}[Y]$ is cash-invariant for each $i \in [n]$. It follows that IR and core membership of a $\Q$-expected-value transfer rule are invariant within a fixed deterministic-transfer class, and Proposition~\ref{prop:Q-transfer-criterion} shows that the transfer vector is unique once $\Q$ and the representative class $\bm{X}+H$ are fixed. Nonuniqueness enters across transfer classes. Within the lower-envelope comonotonic family from Theorem~\ref{thm:lauzier-thm3}, ties in the lower envelope generate multiple transfer-equivalence classes: when $|M_x|>1$ on a set of positive $\mu_X$-measure, the equal-split rule~\eqref{eq:lauzier-gi} is only one admissible assignment of marginal layers to minimizers. Any tie-consistent reassignment of these layers produces a different random allocation in $\A_n^+(X)$ and hence a different element of $\Opt(X)_{/\sim}$, while each such class shares the same aggregate cost $\rho_{h_\wedge}(X)$. The lower-envelope Arrow--Debreu pricing measure $\Q^\star$ is still determined by the aggregate loss and the lower-envelope distortion, not by the tie-breaking allocation. What may differ across tied efficient classes is the induced expected-value transfer, since $c_i^{\Q}=\E^{\Q}[\zeta_i]-\E^{\Q}[X_i]$ depends on the random allocation; for $\Q^\star$, the welfare-gain vector remains class-independent because Theorem~\ref{thm:pricing-kernel}\,\textup{(\ref{it:zero-excess})} gives $\Pi_i^{\Q^\star}(X_i)=0$ for every $\bm{X}\in\Opt(X)$ and every $i\in[n]$.

\subsection{Actuarial fairness: transfers under the physical measure}\label{ss:actuarial-fairness}

Actuarial fairness is the special case when $\Q$ is the physical measure $\P$, using a risk-neutral price. The actuarial-fairness rule makes each agent pay a premium equal to the expected value of their allocated loss, preserving expected losses agent by agent. While canonical under risk-neutral evaluation, this rule can conflict with IR when agents use nonlinear functionals. We provide a criterion for compatibility in terms of each agent's risk premium.

\begin{definition}[Actuarial fairness]\label{def:AF}
An allocation $\bm{X}\in\A_n(X)$ is actuarially fair relative to $\bm{\zeta}$ if
\begin{equation}\label{eq:AF}
\E[X_i]=\E[\zeta_i],\qquad \forall i\in [n].
\end{equation}
\end{definition}

\begin{definition}[Risk premium]\label{def:risk-premium}
For each $i\in [n]$, define the risk premium functional $\Pi_i(Y):=\rho_i(Y)-\E[Y]$, for $Y\in\X$.
Cash additivity implies $\Pi_i(Y+c)=\Pi_i(Y)$ for every $i\in[n]$, $Y\in\X$, and $c\in\R$.
\end{definition}

Once an efficient allocation $\bm{X}$ is fixed, actuarial fairness leads to a unique transfer vector within $\bm{X}+H$. Individual rationality then reduces entirely to comparing risk premia before and after reallocation.

\begin{theorem}\label{thm:AF-on-class}
Fix an efficient allocation $\bm{X}\in\Opt(X)$. Define the actuarially fair transfer $c_i^{\AF}:=\E[\zeta_i]-\E[X_i]$ for $i\in [n],$
and set $\bm{X}^{\AF}:=\bm{X}+\bm{c}^{\AF}$.
Then:
\begin{enumerate}[label=\textup{(\roman*)},ref=\roman*]
\item\label{it:AF-unique} $\bm{c}^{\AF}\in H$ and $\bm{X}^{\AF}$ is the unique efficient, actuarially fair allocation in the deterministic-transfer class $\bm{X}+H$.
\item\label{it:AF-IR} $\bm{X}^{\AF}$ is individually rational if and only if
\begin{equation}\label{eq:AF-IR-condition}
\Pi_i(X_i)\le \Pi_i(\zeta_i),\qquad \forall i\in [n].
\end{equation}
\end{enumerate}
\end{theorem}

\begin{proof}[Proof of Theorem~\ref{thm:AF-on-class}]
Part~(\ref{it:AF-unique}): since $\sum_{i=1}^n\zeta_i=\sum_{i=1}^n X_i=X$, we have $\sum_{i=1}^n c_i^{\AF}=0$, hence $\bm{c}^{\AF}\in H$. By cash additivity and $\sum_{i=1}^n c_i^{\AF}=0$, $\sum_{i=1}^n\rho_i(X_i^{\AF})=\sum_{i=1}^n\rho_i(X_i)=C([n])$, so $\bm{X}^{\AF}\in\Opt(X)$, and $\E[X_i^{\AF}]=\E[\zeta_i]$ by construction. Uniqueness: if $\bm{X}=\bm{X}^\star+\bm{c}$ with $\bm{c}\in H$ satisfies~\eqref{eq:AF}, then $c_i=\E[\zeta_i]-\E[X_i^\star]=c_i^{\AF}$ for each $i\in[n]$. Part~(\ref{it:AF-IR}) is Proposition~\ref{prop:Q-transfer-criterion} with $\Q=\P$, noting that $\rho_i(X_i)-\E[X_i]=\Pi_i(X_i)$ for $i\in[n]$.
\end{proof}

Actuarial fairness is IR if and only if each agent's risk premium weakly decreases from autarky to an efficient allocation position. 

Theorem~\ref{thm:AF-on-class} highlights a contrast with the common approach in the risk-sharing literature, which takes the aggregate risk $X$ as given and characterizes optimal allocations through the set $\Opt(X)_{/\sim}$ and the inf-convolution $\dsquare_{i\in [n]}\rho_i (X)$. However, without modelling the endowments $\zeta_1, \dots, \zeta_n$, one cannot assess whether efficient allocations can be implemented by simple rules such as actuarial fairness. Theorem~\ref{thm:AF-on-class} also characterizes when actuarial fairness satisfies individual rationality, but IR is the weakest stability requirement, requiring only that no single agent prefers autarky. A stronger requirement is coalitional stability, which requires that no subgroup of agents can collectively improve upon the proposed welfare-gain vector by pooling among themselves. While a general statement is currently unavailable, we can characterize exactly when actuarial fairness yields a welfare-gain vector in the core of the gain game for distortion risk measures.

Risk premia are constant on each transfer-equivalence class because they are cash-invariant. However, $\Opt(X)$ can contain multiple transfer classes (i.e., $\Opt(X)_{/\sim}$ is not a singleton) corresponding to different random allocations, and the risk premia can depend on which class is chosen. The next theorem provides a criterion for the existence of allocations that are simultaneously efficient, actuarially fair, and IR.

\begin{theorem}\label{thm:AF-IR-existence}
There exists an efficient allocation that is both actuarially fair and IR if and only if there exists $\bm{X}\in\Opt(X)$ satisfying
\begin{equation}\label{eq:AF-IR-existence}
\Pi_i(X_i)\le \Pi_i(\zeta_i),\qquad \forall i\in [n].
\end{equation}
\end{theorem}

\begin{proof}[Proof of Theorem~\ref{thm:AF-IR-existence}]
Fix $\bm{X}\in\Opt(X)$. By Theorem~\ref{thm:AF-on-class}, the unique actuarially fair point in the class $\bm{X}+H$ is IR if and only if \eqref{eq:AF-IR-condition} holds, which is \eqref{eq:AF-IR-existence}. Existence of an efficient AF--IR allocation is therefore equivalent to the existence of at least one efficient transfer class passing this test.
\end{proof}

Notice that if $\Opt(X)$ consists of a single deterministic-transfer class so that $\Opt(X)_{/\sim}$ is a singleton, then AF--IR existence does not depend on the choice of an efficient allocation. In this case, either \eqref{eq:AF-IR-existence} is true, or AF can never satisfy IR.

\begin{corollary}\label{thm:AF-core-characterization}
Let $\bm{X}\in\Opt(X)$. Then the actuarially fair welfare-gain vector $\bm{w}^{\AF}\in\R^n$ with $w_i^{\AF}=\Pi_i(\zeta_i)-\Pi_i(X_i)$ belongs to $\Core(\nu)$ if and only if
\begin{equation}\label{eq:AF-core-condition}
\dsquare_{i\in A}\rho_i(X_A)-\E[X_A]\ge \sum_{i\in A}\left(\rho_i(X_i)-\E[X_i]\right),\qquad \forall\,\varnothing\ne A\subseteq[n].
\end{equation}
In particular, the singleton case $A=\{i\}$ recovers the IR condition~\eqref{eq:AF-IR-condition}.
\end{corollary}

\begin{proof}[Proof of Corollary~\ref{thm:AF-core-characterization}]
Since $\bm{X}\in\Opt(X)$, we have $\sum_{i=1}^n\rho_i(X_i)=C([n])$. Hence $\sum_{i=1}^n w_i^{\AF}=\sum_{i=1}^n(\Pi_i(\zeta_i)-\Pi_i(X_i))=\nu([n])$, using $\rho_i=\Pi_i+\E$ and $\sum_i\zeta_i=\sum_iX_i=X$. Thus, only the coalitional inequalities need to be checked. For a nonempty $A\subseteq[n]$, the core inequality $\sum_{i\in A}w_i^{\AF}\ge\nu(A)$ is equivalent, after substituting $w_i^{\AF}=\Pi_i(\zeta_i)-\Pi_i(X_i)$ and $\nu(A)=\sum_{i\in A}\rho_i(\zeta_i)-C(A)$, to $C(A)-\E[X_A]\ge\sum_{i\in A}\Pi_i(X_i)$. Since $C(A)=\dsquare_{i\in A}\rho_i(X_A)$ and $\Pi_i(X_i)=\rho_i(X_i)-\E[X_i]$, this is exactly~\eqref{eq:AF-core-condition}. For $A=\{i\}$, $\dsquare_{i\in A}\rho_i(X_A)=\rho_i(\zeta_i)$, so~\eqref{eq:AF-core-condition} becomes $\Pi_i(\zeta_i)\ge\Pi_i(X_i)$.
\end{proof}


We note that Corollary \ref{thm:AF-core-characterization} does not require coalitional stability of AF for all sub-allocations, as condition \eqref{eq:AF-core-condition} is computed directly on the aggregate loss $X_A$ using the lower envelope $h_\wedge^A$ of the coalition $A\subseteq[n]$. The AF core condition is therefore coalition-specific and must be verified on a coalition-by-coalition basis.


To interpret the IR criterion \eqref{eq:AF-IR-condition}, we record how the sign of the risk premium $\Pi_h=\rho_h-\E$ depends on whether the distortion $h \in \mathcal H$ lies above or below the identity.
The following result is well known in the theory of distortion risk measures, but we include a proof for completeness and to clarify its connection to actuarial fairness.

\begin{lemma}\label{lemma:distortion-premium-sign}
Let $\rho_h$ be the distortion risk measure defined in \eqref{eq:choquet-integral}. Define $\Pi_h(Y):=\rho_h(Y)-\E[Y]$ for bounded $Y\in\X$. Then:
\begin{enumerate}[label=\textup{(\roman*)},ref=\roman*]
\item if $h(u)\ge u$ for all $u\in[0,1]$, then $\Pi_h(Y)\ge 0$ for all bounded $Y\in\X$;
\item if $h(u)\le u$ for all $u\in[0,1]$, then $\Pi_h(Y)\le 0$ for all bounded $Y\in\X$.
\end{enumerate}
\end{lemma}

\begin{proof}[Proof of Lemma~\ref{lemma:distortion-premium-sign}]
Since $Y\in\X$ is bounded, there exists $m\in\R$ such that $Y+m\ge 0$ almost surely. By cash additivity, $\Pi_h(Y)=\Pi_h(Y+m)$. Since $Y+m\ge 0$, the tail representations
$$\rho_h(Y+m)=\int_0^\infty h\big(\P(Y+m>t)\big)\,\d t,\qquad \E[Y+m]=\int_0^\infty \P(Y+m>t)\,\d t$$
yield
$$\Pi_h(Y)=\Pi_h(Y+m)=\int_0^\infty \Big(h\big(\P(Y+m>t)\big)-\P(Y+m>t)\Big)\,\d t.$$
The integrand is pointwise $\ge 0$ when $h(u)\ge u$ for all $u\in[0,1]$, giving (i), and pointwise $\le 0$ when $h(u)\le u$ for all $u\in[0,1]$, giving (ii).
\end{proof}

Distortions above the identity overweight adverse quantiles and produce a nonnegative premium; those below the identity produce a nonpositive premium. This explains why actuarial fairness may systematically violate IR for agents with conservative distortions.

\begin{example}[Actuarial fairness: diversification determines IR]\label{ex:AF-diversification}
Let $n=2$ and consider dual-power distortions $h_i(u)=1-(1-u)^{\theta_i}$ for $i \in [2]$, where $\theta_1=2$ and $\theta_2=3$, so that $h_1\le h_2$ pointwise and agent~1 is the less conservative. Both agents hold symmetric Bernoulli endowments $\P(\zeta_i=1)=1/2$ for $i \in [2]$, giving $\rho_1(\zeta_1)=3/4$, $\rho_2(\zeta_2)=7/8$, and $\Pi_i(\zeta_i)=\rho_i(\zeta_i)-1/2>0$ for both agents. Since $h_1\le h_2$, the inf-convolution equals $\rho_{h_1}(X)$ and the efficient allocation is $\bm{X}=(X,0)$ regardless of the dependence structure; only the law of $X=\zeta_1+\zeta_2$ changes.
\begin{itemize}
  \item Take $\zeta_1=\zeta_2=Y$ for a single $Y\sim\mathrm{Bernoulli}(1/2)$, so that $X=2Y$. The actuarially fair transfer is $c_1^{\AF}=\E[\zeta_1]-\E[2Y]=-1/2$ and $c_2^{\AF}=1/2$. Agent~1 absorbs the full aggregate loss: positive homogeneity gives $\Pi_1(2Y)=2\Pi_1(Y)=1/2>1/4=\Pi_1(\zeta_1)$, violating \eqref{eq:AF-IR-condition}. Comonotonicity eliminates diversification, so the risk premium doubles, while actuarial fairness compensates only the additional expected cost. The actuarially fair welfare-gain vector is $(w_1^{\AF},w_2^{\AF})=(-1/4,\,3/8)$, which violates IR.
  \item Take $\zeta_1=Y$ and $\zeta_2=1-Y$, so that $X=1$ almost surely. The efficient split from the inf-convolution is $\bm{X}=(1,0)$. The actuarially fair transfer is $c_1^{\AF}=-1/2$ and $c_2^{\AF}=1/2$, and both post-transfer evaluations satisfy IR: $\rho_1(1)+c_1^{\AF}=1/2<3/4=\rho_1(\zeta_1)$ and $\rho_2(0)+c_2^{\AF}=1/2<7/8=\rho_2(\zeta_2)$. Equivalently, $\Pi_i(X_i)=0$ for both agents since each receives a constant, so \eqref{eq:AF-IR-condition} holds with strict inequality. Counter-monotonicity creates a perfect hedge: the aggregate is riskless, and all risk premia vanish. The actuarially fair welfare-gain vector is $(w_1^{\AF},w_2^{\AF})=(1/4,\,3/8)$.
\end{itemize}
In both cases, agent~2 cedes all risk and receives the same welfare gain, $w_2^{\AF}=3/8$; the dependence structure affects only agent~1, who bears the aggregate cost of insufficient diversification.
\end{example}

In Example~\ref{ex:AF-diversification}, actuarial fairness satisfies IR if and only if the efficient allocation reduces each agent's risk premium below its stand-alone level. When the endowments are comonotonic, agent~1 absorbs a scaled copy of her own loss and gains no diversification. The risk premium grows, while actuarial fairness compensates only for the additional expected cost, resulting in a negative welfare gain. When the endowments are countermonotonic, pooling eliminates all risk, and the risk premia vanish, so both agents gain. The extreme case is a risk-averse reinsurer entering the pool with no initial exposure, so that the reinsurer's risk premium is zero. Since the efficient allocation assigns her a nonconstant share, her risk premium can only increase; actuarial fairness offers no compensation beyond expected value, so condition~\eqref{eq:AF-IR-condition} necessarily fails. For this reason, actuarial fairness is viable only when pooling provides sufficient diversification to spread the welfare gain among all participants.

In practice, insurers and reinsurers often charge an actuarially fair premium plus a loading that covers operating costs and profit margins. Adding a loading makes IR more likely but does not guarantee it: if the risk premium shortfall exceeds the loading, the agent still loses from participation. A more systematic remedy is to replace the physical measure $\P$ with a pricing measure that accounts for risk aversion, as discussed in Section~\ref{ss:dual-infconv-pricing}.

\subsection{The Shapley value and the Weber set}\label{ss:shapley}
The Shapley value \citep{Shapley1953} is a classical solution concept in cooperative game theory that assigns to each player a payoff reflecting their average marginal contribution across all possible orderings of the players. It is the unique division rule satisfying four axioms: efficiency, meaning the payoffs exhaust the total surplus $\nu([n])$; symmetry, meaning two players who make identical marginal contributions to every coalition receive the same payoff; the null-player property, meaning a player whose marginal contribution to every coalition is zero receives nothing; and additivity, meaning the value of a sum of two games equals the sum of the values. 

\begin{definition}[Shapley value]\label{def:shapley}
For a TU game $([n],\nu)$, the Shapley value $\bm{\Sh}(\nu)\in\R^n$ is defined by
\begin{equation}\label{eq:shapley}
\Sh_i(\nu)
:=\sum_{A\subseteq [n]\setminus\{i\}}
\frac{|A|!\,(n-|A|-1)!}{n!}\,\big(\nu(A\cup\{i\})-\nu(A)\big),\qquad i\in [n].
\end{equation}
Equivalently, $\Sh_i(\nu)$ is the expected marginal contribution of $i$ under a uniformly random permutation of $[n]$.
\end{definition}

The Shapley value depends only on the characteristic function $\nu$, not on the choice of efficient allocation $\bm{X}$; it is therefore uniquely determined on the quotient space $\Opt(X)_{/\sim}$ and, a fortiori, invariant within any transfer class. Since it averages marginal contributions over all orderings of the agents, monotonicity of $\nu$ implies that the Shapley payoffs are nonnegative and sum to $W$. The transfer construction from Theorem~\ref{thm:transfer-IR-geometry} then yields an implementable allocation on any efficient representative, with the induced final allocation invariant within each transfer-equivalence class.

\begin{theorem}\label{thm:shapley-IR}
Define the Shapley gain split by $w_i^{\mathrm{Sh}}:=\Sh_i(\nu)$ for $i\in[n]$ and $\bm{w}^{\mathrm{Sh}}:=(w_i^{\mathrm{Sh}})_{i\in[n]}\in\R^n$. Then:
\begin{enumerate}[label=\textup{(\roman*)},ref=\roman*]
\item\label{it:Sh-nonneg} $\bm{w}^{\mathrm{Sh}}\in\Delta_W$: in particular, $w_i^{\mathrm{Sh}}\ge 0$ for all $i\in [n]$ and $\sum_{i\in [n]} w_i^{\mathrm{Sh}}=W$.
\item\label{it:Sh-transfer} The transfer $c^{\mathrm{Sh}}_i:=b_i-w_i^{\mathrm{Sh}},$ for $i\in [n],$ belongs to $H_{\IR}(\bm{X}_0)$, and $\bm{X}^{\mathrm{Sh}}:=\bm{X}_0+\bm{c}^{\mathrm{Sh}}$ is feasible, efficient, and IR. Its welfare baselines equal the Shapley gains, that is, $b_i(\bm{X}^{\mathrm{Sh}})=w_i^{\mathrm{Sh}},$ for $i\in [n].$
\item\label{it:Sh-invariant} $\bm{w}^{\mathrm{Sh}}$ depends only on $\nu$ and hence is the same for every $\bm{X}\in\Opt(X)$.
\end{enumerate}
\end{theorem}

\begin{proof}[Proof of Theorem~\ref{thm:shapley-IR}]
To show part~(\ref{it:Sh-nonneg}), recall that the Shapley value is the average of the marginal vectors induced by all orderings of the agents. Every marginal vector has nonnegative components by monotonicity of $\nu$ (Lemma~\ref{lemma:gain-monotone}) and sums to $\nu([n])=W$, so their average lies in $\Delta_W$.

The argument for part~(\ref{it:Sh-transfer}) is identical to the proof of Proposition~\ref{prop:core-implement}\,(\ref{it:core-transfer}): one has $\sum_{i\in [n]}^n c_i^{\mathrm{Sh}}=\sum_{i\in [n]}^n b_i-\sum_{i\in [n]} w_i^{\mathrm{Sh}}=W-W=0$, and $c_i^{\mathrm{Sh}}\le b_i$ because $w_i^{\mathrm{Sh}}\ge 0$ for every $i\in[n]$. Hence $\bm{c}^{\mathrm{Sh}}\in H_{\IR}(\bm{X}_0)$, and cash additivity yields IR together with the welfare-gain identity.

For part~(\ref{it:Sh-invariant}), observe that the Shapley value~\eqref{eq:shapley} is defined entirely in terms of $\nu$. Neither $\bm{X}_0$ nor the representative within a transfer class enters the definition, so $\bm{w}^{\mathrm{Sh}}$ is the same for every element of $\Opt(X)$.
\end{proof}

The only substantive requirement for the Shapley transfer above is that the gain game be well defined and monotone. Coalitional stability requires more. The standard object linking the Shapley value to the core is the Weber set, which also explains why convexity of the gain game is enough to put the Shapley value in the core.

The Weber set is a relaxation of the core obtained from marginal-contribution vectors along orderings of the agents. Before defining it, we recall the class of convex (supermodular) games, in which the Weber set and the core coincide. Games with increasing marginal value are said to be \emph{supermodular}, or convex.
\begin{definition}[Convex (supermodular) game]\label{def:convex-game}
A TU game $([n],\nu)$ is \emph{convex} if for all coalitions $A,B\subseteq [n]$, $\nu(A)+\nu(B)\;\le\; \nu(A\cup B)+\nu(A\cap B).$
\end{definition}

Let $A\subseteq[n]$ be a given coalition and let $i\in[n]\setminus A$ be an individual agent. Convexity requires that $\nu(A\cup \{i\})\geq \nu(A) +\nu(\{i\}).$ In our risk-sharing game, $\nu(\{i\})=0$, and this condition is always satisfied when adding one agent to the pool. However, a necessary condition for supermodularity is that for every pair of disjoint coalitions $A,B\subseteq[n]$ (that is, $A\cap B=\emptyset$) we have $\nu(A\cup B)\geq \nu(A) +\nu(B),$ so that merging any disjoint coalitions yields greater welfare gains, which is an increasing marginal value interpretation. Further, the convexity of $\nu$ can be defined directly as having increasing marginal contributions: for every $i\in [n]$ and all $A\subseteq B\subseteq [n]\setminus\{i\}$, we have $\nu(A\cup\{i\})-\nu(A)\;\le\; \nu(B\cup\{i\})-\nu(B).$
The latter states that adding an agent to a larger coalition always produces a weakly larger marginal contribution than adding the same agent to a smaller subcoalition. For risk-sharing games, this condition holds when each agent's marginal contribution to the welfare gain is greater in a larger coalition than in a smaller one. We do not pursue general primitive conditions for the convexity of the risk-sharing game, as the property typically fails to hold (especially for large pools), though particular setups may yield it.

Convex TU games have many desirable properties, including being balanced (i.e., having a nonempty core). The converse is not true. The Shapley value always belongs to the core of convex games. For risk-sharing games, the convexity of the risk measures is not necessary to obtain the convexity of the game. This opens up the possibility that some risk-sharing games with non-convex risk measures might admit coalitionally stable allocations. General results on the convexity of risk-sharing games remain scarce. Identifying primitive conditions under which a cost game is convex is an active line of research outside risk sharing. For instance, \citet{ghamami2019submodular} study the allocation of over-the-counter derivative trades across central counterparty portfolios. They give two covariance-matrix conditions under which portfolio standard deviation is submodular on the trade subset lattice: nonpositive off-diagonals and a diagonal-dominance bound.

When the core is nonempty, coalitionally stable allocations exist. The Shapley value provides a unique fair split, but it may lie outside the core when the game is not convex. A natural question is whether there exists a universal outer bound that contains both the core and the Shapley value. The Weber set, introduced by \citet{Weber1988}, is the object in question.

\begin{definition}[Marginal vectors and Weber set]\label{def:weber}
Let $\mathfrak{P}([n])$ be the set of all permutations of $[n]$. For $\mathfrak{p}\in\mathfrak{P}([n])$ and $i\in [n]$, define the predecessor set $P_\mathfrak{p}(i):=\{j\in [n]:\mathfrak{p}(j)<\mathfrak{p}(i)\}$ and the associated \emph{marginal vector} $\bm{m}_\mathfrak{p}(\nu)\in\R^n$ by
$$  m_{\mathfrak{p},i}(\nu):= \nu\bigl(P_\mathfrak{p}(i)\cup\{i\}\bigr)-\nu\bigl(P_\mathfrak{p}(i)\bigr),\qquad i\in [n].$$
Note that $\sum_{i\in [n]} m_{\mathfrak{p},i}(\nu)=\nu([n])$ for every $\mathfrak{p}$. The \emph{Weber set} of $([n],\nu)$ is the convex hull of all marginal vectors:
$$  \Weber(\nu):= \operatorname{conv}\bigl\{\, \bm{m}_\mathfrak{p}(\nu) : \mathfrak{p}\in\mathfrak{P}([n])\,\bigr\}.$$
\end{definition}

Equivalently, $\Weber(\nu)$ is the set of random-order values: $\bm{w}\in \Weber(\nu)$ if and only if $\bm{w}=\sum_{\mathfrak{p}\in \mathfrak{P}([n])} q(\mathfrak{p})\,\bm{m}_\mathfrak{p}(\nu)$ for some probability distribution $q$ on $\mathfrak{P}([n])$. The Shapley value is the barycentre of the Weber set, that is, the random-order value under the uniform distribution, $\bm{\Sh}(\nu)=\frac{1}{n!}\sum_\mathfrak{p} \bm{m}_{\mathfrak{p} \in \mathfrak{P}([n])}(\nu)$, so $\bm{\Sh}(\nu) \in \Weber(\nu)$ always holds.

It is well known that the core, when nonempty, is contained in the Weber set. Welfare-gain vectors in the Weber set are sometimes called \emph{admissible}; core stability is thus a refinement of admissibility. Finally, when the game is convex (supermodular), every marginal vector is itself a core element, so the entire Weber set collapses to the core. The next proposition collects these facts. The inclusion of the core in the Weber set is due to \citet{Weber1988}, and the convex-game equivalence was established by \citet{shapley1971cores} and \citet{Ichiishi1981}.

\begin{proposition}\label{prop:core-weber}
For every TU game $([n],\nu)$, $\Core(\nu)\subseteq \Weber(\nu).$
Moreover, $\nu$ is convex
if and only if $\Core(\nu)=\Weber(\nu)$. In particular, $\bm{\Sh}(\nu)\in \Weber(\nu)$ always.
\end{proposition}

In this paper, payoffs $\bm{w}$ are post-transfer welfare shares relative to autarky: for an efficient allocation $\bm{X}$ and a balanced transfer $\bm{c}\in H$, agent $i$'s welfare share is $w_i=b_i(\bm{X})-c_i=\rho_i(\zeta_i)-\rho_i(X_i+c_i)$ for $i\in[n]$. On any fixed efficient transfer class, the set of individually rational Pareto-optimal allocations coincides with the welfare simplex $\Delta_{\nu([n])}=\{\bm{w}\in\R^n_+:\sum_{i\in [n]} w_i=\nu([n])\}$ from Theorem~\ref{thm:transfer-IR-geometry}. Since every marginal vector sums to $\nu([n])$ and has nonnegative components by monotonicity of $\nu$ (Lemma~\ref{lemma:gain-monotone}), it follows that $\Weber(\nu) \subseteq \Delta_{\nu([n])}$. The admissible set in welfare-share coordinates is therefore $\Weber(\nu) \cap \Delta_{\nu([n])} = \Weber(\nu)$, and the corresponding evaluation vectors $\bm{r}=\bigl(\rho_i(\zeta_i)\bigr)_{i\in [n]}-\bm{w}$ satisfy $r_i\le \rho_i(\zeta_i)$ for $i\in[n]$. The preceding results give the following chain of cooperative requirements on a welfare-gain vector $\bm{w}$ with $\sum_{i\in[n]} w_i = \nu([n])$:
\begin{equation}\label{eq:chain-of-subsets}
\text{efficiency}\;\supseteq\;\underbrace{\Delta_{\nu([n])}}_{\text{IR}}\;\supseteq\;\underbrace{\Weber(\nu)}_{\text{admissibility}}\;\supseteq\;\underbrace{\Core(\nu)}_{\text{core stability}}.
\end{equation}


\subsection{The nucleolus}\label{ss:nucleolus}

The Shapley value selects a welfare-gain vector by averaging the marginal contributions across all orderings of the set of agents; it is the unique rule that satisfies efficiency, symmetry, and additivity. A complementary approach is to select the core element that is most robust to coalitional deviation. The nucleolus, introduced by \citet{Schmeidler1969}, is the unique imputation that lexicographically minimizes the maximum excess across all coalitions. It therefore identifies the welfare-gain vector under which the most dissatisfied coalition has the least incentive to leave the grand coalition.

\begin{definition}[Excess]\label{def:excess}
For a TU game $([n],\nu)$ with $\nu(\varnothing)=0$ and a payoff vector $\bm{w}\in\R^n$ with $\sum_{i=1}^n w_i=\nu([n])$, the \emph{excess} of coalition $A\subseteq [n]$ at $\bm{w}$ is
\begin{equation}\label{eq:excess}
e(A,\bm{w}):=\nu(A)-\sum_{i\in A} w_i.
\end{equation}
\end{definition}

\begin{definition}[Nucleolus]\label{def:nucleolus}
Let $([n],\nu)$ be a TU game with $\nu(\varnothing)=0$ and let $\mathcal{I}(\nu):=\{\bm{w}\in\R^n:\sum_{i\in[n]} w_i=\nu([n]),\; w_i\ge \nu(\{i\})\;\forall i\in[n]\}$ denote the set of \emph{imputations}, that is, efficient payoff vectors that give each player at least their stand-alone value. For each $\bm{w}\in\mathcal{I}(\nu)$, define the \emph{excess vector} $\bm{\theta}(\bm{w})\in\R^{2^n-2}$ as the vector of excesses $e(A,\bm{w})$ for all proper nonempty coalitions $A\subset [n]$, arranged in non-increasing order. The \emph{nucleolus} $\Nuc(\nu)$ is the unique imputation that lexicographically minimises $\bm{\theta}(\bm{w})$ over $\mathcal{I}(\nu)$:
\begin{equation}\label{eq:nucleolus-def}
\Nuc(\nu):=\argmin_{\bm{w}\in\mathcal{I}(\nu)}\; \bm{\theta}(\bm{w}).
\end{equation}
\end{definition}

The nucleolus has several fundamental properties that make it a natural transfer rule. It always exists and is unique \citep{Schmeidler1969}; if the core is nonempty, the nucleolus lies in the core. In particular, since the gain game is totally balanced under convex risk measures (Theorem~\ref{thm:convex-implies-core}), the nucleolus is well defined and lies in the core.

Like the Shapley value, the nucleolus depends only on $\nu$ and is therefore uniquely determined on the quotient space $\Opt(X)_{/\sim}$. Since the nucleolus is an imputation, it lies in $\Delta_W$, and the transfer construction from Theorem~\ref{thm:transfer-IR-geometry} yields an implementable allocation on any efficient representative.

\begin{theorem}\label{thm:nuc-IR}
Define the nucleolus gain split by $w_i^{\Nuc}:=\Nuc_i(\nu)$ for $i\in[n]$ and $\bm{w}^{\Nuc}:=(w_i^{\Nuc})_{i\in[n]}\in\R^n$. Then:
\begin{enumerate}[label=\textup{(\roman*)},ref=\roman*]
\item\label{it:Nuc-nonneg} $\bm{w}^{\Nuc}\in\Delta_W$: in particular, $w_i^{\Nuc}\ge 0$ for all $i\in [n]$ and $\sum_{i\in [n]} w_i^{\Nuc}=W$.
\item\label{it:Nuc-transfer} The transfer $c^{\Nuc}_i:=b_i-w_i^{\Nuc}$ for each $i\in [n]$ belongs to $H_{\IR}(\bm{X}_0)$, and $\bm{X}^{\Nuc}:=\bm{X}_0+\bm{c}^{\Nuc}$ is feasible, efficient, and IR. Its welfare baselines equal the nucleolus gains, $b_i(\bm{X}^{\Nuc})=w_i^{\Nuc}$ for each $i\in [n]$.
\item\label{it:Nuc-invariant} $\bm{w}^{\Nuc}$ depends only on $\nu$ and hence is the same for every $\bm{X}\in\Opt(X)$.
\end{enumerate}
\end{theorem}

\begin{proof}[Proof of Theorem~\ref{thm:nuc-IR}]
To show part~(\ref{it:Nuc-nonneg}), recall that the nucleolus is an imputation \citep{Schmeidler1969}, so $w_i^{\Nuc}\ge\nu(\{i\})=0$ for every $i\in [n]$ and $\sum_{i\in [n]} w_i^{\Nuc}=\nu([n])=W$, hence $\bm{w}^{\Nuc}\in\Delta_W$. The argument for part~(\ref{it:Nuc-transfer}) is identical to that of Theorem~\ref{thm:shapley-IR}\,(\ref{it:Sh-transfer}). Finally, for part~(\ref{it:Nuc-invariant}), the nucleolus~\eqref{eq:nucleolus-def} is defined entirely in terms of $\nu$, so $\bm{w}^{\Nuc}$ is the same for every element of $\Opt(X)$.
\end{proof}

Together with the lower-envelope $\Q^\star$ rule from Section~\ref{ss:dual-infconv-pricing}, these transfer rules divide the welfare gain according to different principles. The $\Q^\star$-pricing rule compensates each agent at the risk-adjusted cost of their endowment: the agent who absorbs the aggregate loss bears more risk but is also the one willing to accept it at the lowest cost, so $\Q^\star$ assigns this agent a small share of the surplus. The Shapley value averages marginal contributions across all coalition orderings, yielding a more balanced split. The nucleolus pushes welfare toward the agent whose absence would cost the grand coalition the most, because equalizing coalitional excesses concentrates surplus on agents that no subgroup can easily replace. The examples of Appendix~\ref{sec:examples} confirm this ordering across several dependence structures.

\subsection{Canonical transfer rules are not PMAS}\label{ss:canonical-not-PMAS}

The proportional-cost rule (for any fixed positive weights), the actuarially fair rule, the coalition-wise lower-envelope pricing rule, the coalition-wise Shapley rule, and the coalition-wise nucleolus are not entry-monotone in general. Applying a rule separately to every restricted game may destroy consistency across coalitions. None of the canonical rules considered here is, in general, a PMAS. The separation examples below exhibit instances where proportional-cost sharing satisfies PMAS while lower-envelope pricing fails, and instances where the reverse holds; neither rule dominates the other.

\begin{theorem}\label{thm:canonical-not-PMAS}
The proportional-cost rule, for fixed positive weights, the actuarially fair rule, the coalition-wise lower-envelope pricing rule, the coalition-wise Shapley rule, and the coalition-wise nucleolus are not PMAS in general.
\end{theorem}

\begin{proof}
For the proportional-cost rule, fix positive weights $\xi_1,\xi_2,\xi_3$ and choose $M>\xi_3/(\xi_1+\xi_2)$. Let $Y\sim\mathrm{Bernoulli}(1/2)$, set $\zeta_1=Y$, $\zeta_2=1-Y$, and $\zeta_3=MY$, and let all agents use $h(u)=\min\{2u,1\}$, so that $C(\{1,2\})=\rho_h(1)=1$ and $C(\{1,2,3\})=\rho_h(1+MY)=1+M$. Hence incumbent~1's proportional costs are $C_1^\xi(\{1,2\})=\xi_1/(\xi_1+\xi_2)$ and $C_1^\xi(\{1,2,3\})=\xi_1(1+M)/(\xi_1+\xi_2+\xi_3)$. By the choice of $M$, the latter is larger, so $w_1^\xi(\{1,2,3\})<w_1^\xi(\{1,2\})$ and entry monotonicity fails.

For the actuarially fair rule, take $n=2$, let $Y\sim\mathrm{Bernoulli}(1/2)$, set $\zeta_1=\zeta_2=Y$, and use dual-power distortions $h_1(u)=1-(1-u)^2$ and $h_2(u)=1-(1-u)^3$.
Since $h_1\le h_2$, the efficient allocation assigns the aggregate to agent~1: $X_1=2Y$ and $X_2=0$. With $\rho_{h_1}(Y)=h_1(1/2)=3/4$, positive homogeneity gives $\Pi_1(\zeta_1)=\rho_{h_1}(Y)-\E[Y]=1/4$ and $\Pi_1(X_1)=2\Pi_1(Y)=1/2$.
Thus $w_1^{\AF}(\{1,2\})=-1/4<0=w_1^{\AF}(\{1\})$, so entry monotonicity fails.

For coalition-wise lower-envelope pricing, it suffices to violate entry monotonicity, since the rule is core-selecting at every coalition. Let $Z_1,Z_2,Z_3$ be independent $\mathrm{Bernoulli}(1/2)$ variables, set $\zeta_1=Z_1$, $\zeta_2=2Z_2$, and $\zeta_3=2Z_3$, and let all agents use $h(u)=\min\{2u,1\}$. For $A=\{1,2\}$, the aggregate $X_A=Z_1+2Z_2$ takes the four values $0,1,2,3$ with equal probability, and the lower-envelope pricing measure prices the worst half of the states. Hence $\E^{\Q_A^\star}[\zeta_1]=\E[Z_1\mid Z_2=1]=1/2$. For $B=\{1,2,3\}$, the worst half of $X_B=Z_1+2Z_2+2Z_3$ is $\{X_B\ge3\}$; among those four states, three have $Z_1=1$. Hence $\E^{\Q_B^\star}[\zeta_1]=3/4$.
Since $\rho_h(\zeta_1)=1$ is coalition-independent, incumbent~1's welfare falls from $1/2$ to $1/4$ when player~3 enters coalition $\{1,2\}$.

For the coalition-wise Shapley rule, we use a different setup. Let $U,V$ be independent $\mathrm{Bernoulli}(1/2)$ variables, set $\zeta_1=\zeta_2=\id_{\{U=V=1\}}$ and $\zeta_3=\id_{\{U\ne V\}}$, and again let all agents use $h(u)=\min\{2u,1\}$. The induced gain game satisfies $\nu(\{i\})=\nu(\{1,2\})=0$, $\nu(\{1,3\})=\nu(\{2,3\})=1/2$, and $\nu([3])=1/2$.
For the subgame on $A=\{1,3\}$, the Shapley value splits $\nu(A)=1/2$ equally, so $\Sh_1(\nu^A)=1/4$. In the grand coalition,
\(\Sh_1(\nu)=\frac13\nu(\{1\})+\frac16(\nu(\{1,2\})-\nu(\{2\}))+\frac16(\nu(\{1,3\})-\nu(\{3\}))+\frac13(\nu([3])-\nu(\{2,3\}))=1/12\).
Thus entry of player~2 harms incumbent~1.

For the coalition-wise nucleolus, the same gain game gives a counterexample. The two-player subgame on $A=\{1,3\}$ has nucleolus $(1/4,1/4)$, so $\Nuc_1(\nu^A)=1/4$. In the grand coalition, the pair constraints $w_1+w_3\ge1/2$ and $w_2+w_3\ge1/2$, together with efficiency and nonnegativity, force $\Core(\nu)=\{(0,0,1/2)\}$.
Hence $\Nuc(\nu)=(0,0,1/2)$ and incumbent~1's nucleolus payoff falls from $1/4$ to $0$ after player~2 enters. This failure is not only a nonconvex-game pathology; \citet{sonmez1993population} gives analogous failures for convex games.
\end{proof}

The PMAS criteria for the mean-proportional rule and for lower-envelope pricing are logically independent. Return to the lower-envelope-pricing counterexample in the proof above: let $Z_1,Z_2,Z_3$ be independent $\mathrm{Bernoulli}(1/2)$ variables, set $\zeta_1=Z_1$, $\zeta_2=2Z_2$, and $\zeta_3=2Z_3$, and let all agents use $h(u)=\min\{2u,1\}$. With mean weights $\mu=(\tfrac12,1,1)$, the normalized costs $\bar C(A):=C(A)/\mu_A$ are antitone under inclusion on $[3]$, so Theorem~\ref{thm:PC-PMAS} makes the mean-proportional rule a PMAS; lower-envelope pricing, by contrast, charges incumbent~1 $C_1^\star(\{1,2\})=\tfrac12$ and $C_1^\star([3])=\tfrac34$, so its welfare falls when agent~3 enters $\{1,2\}$.

Conversely, the fully comonotonic Bernoulli scenario in Table~\ref{tab:bernoulli-dependence-rules}, with $\zeta_i=\id_{\{U\le p_i\}}$ for $i\in[3]$, $(p_1,p_2,p_3)=(0.3,0.4,0.5)$, and $h_1\le h_2\le h_3$, gives the reverse separation: comonotonicity makes lower-envelope pricing a PMAS by Corollary~\ref{cor:comonotonic-PMAS}, whereas the mean-proportional rule with $\mu_i=p_i$ for every $i\in[3]$ violates the normalized-cost antitonicity of Theorem~\ref{thm:PC-PMAS} already along $\{1,3\}\subseteq[3]$. Neither proportional-cost sharing nor lower-envelope Arrow--Debreu pricing dominates the other as a route to PMAS.

\section{Additional PMAS examples}\label{app:additional-pmas-examples}

This appendix records additional applications of the PMAS criteria from Section~\ref{sec:stable-expansion-dual}. 

\begin{example}[Elliptical dual-pricing formula]\label{ex:elliptical-dual-pricing}
Suppose $\bm\zeta=(\zeta_1,\dots,\zeta_n)$ is elliptically distributed with finite second moments, mean vector $\bm\mu$, and positive-definite covariance matrix $\Sigma$. For every nonempty $A\subseteq[n]$, assume $\Var(X_A)>0$ and write $\beta_{i,A}:=\Cov(\zeta_i,X_A)/\Var(X_A)$ for $i\in A$. Elliptical regression from \cite{cambanis1981theory} gives $\E[\zeta_i\mid X_A]=\mu_i+\beta_{i,A}(X_A-\mu_A)$. If a dual optimizer $\pi_A^\star$ is $\sigma(X_A)$-measurable, then the convex-penalty dual-pricing share is
$$ C_i^\star(A) = \mu_i + \beta_{i,A}\left(\E[\pi_A^\star X_A]-\mu_A\right) - \alpha_i(\pi_A^\star), \qquad i\in A.$$
For coherent risk measures, this reduces to $C_i^\star(A)=\mu_i+\beta_{i,A}(C(A)-\mu_A)$ for $i \in A$. For distortion risk measures, $C_i^\star(A)=\mu_i+\sigma_i r_{i,A}\rho_{h_A}(Z)$, where $\sigma_i:=\sqrt{\Var(\zeta_i)}$ and $r_{i,A}:=\Cor(\zeta_i,X_A)$.
\end{example}

The Gaussian case is the specialization of Corollary~\ref{cor:elliptical-PMAS} in which the elliptical regression coefficients are correlations. The dual condition is a comparison between the standardized aggregate deviations and the individual penalties.

\begin{corollary}\label{cor:gaussian-PMAS}
Suppose that Assumption~\ref{ass:dual-pricing} holds, and that $(\zeta_1,\dots,\zeta_n)$ is jointly Gaussian with mean vector $\bm{\mu}$ and positive-definite covariance matrix $\Sigma$. For every nonempty $A\subseteq[n]$, define $\mu_A:=\E[X_A]$, $\sigma_A:=\sqrt{\Var(X_A)}$, and $d_A^\star:=(\E[\pi_A^\star X_A]-\mu_A)/\sigma_A$; for every $i\in[n]$, define $\sigma_i:=\sqrt{\Var(\zeta_i)}$, and for every $i\in A$, define $r_{i,A}:=\Cor(\zeta_i,X_A)$. If, for every $\varnothing\ne A\subseteq B\subseteq[n]$ and every $i\in A$,
$$ \sigma_i r_{i,B}d_B^\star-\alpha_i(\pi_B^\star) \le \sigma_i r_{i,A}d_A^\star-\alpha_i(\pi_A^\star),$$
then the dual-pricing allocation is a PMAS.
\end{corollary}

\begin{proof}
For Gaussian endowments, the elliptical formula in Corollary~\ref{cor:elliptical-PMAS} gives
$$ C_i^\star(A) = \mu_i+\sigma_i r_{i,A}d_A^\star-\alpha_i(\pi_A^\star), \qquad i\in A.$$
The displayed condition is exactly the cross-coalition antitonicity $C_i^\star(B)\le C_i^\star(A)$ for $\varnothing\ne A\subseteq B\subseteq[n]$ and $i\in A$. Proposition~\ref{prop:dual-pricing-PMAS} gives the PMAS conclusion.
\end{proof}

Mean--variance preferences provide a useful setup because the coalition inf-convolution is explicit, and the PMAS condition can be read directly from the normalized variance. In this case, the risk measure is not monotone, implying that the dual price need not be a probability measure, unless the supporting kernel is non-negative. 

\begin{example}[Mean--variance PMAS conditions]\label{ex:mean-variance-PMAS}
For each $i\in[n]$ and $Y\in L^2$, let $\rho_i(Y)=\E[Y]+\gamma_i\Var(Y)$ with $\gamma_i>0$, and write $\vartheta_i:=1/\gamma_i$. For every nonempty coalition $A\subseteq[n]$, set $\vartheta_A:=\sum_{i\in A}\vartheta_i$, $\mu_A:=\E[X_A]$, and $\bar X_A:=X_A/\mu_A$ when $\mu_A>0$. Then $\dsquare_{i\in A}\rho_i (Y)=\E[Y]+\vartheta_A^{-1}\Var(Y)$ and $C(A)=\mu_A+\vartheta_A^{-1}\Var(X_A)$.
Assume $\mu_i>0$ for every $i\in[n]$. The mean-proportional allocation \eqref{eq:proportional-selector-main} defines a PMAS if and only if
\begin{equation}\label{eq:mean-variance-mean-pmas}
    \frac{\Var(X_B)}{\vartheta_B\mu_B}
    \le
    \frac{\Var(X_A)}{\vartheta_A\mu_A},
    \qquad \varnothing\ne A\subseteq B\subseteq[n].
\end{equation}
Equivalently, $(\mu_B/\vartheta_B)\Var(\bar X_B)\le(\mu_A/\vartheta_A)\Var(\bar X_A)$ for every $\varnothing\ne A\subseteq B\subseteq[n]$. A primitive sufficient condition for \eqref{eq:mean-variance-mean-pmas} is $\bar X_B\cx \bar X_A$ and $\vartheta_B/\mu_B\ge \vartheta_A/\mu_A$ for every $\varnothing\ne A\subseteq B\subseteq[n]$, since the first comparison gives $\Var(\bar X_B)\le\Var(\bar X_A)$ and the second comparison gives $\mu_B/\vartheta_B\le\mu_A/\vartheta_A$.


For the dual-pricing allocation of Example~\ref{ex:mean-variance-dual}, the exact PMAS condition is the antitonicity of the displayed cost shares: for every $\varnothing\ne A\subseteq B\subseteq[n]$ and every $i\in A$,
$$ 2\gamma_B\Cov(\zeta_i,X_B) -\frac{\gamma_B^2}{\gamma_i}\Var(X_B) \le 2\gamma_A\Cov(\zeta_i,X_A)    -\frac{\gamma_A^2}{\gamma_i}\Var(X_A).$$
As noted in Example~\ref{ex:mean-variance-dual}, the optimizer $\pi_A^\star=1+2\gamma_A(X_A-\E[X_A])$ is a probability density only under the additional nonnegativity condition $\pi_A^\star\ge0$ almost surely; otherwise it should be interpreted as a signed pricing kernel.
\end{example}

We now turn to entropic risk measures. This class is useful here for two reasons: the inf-convolution adds risk tolerances exactly, and the optimal sharing rule is linear in aggregate losses, which makes the connection with \citet{feng2024expansion} transparent. The actuarially fair entropic selector is most naturally written by separating the physical mean from the centred risk premium. This is the exponential analogue of the mean-proportional logic above, but the scalar that must be antitone is now a normalized cumulant-generating function rather than a normalized variance.

\begin{corollary}\label{cor:entropic-centered-PMAS}
Assume risk measures are entropic as in Example~\ref{ex:entropic-dual}. For every nonempty coalition $A\subseteq[n]$, define $\mu_A:=\E[X_A]$ and $\psi_A:=\log\E[\e^{(X_A-\mu_A)/\iota_A}]$. For $i\in A$, define $C_i^{\AF}(A):=\mu_i+\iota_i\psi_A$. If $\psi_B\le\psi_A$ whenever $\varnothing\ne A\subseteq B\subseteq[n]$, then the welfare gains $w_i^{\AF}(A)=\rho_i(\zeta_i)-C_i^{\AF}(A)$ define a PMAS. In particular, the condition holds if
\begin{equation}\label{eq:entropic-centered-cx}
    \frac{X_B-\mu_B}{\iota_B} \cx \frac{X_A-\mu_A}{\iota_A},\qquad \varnothing\ne A\subseteq B\subseteq[n].
\end{equation}
\end{corollary}

\begin{proof}
Efficiency follows from $\sum_{i\in A}C_i^{\AF}(A)=\mu_A+\iota_A\psi_A=\iota_A\log\E[\e^{X_A/\iota_A}]=C(A)$. If $\psi_B\le\psi_A$ and $i\in A\subseteq B$, then $C_i^{\AF}(B)\le C_i^{\AF}(A)$, so the rule is entry-monotone. Proposition~\ref{prop:entry-not-core} establishes core membership for every restricted coalition; hence, the rule is a PMAS. Finally, \eqref{eq:entropic-centered-cx} implies $\psi_B\le\psi_A$ because $x\mapsto\e^x$ is convex.
\end{proof}

The preceding condition has the expected homogeneous-pool specialization. Exchangeability yields the convex-order contraction, while common risk tolerance ensures compatibility of the normalization across coalitions.

\begin{corollary}\label{cor:entropic-exchangeable-PMAS}
Assume the entropic setup of Example~\ref{ex:entropic-dual}. Suppose $(\zeta_1,\dots,\zeta_n)$ is exchangeable, $\mu_i=\mu$ for every $i\in[n]$, and $\iota_i=\iota$ for every $i\in[n]$, where $\mu\in\R$ and $\iota>0$. Then the actuarially fair allocation in Corollary~\ref{cor:entropic-centered-PMAS} defines a PMAS. Moreover, the entropic dual-pricing allocation coincides with this selector.
\end{corollary}

\begin{proof}
For $\varnothing\ne A\subseteq B\subseteq[n]$, exchangeability gives the sample-average convex-order contraction $X_B/|B|\cx X_A/|A|$, hence $(X_B-|B|\mu)/(|B|\iota)\cx (X_A-|A|\mu)/(|A|\iota)$. Corollary~\ref{cor:entropic-centered-PMAS} gives the PMAS property. For the pricing identity, $\pi_A^\star$ is $\sigma(X_A)$-measurable and exchangeability gives $\E[\zeta_i\mid X_A]=X_A/|A|$ for every nonempty coalition $A\subseteq[n]$ and every $i\in A$. Therefore $C_i^\star(A)=|A|^{-1}\E^{\Q_A^\star}[X_A]-\iota\E^{\Q_A^\star}[\log\pi_A^\star]=\iota\log\E[\e^{X_A/(|A|\iota)}]=\mu+\iota\psi_A=C_i^{\AF}(A)$.
\end{proof}

Comonotonicity has a different interpretation in the entropic model. It does not by itself produce entry monotonicity; one also needs the risk-tolerance scale to grow proportionally with mean exposure.

\begin{corollary}\label{cor:entropic-comonotonic-scale-PMAS}
Assume entropic risk measures as in Example~\ref{ex:entropic-dual}. Suppose $\zeta_i=\mu_iZ$ for every $i\in[n]$, where $\mu_i>0$, $\E[Z]=1$, and $\E[\e^{Z/\kappa}]<\infty$ for some $\kappa>0$. If $\iota_i=\kappa\mu_i$ for every $i\in[n]$, then the actuarially fair transfer in Corollary~\ref{cor:entropic-centered-PMAS} and the entropic dual-pricing allocation coincide and define a PMAS.
\end{corollary}

\begin{proof}
For every nonempty coalition $A\subseteq[n]$, one has $X_A=\mu_AZ$ and $\iota_A=\kappa\mu_A$, so $(X_A-\mu_A)/\iota_A=(Z-1)/\kappa$ is independent of $A$. Hence $\psi_A$ is constant over all nonempty coalitions, and Corollary~\ref{cor:entropic-centered-PMAS} gives the PMAS property. The entropic dual optimizer satisfies $\pi_A^\star=\e^{Z/\kappa}/\E[\e^{Z/\kappa}]$, which is also independent of $A$. Substituting $\zeta_i=\mu_iZ$ and $\iota_i=\kappa\mu_i$ in the dual-pricing formula in Example~\ref{ex:entropic-dual} gives $C_i^\star(A)=\mu_i+\iota_i\psi_A=C_i^{\AF}(A)$.
\end{proof}

The jointly normal specialization recovers a condition identified in \citet{feng2024expansion}. It also shows how their equilibrium-pricing rule fits the dual-pricing allocation in Section~\ref{ss:dual-pricing-selector}.

\begin{example}[Normal entropic expansion]\label{ex:normal-entropic-expansion}
Assume the entropic setup of Example~\ref{ex:entropic-dual}, and suppose the coalition aggregate risks are normally distributed. If $X_A$ has standard deviation $\sigma_A$, then $\psi_A=\frac12(\sigma_A/\iota_A)^2$. If each $\zeta_i$, $i\in[n]$, is normal with standard deviation $\sigma_i$, then the actuarially fair welfare gain in coalition $A\subseteq[n]$ is $w_i^{\AF}(A)=(\iota_i/2)\{(\sigma_i/\iota_i)^2-(\sigma_A/\iota_A)^2\}$ for $i\in A$. Thus, individual rationality under actuarial fairness is governed by the ratio $\sigma_A/\iota_A$. For a one-step expansion $B=A\cup\{j\}$, the actuarially fair incumbent costs weakly decrease when $\sigma_B/\iota_B\le \sigma_A/\iota_A$, and the entrant is willing to join when $\sigma_B/\iota_B\le\sigma_j/\iota_j$; these are exactly the strong-consensus inequalities in the terminology of \citet{feng2024expansion}, while their weak-consensus condition requires the new aggregate ratio to be no larger than the incumbent individual ratios.

For jointly normal endowments, exponential tilting gives $\E^{\Q_A^\star}[\zeta_i] = \mu_i + \Cov(\zeta_i,X_A)/\iota_A$ and $\E^{\Q_A^\star}[\log\pi_A^\star]=\Var(X_A)/(2\iota_A^2)$. Consequently, $C_i^\star(A)=\mu_i+\Cov(\zeta_i,X_A)/\iota_A-\iota_i\Var(X_A)/(2\iota_A^2)$ for $i\in A$.
The dual-pricing allocation is a PMAS if and only if, for every $\varnothing\ne A\subseteq B\subseteq[n]$ and every $i\in A$,
$$\frac{\Cov(\zeta_i,X_B)}{\iota_B} - \frac{\iota_i}{2\iota_B^2}\Var(X_B) \le \frac{\Cov(\zeta_i,X_A)}{\iota_A} - \frac{\iota_i}{2\iota_A^2}\Var(X_A).$$
Equivalently, since $w_i^\star(A)=(2\iota_i)^{-1}\Var(\zeta_i-(\iota_i/\iota_A)X_A)$, entry monotonicity under dual pricing is equivalent to
$$ \Var\left(\zeta_i-\frac{\iota_i}{\iota_B}X_B\right) \ge \Var\left(\zeta_i-\frac{\iota_i}{\iota_A}X_A\right), \qquad \varnothing\ne A\subseteq B\subseteq[n],\ i\in A.$$
This identity explains why equilibrium pricing automatically gives fixed-coalition stability in the normal entropic model, and why strong consensus under entry is a sharper requirement than weak consensus.
\end{example}

The remaining applications are proportional conditional-mean risk-sharing models. In each case, the conditional expectation of an individual endowment, given the coalition aggregate, equals the deterministic mean share of that aggregate. Combined with the proportional penalty split $\alpha_i(\pi_A^\star)=(\mu_i/\mu_A)\alpha_A(\pi_A^\star)$, which holds automatically for coherent risk measures, Theorem~\ref{thm:proportional-cost-dual-pricing-PMAS} identifies the dual-pricing allocation with the mean-proportional allocation and yields a PMAS.

The first such model is the convolution-semigroup case, in which each agent's endowment has the law of a common nonnegative subordinator at a time parameter equal to the agent's mean. This family covers several classical actuarial aggregate-loss models, including compound Poisson claim totals, gamma aggregate losses, and inverse-Gaussian aggregate-loss models, with the mean parameter $\mu_i$ for $i\in[n]$ scaling each agent's exposure. If $(Z_t)_{t\ge0}$ is an integrable subordinator with Laplace exponent $\Psi$, then $\E[e^{-uZ_t}]=\exp\{-t\Psi(u)\}$ for $u\ge0$, the laws of $Z_t$ form a convolution semigroup, and each marginal law is infinitely divisible. The normalization $\partial_+\Psi(0)=1$, where $\partial_+\Psi(0):=\lim_{u\downarrow0}\{\Psi(u)-\Psi(0)\}/u$ is the finite right derivative at the origin, scales the semigroup parameter to equal the mean; see \citet{sato1999levy}.

\begin{corollary}\label{cor:convolution-semigroup-PMAS}
Under Assumption~\ref{ass:dual-pricing}, suppose $\zeta_1,\dots,\zeta_n$ are independent and $\E[e^{-t\zeta_i}]=e^{-\mu_i\Psi(t)}$ for $t\ge0$, where $\mu_i>0$ for $i\in[n]$ and $\Psi$ is the Laplace exponent of an integrable nonnegative convolution semigroup with $\partial_+\Psi(0)=1$, and suppose $\alpha_i(\pi_A^\star)=(\mu_i/\mu_A)\alpha_A(\pi_A^\star)$ for every nonempty $A\subseteq[n]$ and every $i\in A$. The dual-pricing allocation coincides with the mean-proportional allocation, $C_i^\star(A)=\mu_iC(A)/\mu_A$ for every nonempty $A\subseteq[n]$ and every $i\in A$, and is a PMAS.
\end{corollary}
\begin{proof}

For every nonempty $A\subseteq[n]$ and every $i\in A$, independence gives, for $t\ge0$, $\mathcal{L}_A(t):=\E[e^{-tX_A}]=e^{-\mu_A\Psi(t)}$
and
$$ \E[\zeta_i e^{-tX_A}] = \mu_i\Psi'(t)e^{-\mu_A\Psi(t)} = \frac{\mu_i}{\mu_A}(-\mathcal{L}_A'(t)).$$
Proposition~\ref{prop:laplace-proportionality} yields the proportional conditional means \eqref{eq:proportional-cmrs-main}. With the proportional penalty split, Theorem~\ref{thm:proportional-cost-dual-pricing-PMAS} gives the coincidence with the mean-proportional allocation and the PMAS property.
\end{proof}

Dirichlet--Liouville models reallocate a random aggregate loss across participants via Dirichlet-distributed shares that are independent of the loss. This fits settings such as catastrophe pools, in which aggregate event severity is random, and the participants' losses fluctuate around a fixed exposure profile.

\begin{corollary}\label{cor:dirichlet-liouville-PMAS}
Under Assumption~\ref{ass:dual-pricing}, fix $\tau>0$ and $\mu_i>0$ for $i\in[n]$, and write $\mu_A:=\sum_{i\in A}\mu_i$ and $\mu:=\sum_{i=1}^n\mu_i$. Suppose $\bm{\zeta}=(\zeta_1,\dots,\zeta_n)$ has a Dirichlet--Liouville distribution with concentration $\tau$ and mean vector $\bm{\mu}=(\mu_1,\dots,\mu_n)$: $\bm{\zeta}=R\bm{P}$ almost surely, where $R\ge0$ is integrable with $\E[R]=\mu$ and $\bm{P}=(P_1,\dots,P_n)\sim\mathrm{Dirichlet}(\tau\mu_1,\dots,\tau\mu_n)$ is independent of $R$. Suppose also that $\alpha_i(\pi_A^\star)=(\mu_i/\mu_A)\alpha_A(\pi_A^\star)$ for every nonempty $A\subseteq[n]$ and every $i\in A$. The dual-pricing allocation coincides with the mean-proportional allocation and is a PMAS.
\end{corollary}

\begin{proof}
By the representation, $\zeta_i=RP_i$ for every $i\in[n]$. Dirichlet aggregation and neutrality give $\E[P_i\mid P_A]=(\mu_i/\mu_A)P_A$, where $P_A:=\sum_{i\in A}P_i$, for every nonempty $A\subseteq[n]$ and every $i\in A$.
Since $R$ is independent of the Dirichlet vector, $\E[\zeta_i\mid X_A]=(\mu_i/\mu_A)X_A$, which is \eqref{eq:proportional-cmrs-main}. With the proportional penalty split, Theorem~\ref{thm:proportional-cost-dual-pricing-PMAS} gives the coincidence with the mean-proportional allocation and the PMAS property.
\end{proof}

Conditional-convolution models with a shared latent factor cover shared-frailty and mixed-Poisson constructions common in actuarial portfolios: the factor $\Theta$ encodes exposure to a common source of risk, and the endowments are conditionally independent given $\Theta = \theta$. They recover the Archimedean copula family in the multiplicative-frailty case $\Psi_\theta(t)=\theta\,\Psi(t)$, for $\theta > 0$, in which the Laplace transform of $\Theta$ generates the copula.

\begin{corollary}\label{cor:conditional-convolution-PMAS}
Under Assumption~\ref{ass:dual-pricing}, suppose $\zeta_1,\dots,\zeta_n$ are conditionally independent given a real-valued latent factor $\Theta$, with $\E[e^{-t\zeta_i}\mid\Theta]=e^{-\mu_i\Psi_\Theta(t)}$ for $t\ge0$ and $\mu_i>0$ for $i\in[n]$, and $\E[\partial_+\Psi_\Theta(0)]=1$. Suppose also that $\alpha_i(\pi_A^\star)=(\mu_i/\mu_A)\alpha_A(\pi_A^\star)$ for every nonempty $A\subseteq[n]$ and every $i\in A$. The dual-pricing allocation coincides with the mean-proportional allocation and is a PMAS.
\end{corollary}

\begin{proof}
For every nonempty $A\subseteq[n]$ and every $i\in A$, conditional independence gives, for $t\ge0$, $\E[e^{-tX_A}\mid\Theta]=e^{-\mu_A\Psi_\Theta(t)}$ and $\E[\zeta_i e^{-tX_A}\mid\Theta]=(\mu_i/\mu_A)(-\partial_t\E[e^{-tX_A}\mid\Theta])$. Proposition~\ref{prop:laplace-proportionality} applied conditionally on $\Theta$ yields $\E[\zeta_i\mid X_A,\Theta]=(\mu_i/\mu_A)X_A$ almost surely; the tower property and the deterministic coefficient give $\E[\zeta_i\mid X_A]=(\mu_i/\mu_A)X_A$, which is \eqref{eq:proportional-cmrs-main}. With the proportional penalty split, Theorem~\ref{thm:proportional-cost-dual-pricing-PMAS} gives the coincidence with the mean-proportional allocation and the PMAS property.
\end{proof}

The preceding corollaries establish PMAS by showing that dual pricing coincides with the mean-proportional allocation. However, PMAS transfer rules need not coincide with the dual price. The mean-proportional allocation \eqref{eq:proportional-selector-main} charges each agent the same cost per unit of mean exposure, and the following example shows that it can support a PMAS even when the conditional-mean shares are not proportional, so this PMAS is not a dual price. 

\begin{corollary}\label{cor:gamma-pmas-not-pricing}
Fix $\alpha>0$ and $0<a<b$, let $G_1,\dots,G_n$ be independent and identically distributed $\mathrm{Gamma}(\alpha,1)$ random variables, and set $\zeta_1=aG_1/\alpha$ and $\zeta_i=bG_i/\alpha$ for $i\in[n]\setminus\{1\}$, so that $\mu_1=a$ and $\mu_i=b$ for $i\in[n]\setminus\{1\}$. Suppose Assumption~\ref{ass:dual-pricing} holds. Suppose also that, for every $\varnothing\ne A\subseteq B\subseteq[n]$, $R_B^{\bm \mu}(\bar X_A)\le R_A^{\bm \mu}(\bar X_A)$. Then the mean-proportional allocation $C_i^{\bm \mu}(A)=\mu_iC(A)/\mu_A$, $i\in A$, is a PMAS. If, for some $j\in[n]\setminus\{1\}$ and $A=\{1,j\}$, the selected coherent dual optimizer $\pi_A^\star$ is a nonconstant nondecreasing function of $X_A$, then this PMAS does not coincide with the dual-pricing allocation.
\end{corollary}

\begin{proof}
Set $Z_i:=G_i/\alpha$, so the $Z_i$ are independent and identically distributed with mean one and $X_A/\mu_A=\sum_{i\in A}(\mu_i/\mu_A)Z_i$. Fix $\emptyset\ne A\subseteq B\subseteq[n]$. Because the exposures $(\mu_i)_{i\in[n]}$ take only the two values $a$ and $b$ with $a<b$, comparing decreasing-order partial sums shows that $(\mu_i/\mu_B)_{i\in B}\in\R^B$ is a convex combination of permutations of the vector $p\in\R^B$ with $p_i=\mu_i/\mu_A$ for $i\in A$ and $p_i=0$ for $i\in B\setminus A$. Exchangeability of the $Z_i$ then gives, via Jensen's inequality, $X_B/\mu_B\cx X_A/\mu_A$, which is the normalized convex-order condition \eqref{eq:normalized-cx-main}. Theorem~\ref{thm:mean-proportional-PMAS} delivers the PMAS conclusion.

For the second claim, fix $j\in[n]\setminus\{1\}$ and let $A=\{1,j\}$. The pair $(aZ_1,bZ_j)$ has independent $\mathrm{Gamma}$ marginals with common shape $\alpha$ but distinct rates $\alpha/a$ and $\alpha/b$, so the conditional density of $aZ_1/X_A$ given $X_A=s$ is proportional to $t^{\alpha-1}(1-t)^{\alpha-1}\exp(-\alpha s t(b-a)/(ab))$ on $(0,1)$ and depends on $s$ when $a\ne b$. Consequently $s\mapsto\E[\zeta_j\mid X_A=s]$ is strictly increasing in $s$ and not proportional to $s$, so the proportional conditional-mean identity \eqref{eq:proportional-cmrs-main} fails on $A$. For coherent risk measures, the dual-pricing cost-share difference is $C_j^\star(A)-C_j^{\bm \mu}(A)=\E[\pi_A^\star r(X_A)]$, where $r(s):=\E[\zeta_j\mid X_A=s]-b s/(a+b)$. The residual has mean zero but is not monotone; instead $r(s)/s=\E[1-aZ_1/X_A\mid X_A=s]-b/(a+b)$ is nondecreasing in $s$, because the conditional law of $aZ_1/X_A$ given $X_A=s$ shifts toward smaller values as $s$ increases. Thus $r$ has a single sign change, say at $s_0$, and for $\pi_A^\star=g(X_A)$ nondecreasing, $\E[\pi_A^\star r(X_A)]=\E[(g(X_A)-g(s_0))r(X_A)]>0$ under the stated nonconstant selection. Hence dual pricing charges the high-exposure agent strictly more than the proportional-cost rule, and the two rules differ on $A$.
\end{proof}

\section{Transfer rules and core membership examples}\label{sec:examples}

We now illustrate the theory developed in Sections~\ref{sec:main-results}--\ref{sec:stable-expansion-dual} and Appendices~\ref{app:distortion-background}--\ref{sec:transfer-rules} with two examples. We suppose Assumption~\ref{ass:distortion} holds for the next examples. The first is a Bernoulli model in which all agents share ordered distortions, so the efficient allocation assigns the entire aggregate to the least conservative agent. We will illustrate how different transfer rules allocate the welfare gain from risk sharing. By fixing marginal laws and varying only the dependence structure, we isolate the effect of diversification on the cooperative game, the core geometry, and the four canonical transfer rules. The second is a five-agent comonotonic example designed to illustrate Corollary~\ref{cor:comonotonic-PMAS}: different agents attain the lower envelope on different layers of the aggregate loss, so the efficient allocation is multi-tranche, and the coalition-wise lower-envelope rule yields an explicit PMAS.

\subsection{A Bernoulli illustration}\label{ss:bernoulli-illustration}

Let $\zeta_i:=B_i$ for $i\in\{1,2,3\}$, where each $B_i$ is Bernoulli with $\P(B_1=1)=0.3$, $\P(B_2=1)=0.4$, and $\P(B_3=1)=0.5$. We fix the marginal laws throughout and vary only the dependence structure of $(B_1,B_2,B_3)$; doing so isolates the diversification channel in the gain game $\nu$ while keeping stand-alone evaluations unchanged. For $i\in\{1,2,3\}$, agent $i$ evaluates losses through the dual-power distortion $h_i(u):=1-(1-u)^{\theta_i}$ with parameters $\theta_1=1.4$, $\theta_2=2.5$, and $\theta_3=3$. Since $\theta_1<\theta_2<\theta_3$, one has $h_1\le h_2\le h_3$ pointwise, so agent~1 is the least conservative and the lower envelope is $h_\wedge=h_1$.

For integer-valued losses, the distortion risk measure reduces to $\rho_{h_i}(Y)=\sum_{\ell=1}^{K} h_i\bigl(\P(Y\ge \ell)\bigr)$ for $i\in\{1,2,3\}$, so coalition costs and welfare gains may be computed directly from tail probabilities. Since $h_1\le h_2\le h_3$, every layer of the aggregate loss $X:=B_1+B_2+B_3$ is priced most favourably by agent~1; the inf-convolution is therefore $(\rho_{h_1}\square \rho_{h_2}\square \rho_{h_3})(X)=\rho_{h_1}(X)$, with efficient allocation $X_1=X$, $X_2=X_3=0$. The economic content of the transfer problem is not how to split the random loss, but how to split the welfare gain created by letting the least conservative agent absorb it. The autarky evaluations are $\rho_1(B_1)=0.393072$, $\rho_2(B_2)=0.721145$, and $\rho_3(B_3)=0.875$, and remain the same across all four dependence structures described below.

We consider four dependence structures for $(B_1,B_2,B_3)$ that preserve the same Bernoulli margins as above. Let $U\sim\mathrm{Unif}(0,1)$. The first case is independence, which serves as the baseline. The second is full comonotonicity, in which $B_i=\id_{\{U\le p_i\}}$ for $i\in\{1,2,3\}$. The third and fourth cases are mixed structures that combine a comonotonic pair with one agent moving in the opposite direction: in the third, agents~$1$ and~$2$ are comonotonic while agent~$3$ is counter-monotone, so that $B_1=\id_{\{U\le 0.3\}}$, $B_2=\id_{\{U\le 0.4\}}$, and $B_3=\id_{\{U>0.5\}}$; in the fourth, agents~$2$ and~$3$ are comonotonic while agent~$1$ is counter-monotone, so that $B_1=\id_{\{U>0.7\}}$, $B_2=\id_{\{U\le 0.4\}}$, and $B_3=\id_{\{U\le 0.5\}}$. The remaining mixed configuration, namely $(1,3)$ comonotonic with agent~$2$ counter-monotone, is qualitatively similar and is omitted. Since only the law of $X$ changes, the welfare simplex, cooperative game, core, and all transfer rules are modified solely through the dependence channel.

Table~\ref{tab:bernoulli-dependence-game} records the induced welfare game: the grand-coalition surplus $W$, the three pair-coalition values, and whether the game is convex. Table~\ref{tab:bernoulli-dependence-rules} collects the welfare-gain vectors, surplus shares, and IR/core diagnostics for actuarial fairness, $\Q^\star$-pricing, the Shapley value, the nucleolus, and the mean-weighted proportional-cost rule (PC) from \eqref{eq:proportional-selector-main} (i.e.\ $\xi=\mu_i=p_i$ for every $i\in[3]$) under each dependence structure. Figure~\ref{fig:bernoulli-dependence-comparison} displays the resulting welfare geometry; in each panel, the coordinates are $(w_1,w_2)$ with $w_3=W-w_1-w_2$, the grey triangle is the IR simplex, the blue polygon is the core, and the markers indicate the five rules.

\begin{table}[ht]
\centering
\begin{tabular}{lccccc}
\toprule
Scenario & $W=\nu(\{1,2,3\})$ & $\nu(\{1,2\})$ & $\nu(\{1,3\})$ & $\nu(\{2,3\})$ & Convex \\
\midrule
Independent & 0.5658 & 0.2472 & 0.2946 & 0.2179 & Yes \\
All comonotonic & 0.4642 & 0.2103 & 0.2539 & 0.0518 & Yes \\
$(1,2)$ com., $3$ counter & 0.6360 & 0.2103 & 0.3731 & 0.5994 & No \\
$(2,3)$ com., $1$ counter & 0.5834 & 0.2996 & 0.3731 & 0.0518 & No \\
\bottomrule
\end{tabular}
\caption{The welfare game induced by each dependence structure; singleton coalition values are zero in all four cases.}
\label{tab:bernoulli-dependence-game}
\end{table}

\begin{figure}[!htb]
\centering
\begin{tikzpicture}
\begin{groupplot}[
  group style={group size=2 by 2, horizontal sep=1.3cm, vertical sep=2.1cm},
  width=0.34\textwidth,
  height=0.34\textwidth,
  axis lines=left,
  tick label style={font=\scriptsize},
  label style={font=\small},
  title style={font=\small},
  clip=false,
  scale only axis
]

\nextgroupplot[
  title={I. Independent},
  ylabel={$w_2$},
  xmin=-0.18, xmax=0.60,
  ymin=0.00, ymax=0.60
]
\addplot[draw=gray!70, dashed, fill=gray!20, line width=0.8pt]
  coordinates {(0,0) (0.5658,0) (0,0.5658)}
\closedcycle;
\addplot[draw=RoyalBlue!80!black, fill=Cyan!15, line width=0.9pt]
  coordinates {(0.3480,0.0000) (0.2472,0.0000) (0.0000,0.2472) (0.0000,0.2713) (0.2946,0.2713) (0.3480,0.2179)}
\closedcycle;
\addplot[draw=RoyalBlue!80!black, fill=Cyan!30, line width=0.9pt]
  coordinates {(0.3480,0) (0.2472,0) (0,0.2472) (0,0.2713) (0.2946,0.2713) (0.3480,0.2179) (0.3480,0)};
\addplot[only marks, mark=triangle*, mark options={draw=BrickRed, fill=BrickRed, fill opacity=1}, color=BrickRed, line width=0.8pt, mark size=2.8pt]
  coordinates {(-0.1303,0.3211)};
\addplot[only marks, mark=*, mark options={draw=ForestGreen!70!black, fill=ForestGreen!70!black}, color=ForestGreen!70!black, mark size=2.5pt]
  coordinates {(0.0299,0.2448)};
\addplot[only marks, mark=square*, mark options={draw=Magenta!80!black, fill=Magenta!80!black}, color=Magenta!80!black, mark size=2pt]
  coordinates {(0.2063,0.1679)};
\addplot[only marks, mark=diamond*, mark options={draw=Orange!80!black, fill=Orange!80!black}, color=Orange!80!black, mark size=2.4pt]
  coordinates {(0.2240,0.1473)};
\addplot[only marks, mark=pentagon*, mark options={draw=Brown!80!black, fill=Brown!80!black}, color=Brown!80!black, mark size=2.5pt]
  coordinates {(0.0372,0.2467)};

\nextgroupplot[
  title={II. All comonotonic},
  xmin=-0.30, xmax=0.52,
  ymin=0.00, ymax=0.50
]
\addplot[draw=gray!70, dashed, fill=gray!20, line width=0.8pt]
  coordinates {(0,0) (0.4642,0) (0,0.4642)}
\closedcycle;
\addplot[draw=RoyalBlue!80!black, fill=Cyan!15, line width=0.9pt]
  coordinates {(0.4124,0.0518) (0.4124,0.0000) (0.2103,0.0000) (0.0000,0.2103) (0.2539,0.2103)}
\closedcycle;
\addplot[draw=RoyalBlue!80!black, fill=Cyan!30, line width=0.9pt]
  coordinates {(0.4124,0.0518) (0.4124,0) (0.2103,0) (0,0.2103) (0.2539,0.2103) (0.4124,0.0518)};
\addplot[only marks, mark=triangle*, mark options={draw=BrickRed, fill=BrickRed, fill opacity=1}, color=BrickRed, line width=0.8pt, mark size=2.8pt]
  coordinates {(-0.2320,0.3211)};
\addplot[only marks, mark=*, mark options={draw=ForestGreen!70!black, fill=ForestGreen!70!black}, color=ForestGreen!70!black, mark size=2.5pt]
  coordinates {(0.0000,0.2103)};
\addplot[only marks, mark=square*, mark options={draw=Magenta!80!black, fill=Magenta!80!black}, color=Magenta!80!black, mark size=2pt]
  coordinates {(0.2148,0.1138)};
\addplot[only marks, mark=diamond*, mark options={draw=Orange!80!black, fill=Orange!80!black}, color=Orange!80!black, mark size=2.4pt]
  coordinates {(0.2749,0.0728)};
\addplot[only marks, mark=pentagon*, mark options={draw=Brown!80!black, fill=Brown!80!black}, color=Brown!80!black, mark size=2.5pt]
  coordinates {(0.0118,0.2128)};

\nextgroupplot[
  title={III. (1,2) com., 3 counter.},
  xlabel={$w_1$},
  ylabel={$w_2$},
  xmin=-0.10, xmax=0.68,
  ymin=0.00, ymax=0.68
]
\addplot[draw=gray!70, dashed, fill=gray!20, line width=0.8pt]
  coordinates {(0,0) (0.6360,0) (0,0.6360)}
\closedcycle;
\addplot[draw=RoyalBlue!80!black, fill=Cyan!15, line width=0.9pt]
  coordinates {(0.2103,0.0000) (0.0366,0.0000) (0.0000,0.2103) (0.0000,0.2628) (0.0366,0.5993) (0.3731,0.2628)}
\closedcycle;
\addplot[draw=RoyalBlue!80!black, fill=Cyan!30, line width=0.9pt]
  coordinates {(0.0366,0.1736) (0.0000,0.2103) (0.0000,0.2628) (0.0366,0.2628) (0.0366,0.1736)};
\addplot[only marks, mark=triangle*, mark options={draw=BrickRed, fill=BrickRed, fill opacity=1}, color=BrickRed, line width=0.8pt, mark size=2.8pt]
  coordinates {(-0.0602,0.3211)};
\addplot[only marks, mark=*, mark options={draw=ForestGreen!70!black, fill=ForestGreen!70!black}, color=ForestGreen!70!black, mark size=2.5pt]
  coordinates {(0.0000,0.2336)};
\addplot[only marks, mark=square*, mark options={draw=Magenta!80!black, fill=Magenta!80!black}, color=Magenta!80!black, mark size=2pt]
  coordinates {(0.1094,0.2225)};
\addplot[only marks, mark=diamond*, mark options={draw=Orange!80!black, fill=Orange!80!black}, color=Orange!80!black, mark size=2.4pt]
  coordinates {(0.0069,0.2332)};
\addplot[only marks, mark=pentagon*, mark options={draw=Brown!80!black, fill=Brown!80!black}, color=Brown!80!black, mark size=2.5pt]
  coordinates {(0.0548,0.2701)};

\nextgroupplot[
  title={IV. (2,3) com., 1 counter.},
  xlabel={$w_1$},
  xmin=-0.15, xmax=0.60,
  ymin=0.00, ymax=0.60
]
\addplot[draw=gray!70, dashed, fill=gray!20, line width=0.8pt]
  coordinates {(0,0) (0.5834,0) (0,0.5834)}
\closedcycle;
\addplot[draw=RoyalBlue!80!black, fill=Cyan!15, line width=0.9pt]
  coordinates {(0.5316,0.0518) (0.5316,0.0000) (0.2996,0.0000) (0.0000,0.2103) (0.0000,0.2996) (0.3731,0.2103)}
\closedcycle;
\addplot[draw=RoyalBlue!80!black, fill=Cyan!30, line width=0.9pt]
  coordinates {(0.5316,0.0518) (0.5316,0.0000) (0.2996,0.0000) (0.0893,0.2103) (0.3731,0.2103) (0.5316,0.0518)};
\addplot[only marks, mark=triangle*, mark options={draw=BrickRed, fill=BrickRed, fill opacity=1}, color=BrickRed, line width=0.8pt, mark size=2.8pt]
  coordinates {(-0.1128,0.3211)};
\addplot[only marks, mark=*, mark options={draw=ForestGreen!70!black, fill=ForestGreen!70!black}, color=ForestGreen!70!black, mark size=2.5pt]
  coordinates {(0.1050,0.2103)};
\addplot[only marks, mark=square*, mark options={draw=Magenta!80!black, fill=Magenta!80!black}, color=Magenta!80!black, mark size=2pt]
  coordinates {(0.2893,0.1286)};
\addplot[only marks, mark=diamond*, mark options={draw=Orange!80!black, fill=Orange!80!black}, color=Orange!80!black, mark size=2.4pt]
  coordinates {(0.3842,0.0629)};
\addplot[only marks, mark=pentagon*, mark options={draw=Brown!80!black, fill=Brown!80!black}, color=Brown!80!black, mark size=2.5pt]
  coordinates {(0.0416,0.2525)};

\end{groupplot}
\end{tikzpicture}
\vspace{0.4em}

\begin{tabular}{llllllll}
\tikz[baseline=-0.6ex]\draw[gray!70, dashed, fill=gray!20] (0,0) rectangle (0.28,0.12);~IR simplex &
\tikz[baseline=-0.6ex]\draw[RoyalBlue!80!black, fill=Cyan!15] (0,0) rectangle (0.28,0.12);~Weber set &
\tikz[baseline=-0.6ex]\draw[RoyalBlue!80!black, fill=Cyan!30] (0,0) rectangle (0.28,0.12);~Core &
\tikz[baseline=-0.6ex]\draw[BrickRed, fill=BrickRed, line width=0.8pt] (0,0.12) -- (0.10,-0.06) -- (-0.10,-0.06) -- cycle;~AF &
\tikz[baseline=-0.6ex]\filldraw[draw=ForestGreen!70!black, fill=ForestGreen!70!black] (0,0) circle (1.6pt);~$\Q^\star$ &
\tikz[baseline=-0.6ex]\filldraw[draw=Magenta!80!black, fill=Magenta!80!black] (-0.08,-0.08) rectangle (0.08,0.08);~Shapley &
\tikz[baseline=-0.6ex]\filldraw[draw=Orange!80!black, fill=Orange!80!black, rotate=45] (-0.09,-0.09) rectangle (0.09,0.09);~Nucleolus &
\tikz[baseline=-0.6ex]\filldraw[draw=Brown!80!black, fill=Brown!80!black] (0,0.10) -- (0.095,0.031) -- (0.059,-0.081) -- (-0.059,-0.081) -- (-0.095,0.031) -- cycle;~PC
\end{tabular}
\caption{IR simplex, Weber set, core, and five transfer rules in the Bernoulli illustration; the plots are drawn directly from the welfare coordinates of the four dependence structures described above. In panels~I and~II the game is convex, so the Weber set coincides with the core. The nucleolus lies in the core in all four cases and consistently tilts toward the risk-absorbing agent relative to the Shapley value. The mean-weighted proportional-cost rule (PC) lies in the core only in panel~I.}
\label{fig:bernoulli-dependence-comparison}
\end{figure}
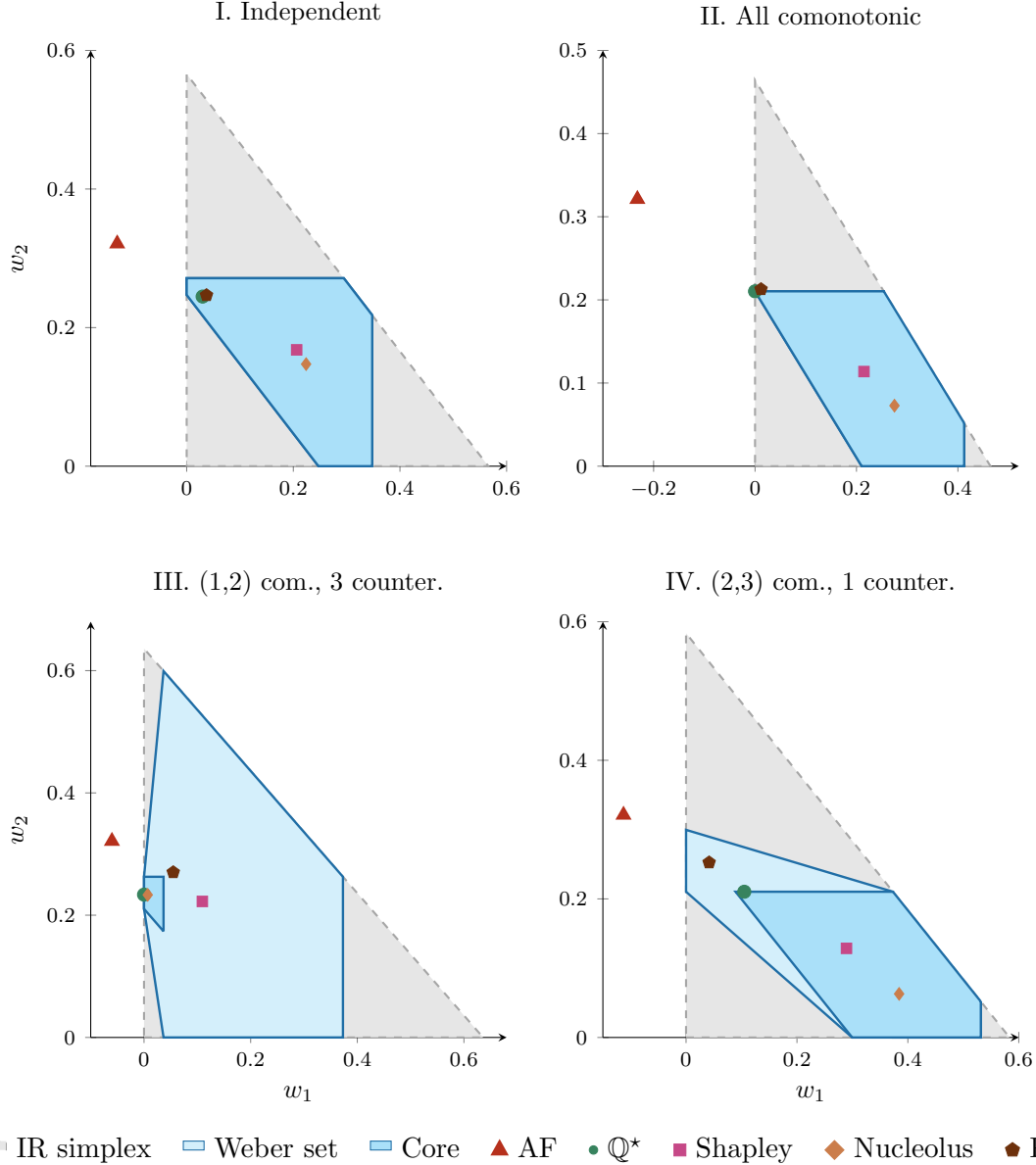

\begin{table}[ht]
\centering
\begin{tabular}{llccccccccc}
\toprule
Scenario & Rule & $W$ & $w_1$ & $w_2$ & $w_3$ & IR & Weber & Core & EM \\
\midrule
\multicolumn{2}{l}{I. Independent} & 0.5658 &  &  &  &  &  &  &  \\
 & AF &  & -0.1303 & 0.3211 & 0.3750 & N & N & N & N \\
 & $\Q^\star$ &  & 0.0299 & 0.2448 & 0.2912 & Y & Y & Y & Y \\
 & Shapley &  & 0.2063 & 0.1679 & 0.1916 & Y & Y & Y & Y \\
 & Nucleolus &  & 0.2240 & 0.1473 & 0.1946 & Y & Y & Y & Y \\
 & PC &  & 0.0372 & 0.2467 & 0.2819 & Y & Y & Y & Y \\
\addlinespace
\multicolumn{2}{l}{II. All comonotonic} & 0.4642 &  &  &  &  &  &  &  \\
 & AF &  & -0.2320 & 0.3211 & 0.3750 & N & N & N & N \\
 & $\Q^\star$ &  & 0.0000 & 0.2103 & 0.2539 & Y & Y & Y & Y \\
 & Shapley &  & 0.2148 & 0.1138 & 0.1356 & Y & Y & Y & Y \\
 & Nucleolus &  & 0.2749 & 0.0728 & 0.1164 & Y & Y & Y & N \\
 & PC &  & 0.0118 & 0.2128 & 0.2396 & Y & N & N & N \\
\addlinespace
\multicolumn{2}{l}{III. $(1,2)$ com., $3$ counter} & 0.6360 &  &  &  &  &  &  &  \\
 & AF &  & -0.0602 & 0.3211 & 0.3750 & N & N & N & N \\
 & $\Q^\star$ &  & 0.0000 & 0.2336 & 0.4024 & Y & Y & Y & N \\
 & Shapley &  & 0.1094 & 0.2225 & 0.3040 & Y & Y & N & N \\
 & Nucleolus &  & 0.0069 & 0.2332 & 0.3960 & Y & Y & Y & N \\
 & PC &  & 0.0548 & 0.2701 & 0.3111 & Y & Y & N & N \\
\addlinespace
\multicolumn{2}{l}{IV. $(2,3)$ com., $1$ counter} & 0.5834 &  &  &  &  &  &  &  \\
 & AF &  & -0.1128 & 0.3211 & 0.3750 & N & N & N & N \\
 & $\Q^\star$ &  & 0.1050 & 0.2103 & 0.2681 & Y & Y & Y & N \\
 & Shapley &  & 0.2893 & 0.1286 & 0.1654 & Y & Y & Y & N \\
 & Nucleolus &  & 0.3842 & 0.0629 & 0.1364 & Y & Y & Y & N \\
 & PC &  & 0.0416 & 0.2525 & 0.2892 & Y & Y & N & N \\
\bottomrule
\end{tabular}
\caption{Welfare gains and stability diagnostics for five transfer rules under four dependence structures. The proportional-cost row (PC) uses mean weights $\xi_i=\mu_i=p_i$ for every $i\in[3]$. EM marks entry monotonicity, i.e.\ whether incumbents' welfare gains weakly increase along every maximal chain.}
\label{tab:bernoulli-dependence-rules}
\end{table}

The four scenarios highlight different mechanisms. In the fully comonotonic case, diversification is weakest, so the grand-coalition surplus is smallest; no pair coalition becomes disproportionately attractive, the game remains convex, and both $\Q^\star$ and the Shapley value stay in the core. The mixed cases are more instructive. When agents~$1$ and~$2$ are comonotonic with agent~$3$ counter-monotone, agents~$2$ and~$3$ nearly hedge each other: one has $B_2+B_3=1$ with probability $0.9$, so coalition $\{2,3\}$ is already very efficient on its own. The core inequality $w_2+w_3\ge \nu(\{2,3\})$ becomes very restrictive, sharply capping the amount assignable to agent~$1$. This explains both the loss of convexity and the failure of the Shapley value to remain in the core: agent~$1$ contributes substantially when joining weaker coalitions but adds comparatively little to the already strong $\{2,3\}$, so the Shapley average overestimates agent~$1$'s fair share from coalition $\{2,3\}$'s perspective. In the case $(2,3)$ comonotonic with $1$ counter, the game is again non-convex, but now no single pair dominates the core geometry, and the Shapley value remains coalitionally stable.

The $\Q^\star$-pricing rule exhibits a complementary pattern. In the fully comonotonic case, every endowment is an increasing function of the aggregate loss: $B_3=\id_{\{X\ge 1\}}$, $B_2=\id_{\{X\ge 2\}}$, and $B_1=\id_{\{X\ge 3\}}$. Since agent~1 carries the lower envelope ($h_\wedge=h_1$) and $\zeta_1$ is comonotonic with~$X$, the lower-envelope price of $\zeta_1$ equals its stand-alone distortion evaluation, giving $w_1^{\Q^\star}=0$ as confirmed by Table~\ref{tab:bernoulli-dependence-rules}; the welfare-gain vector sits at the vertex $(w_1,w_2)=(0,\,0.2103)$ in panel~2 of Figure~\ref{fig:bernoulli-dependence-comparison}. In case~3, $(1,2)$ comonotonic with $3$ counter, agent~1's endowment remains comonotonic with $X$ because $B_1=\id_{\{X\ge 2\}}$, so $w_1^{\Q^\star}=0$ persists by the same mechanism; however, agent~3's endowment is not an increasing function of $X$, so the welfare-gain vector moves off the vertex onto the boundary edge $w_1=0$ in the relative interior of the core, visible in panel~3. Finally, in case~4, $(2,3)$ comonotonic with $1$ counter, agent~1's endowment $B_1=\id_{\{U>0.7\}}$ is no longer comonotonic with $X$: the aggregate loss takes the value $X=1$ both when $B_1=1$ (for $U>0.7$) and when $B_1=0$ (for $0.4<U\le 0.5$), so this equality no longer holds and $w_1^{\Q^\star}=0.1050>0$.

The EM column of Table~\ref{tab:bernoulli-dependence-rules} tracks entry monotonicity along every maximal chain. AF fails in every scenario, a concrete instance of Theorem~\ref{thm:canonical-not-PMAS}: agent~$1$ absorbs all layers in the efficient allocation, so expectation-preserving transfers strip welfare from agent~$1$ as the coalition grows. The $\Q^\star$-pricing rule is EM in scenarios~I and~II, matching the independent log-concave result of \citet{chenhuwang2017stable} and Corollary~\ref{cor:comonotonic-PMAS}, and fails in the mixed scenarios where lower-envelope premia are no longer antitone under inclusion. The Shapley value is EM precisely in the two convex-game scenarios, in line with Proposition~\ref{prop:convex-PMAS}. The nucleolus is more fragile: it fails EM already in scenario~II despite the game being convex, a three-agent analogue of \citet[Proposition~1]{sonmez1993population}. The mean-weighted proportional-cost rule is EM only in scenario~I, where the normalized aggregate cost $\bar C_\mu(A):=C(A)/\mu_A$ is antitone under inclusion on $[3]$; in scenarios~II--IV the chain $\{1,3\}\subseteq[3]$ already violates antitonicity, since adding agent~$2$ raises the normalized cost, so by Theorem~\ref{thm:PC-PMAS} the rule fails to be a PMAS. In particular, scenario~II realizes the reverse separation: lower-envelope pricing is a PMAS by comonotonicity, while the mean-weighted proportional-cost rule is not.

\subsection{A comonotonic PMAS illustration}\label{ss:pmas-gamma-illustration}

We next illustrate Corollary~\ref{cor:comonotonic-PMAS} in a five-agent example with comonotonic endowments. Let $U\sim\mathrm{Unif}(0,1)$ and define $\zeta_i:=G_i^{-1}(U)$ for $i\in[5]$, where $G_i$, $i\in[5]$, is agent $i$'s marginal distribution function, so $(\zeta_1,\dots,\zeta_5)$ is comonotonic. Table~\ref{tab:pmas-gamma-setup} reports the marginal specifications, distortions, and autarky evaluations for this example.

\begin{table}[ht]
\centering
\renewcommand{\arraystretch}{1.15}
\begin{tabular}{clcr}
\toprule
Agent & Marginal law & Distortion $h_i(u)$ & $\rho_{h_i}(\zeta_i)$ \\
\midrule
1 & $\mathrm{Gamma}(2,1)$ & $u^{0.4}$ & $3.9253$ \\
2 & $\mathrm{Gamma}(4,1)$ & $1-(1-u)^{1.8}$ & $4.9274$ \\
3 & $\mathrm{Gamma}(8,1)$ & $\Phi(\Phi^{-1}(u)+0.4)$ & $9.1692$ \\
4 & $\mathrm{Gamma}(14,1)$ & $u^{0.6}$ & $16.0803$ \\
5 & $0$ & $1-(1-u)^{1.7}$ & $0$ \\
\midrule
\multicolumn{3}{r}{Total autarky cost} & $34.1022$ \\
\bottomrule
\end{tabular}
\caption{Marginal information and autarky evaluations in the five-agent comonotonic PMAS illustration.}
\label{tab:pmas-gamma-setup}
\end{table}

\begin{figure}[t]
\centering
\pgfplotstableread[col sep=comma]{
    x,      h3
    0.00,   1.000000
    0.01,   0.996798
    0.02,   0.992931
    0.04,   0.984250
    0.06,   0.974695
    0.08,   0.964468
    0.10,   0.953672
    0.12,   0.942370
    0.14,   0.930606
    0.16,   0.918410
    0.18,   0.905806
    0.20,   0.892812
    0.22,   0.879440
    0.24,   0.865702
    0.26,   0.851606
    0.28,   0.837157
    0.30,   0.822361
    0.32,   0.807220
    0.34,   0.791737
    0.36,   0.775912
    0.38,   0.759744
    0.40,   0.743234
    0.42,   0.726377
    0.44,   0.709173
    0.45,   0.700438
    0.46,   0.691615
    0.4657, 0.686546
    0.47,   0.682703
    0.48,   0.673700
    0.50,   0.655422
    0.52,   0.636773
    0.54,   0.617746
    0.56,   0.598332
    0.58,   0.578519
    0.60,   0.558297
    0.62,   0.537652
    0.64,   0.516568
    0.66,   0.495028
    0.68,   0.473013
    0.70,   0.450499
    0.71,   0.439047
    0.72,   0.427461
    0.7224, 0.424660
    0.73,   0.415736
    0.74,   0.403869
    0.76,   0.379687
    0.78,   0.354874
    0.80,   0.329382
    0.82,   0.303149
    0.84,   0.276103
    0.86,   0.248151
    0.88,   0.219174
    0.90,   0.189010
    0.92,   0.157431
    0.94,   0.124092
    0.96,   0.088398
    0.98,   0.049089
    0.99,   0.027030
    1.00,   0.000000
}\dataHthree
\begin{tikzpicture}
\begin{axis}[
  width=0.52\textwidth,
  height=0.30\textwidth,
  scale only axis,
  xmin=0, xmax=1,
  ymin=0, ymax=1,
  xlabel={$u$},
  ylabel={$h_i(1-u)$},
  axis lines=left,
  clip=false,
  samples=200,
  domain=0:1,
  legend style={
    at={(1.04,0.5)},
    anchor=west,
    draw=none,
    fill=white,
  },
  legend cell align={left}
]
\addplot[name path=axis, draw=none, forget plot] coordinates {(0,0) (1,0)};

\addplot[name path=env_a4, draw=none, forget plot] {(1-x)^0.6};
\addplot[name path=env_a3, draw=none, forget plot] table[x=x, y=h3] {\dataHthree};
\addplot[name path=env_a5, draw=none, forget plot] {1 - x^1.7};

\addplot[draw=none, fill=BurntOrange!35, forget plot] fill between[
  of=env_a4 and axis,
  soft clip={domain=0:0.465681158177}
];
\addplot[draw=none, fill=ForestGreen!30, forget plot] fill between[
  of=env_a3 and axis,
  soft clip={domain=0.465681158177:0.7224171274}
];
\addplot[draw=none, fill=Plum!35, forget plot] fill between[
  of=env_a5 and axis,
  soft clip={domain=0.7224171274:1}
];

\addplot[RoyalBlue, line width=1.1pt] {(1-x)^0.4};
\addplot[BrickRed, line width=1.1pt] {1 - x^1.8};
\addplot[ForestGreen!70!black, dashed, line width=1.1pt]
  table[x=x, y=h3] {\dataHthree};
\addplot[BurntOrange, dotted, line width=1.1pt] {(1-x)^0.6};
\addplot[Plum!80!black, dashdotted, line width=1.1pt] {1 - x^1.7};

\addplot[BurntOrange, line width=1.9pt, solid, forget plot,
  domain=0:0.465681158177] {(1-x)^0.6};
\addplot[ForestGreen!70!black, line width=1.9pt, solid, forget plot,
  restrict x to domain=0.465681158177:0.7224171274]
  table[x=x, y=h3] {\dataHthree};
\addplot[Plum!80!black, line width=1.9pt, solid, forget plot,
  domain=0.7224171274:1] {1 - x^1.7};

\addlegendentry{$h_1(u)=u^{0.4}$}
\addlegendentry{$h_2(u)=1-(1-u)^{1.8}$}
\addlegendentry{$h_3(u)=\Phi(\Phi^{-1}(u)+0.4)$}
\addlegendentry{$h_4(u)=u^{0.6}$}
\addlegendentry{$h_5(u)=1-(1-u)^{1.7}$}
\addlegendimage{gray!60, solid, line width=1.9pt}
\addlegendentry{$h_\wedge(u)$}
\end{axis}
\end{tikzpicture}
\caption{Distortions in the five-agent PMAS illustration, plotted as functions of $u\mapsto h_i(1-u)$, $i\in[5]$. The shaded lower envelope identifies the grand-coalition risk bearer at each quantile level: agent~4 on low layers, agent~3 on middle layers, and agent~5 on the tail layers.}
\label{fig:pmas-gamma-distortions}
\end{figure}
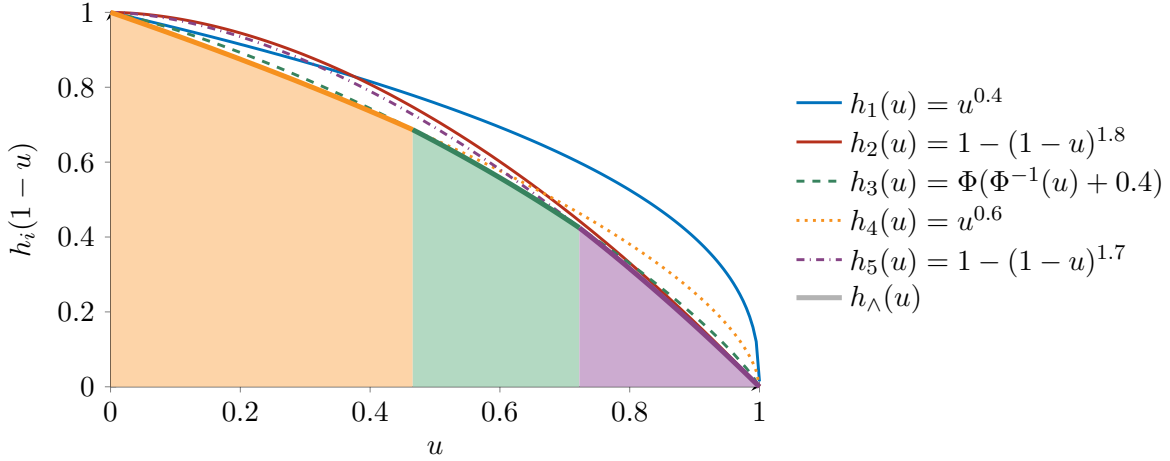

Figure~\ref{fig:pmas-gamma-distortions} shows how the lower-envelope allocation operates in the comonotonic setting. In the grand coalition, the lower envelope is attained on three nondegenerate intervals: agent~4 is cheapest on the body of the loss distribution, agent~3 on the intermediate layers, and the zero-endowment agent~5 on the tail layers. The efficient allocation, therefore, splits the aggregate loss across multiple layers even though the endowments are comonotonic. By Corollary~\ref{cor:comonotonic-PMAS}, the coalition-wise rule $w_i(A)=\rho_{h_i}(\zeta_i)-\rho_{h_\wedge^A}(\zeta_i)$ for $i\in A$ defines a PMAS on the entire coalition lattice. The remainder of the example studies the induced chains and compares the coalition-wise $\Q^\star$ and Shapley welfare splits, with particular attention to the role of agent~5 as a riskless entrant.

\begin{table}[ht]
\centering
\setlength{\tabcolsep}{3.2pt}
\renewcommand{\arraystretch}{1.08}
\resizebox{\textwidth}{!}{%
\begin{tabular}{lccccccccccc}
\toprule
 &  & \multicolumn{5}{c}{$\Q^\star$ welfare split} & \multicolumn{5}{c}{Shapley welfare split} \\
\cmidrule(lr){1-1}\cmidrule(lr){2-2}\cmidrule(lr){3-7}\cmidrule(lr){8-12}
Coalition & Gain & 1 & 2 & 3 & 4 & 5 & 1 & 2 & 3 & 4 & 5 \\
\cmidrule(lr){1-1}\cmidrule(lr){2-2}\cmidrule(lr){3-7}\cmidrule(lr){8-12}
$\{1\}$ & 0.0000 & 0.0000 & -- & -- & -- & -- & 0.0000 & -- & -- & -- & -- \\
$\{1,4\}$ & 1.0354 & 1.0354 & -- & -- & 0.0000 & -- & 0.5177 & -- & -- & 0.5177 & -- \\
$\{1,3,4\}$ & 2.0282 & 1.3544 & -- & 0.0539 & 0.6199 & -- & 0.8471 & -- & 0.6663 & 0.5147 & -- \\
$\{1,2,3,4\}$ & 2.5213 & 1.4192 & 0.2074 & 0.1533 & 0.7415 & -- & 1.0125 & 0.4980 & 0.3841 & 0.6268 & -- \\
$[5]$ & 2.6722 & 1.4442 & 0.2394 & 0.1949 & 0.7937 & 0.0000 & 1.1114 & 0.3207 & 0.2948 & 0.6887 & 0.2567 \\
\bottomrule
\end{tabular}
}
\caption{Coalitionwise lower-envelope pricing and coalition-wise Shapley welfare gains along the chain $1\to4\to3\to2\to5$. Each row is the restricted game induced by the corresponding prefix coalition.}
\label{tab:pmas-gamma-random-chain}
\end{table}

Table~\ref{tab:pmas-gamma-random-chain} reports, on the chain $1\to4\to3\to2\to5$, the coalition-wise lower-envelope pricing rule from Corollary~\ref{cor:comonotonic-PMAS} alongside the coalition-wise Shapley value of the restricted game. This ordering places agent~1 early and isolates the effect of later entrants on the allocation of risk capacity. Under coalition-wise $\Q^\star$ pricing, agent~1 starts alone with zero welfare gain; when agent~4 enters, it creates the entire gain $\nu(\{1,4\})=1.0354$, but that gain is assigned entirely to agent~1 because agent~4 absorbs the cheapest body layer and lowers agent~1's risk measure. As agents~3 and~2 enter, agent~1's $\Q^\star$ payoff still rises, from $1.0354$ to $1.3544$ and then to $1.4192$; thus, the welfare split need not track the share of risk each agent absorbs in the efficient allocation. When the riskless entrant~5 enters last, there is no diversification effect because $\zeta_5=0$, but total welfare still rises from $2.5213$ to $2.6722$: agent~5 adds risk capacity, receives zero welfare gain under $\Q^\star$ pricing, and all incumbent agents benefit. Under the coalition-wise Shapley rule, the pattern is different. Agent~1's Shapley payoff is also increasing along the chain, but the other payoffs are not: agent~4 falls from $0.5177$ to $0.5147$ when agent~3 enters, agent~3 falls from $0.6663$ to $0.3841$ and then to $0.2948$ as agents~2 and~5 enter, and agent~2 falls from $0.4980$ to $0.3207$ when agent~5 enters. Moreover, whereas Corollary~\ref{cor:comonotonic-PMAS} guarantees that the $\Q^\star$ row is in the core at every stage, the coalition-wise Shapley value is in the core only for the first three restricted games $A=\{1\},\{1,4\},\{1,3,4\}$; it leaves the core at $A=\{1,2,3,4\}$, where coalitions $\{1,2,4\}$ and $\{1,3,4\}$ can block, and also at $A=[5]$, where the blocking coalitions are $\{1,5\}$, $\{1,2,4\}$, $\{1,4,5\}$, $\{1,2,3,4\}$, $\{1,2,4,5\}$, and $\{1,3,4,5\}$.

\begin{table}[ht]
\centering
\setlength{\tabcolsep}{3.2pt}
\renewcommand{\arraystretch}{1.08}
\resizebox{\textwidth}{!}{%
\begin{tabular}{lccccccccccc}
\toprule
 &  & \multicolumn{5}{c}{$\Q^\star$ welfare split} & \multicolumn{5}{c}{Shapley welfare split} \\
\cmidrule(lr){1-1}\cmidrule(lr){2-2}\cmidrule(lr){3-7}\cmidrule(lr){8-12}
Coalition & Gain & 1 & 2 & 3 & 4 & 5 & 1 & 2 & 3 & 4 & 5 \\
\cmidrule(lr){1-1}\cmidrule(lr){2-2}\cmidrule(lr){3-7}\cmidrule(lr){8-12}
$\{5\}$ & 0.0000 & -- & -- & -- & -- & 0.0000 & -- & -- & -- & -- & 0.0000 \\
$\{1,5\}$ & 1.3724 & 1.3724 & -- & -- & -- & 0.0000 & 0.6862 & -- & -- & -- & 0.6862 \\
$\{1,3,5\}$ & 1.5695 & 1.4272 & -- & 0.1423 & -- & 0.0000 & 0.9280 & -- & 0.3123 & -- & 0.3292 \\
$\{1,2,3,5\}$ & 1.7772 & 1.4272 & 0.2078 & 0.1423 & -- & 0.0000 & 1.0497 & 0.2515 & 0.2813 & -- & 0.1947 \\
$[5]$ & 2.6722 & 1.4442 & 0.2394 & 0.1949 & 0.7937 & 0.0000 & 1.1114 & 0.3207 & 0.2948 & 0.6887 & 0.2567 \\
\bottomrule
\end{tabular}
}
\caption{Coalitionwise lower-envelope pricing and coalition-wise Shapley welfare gains along the chain $5\to1\to3\to2\to4$. Each row is the restricted game induced by the corresponding prefix coalition.}
\label{tab:pmas-gamma-capacity-first-chain}
\end{table}

Table~\ref{tab:pmas-gamma-capacity-first-chain} reports the chain $5\to1\to3\to2\to4$, in which the riskless entrant enters first. Under coalition-wise $\Q^\star$ pricing, agent~5 always receives zero welfare gain: since $\zeta_5=0$, there is no diversification effect and no gain from having its own endowment re-evaluated, so $w_5(A)=\rho_5(0)-\rho_{h_\wedge^A}(0)=0$ for every coalition $A\subseteq[5]$ containing agent~5. When agent~2 enters the coalition $\{1,3,5\}$, the entire welfare increment, from $1.5695$ to $1.7772$, is assigned to agent~2; the payoffs of agents~1 and~3 do not move. This occurs because coalition-wise $\Q^\star$ pricing allocates welfare through the change in the lower-envelope distortion: a new entrant gains welfare when its distortion becomes the cheapest on layers of its own endowment but leaves the incumbents' valuations essentially unchanged, whereas incumbents gain as well when the entrant lowers the envelope on layers that are active in their own endowments. The final step, when agent~4 enters, has the latter form: the lower envelope shifts on layers that matter for several players, so agents~1,~2, and~3 all gain in addition to agent~4. Under the coalition-wise Shapley rule, the pattern is different. Agent~1's Shapley payoff rises monotonically once pooling starts, from $0.6862$ to $1.1114$, but agent~5 receives positive welfare in every non-singleton coalition and its payoff moves non-monotonically, from $0.6862$ to $0.3292$ to $0.1947$ and then back up to $0.2567$. Hence, the Shapley value is not entry-monotone in this example and therefore is not a PMAS. As in the first chain, the $\Q^\star$ row is in the core at every stage by Corollary~\ref{cor:comonotonic-PMAS}; by contrast, the coalition-wise Shapley value is in the core only for the first two restricted games $A=\{5\}$ and $A=\{1,5\}$. It already leaves the core at $A=\{1,3,5\}$, where $\{1,5\}$ and $\{1,3\}$ can block, and it remains outside the core at $A=\{1,2,3,5\}$ and at $A=[5]$.

These two tables display only two of the $5!=120$ possible maximal chains. PMAS is stronger than monotonicity along maximal chains, since it must also handle inclusions where several agents enter at once; maximal chains alone cover only single-entrant steps. Exhaustively checking the coalition-wise Shapley rule across all $120$ permutations shows that only four chains are entry-monotone for incumbents: $2\to5\to3\to1\to4$, $2\to5\to3\to4\to1$, $5\to2\to3\to1\to4$, and $5\to2\to3\to4\to1$. These four chains have the same ordering structure: among the agents with nonzero endowments, the monotone Shapley orders are $2\to3\to(1,4)$, while the riskless entrant~5 may be inserted only in one of the two early positions, namely before $2$ or between $2$ and $3$. In particular, agent~5 cannot enter last in any entry-monotone Shapley chain: for instance, along $2\to3\to1\to4\to5$, the final entry of agent~5 lowers the Shapley payoffs of agents~2 and~3 from $0.4980$ and $0.3841$ to $0.3207$ and $0.2948$. Nor can agent~5 enter after agent~3, since already along $2\to3\to5\to1\to4$ the entry of~5 lowers agent~2's payoff from $0.1372$ to $0.1299$. This ordering is not dictated by a simple theorem that ranks distortions alone, nor by marginal risks alone; rather, it is a consequence of the full cooperative game generated by the chosen marginals and distortions through the Shapley average-marginal-contribution formula. In this example, entry monotonicity is preserved only when agents~2 and~3 enter before agents~1 and~4, with agent~5 inserted before~2 or between~2 and~3; other orderings reassign average marginal contributions away from incumbents and violate monotonicity. Once the coalition $\{2,3,5\}$ has formed, either agent~1 or agent~4 may enter next without reducing an incumbent's payoff. Table~\ref{tab:pmas-gamma-shapley-monotone-chain} reports the welfare splits along one such chain, $5\to2\to3\to1\to4$, on which every incumbent's Shapley payoff is weakly increasing at each entry. None of these four chains stays in the core at every prefix coalition: in the chain of Table~\ref{tab:pmas-gamma-shapley-monotone-chain}, the Shapley row is outside the core at $\{1,2,3,5\}$ and at $[5]$. Hence, coalition-wise Shapley entry monotonicity holds only for these four permutations, and compatibility with the core along an entire chain never holds; the total number of permutations for which the coalition-wise Shapley allocations form a sequential PMAS is zero.

\begin{table}[ht]
\centering
\setlength{\tabcolsep}{3.2pt}
\renewcommand{\arraystretch}{1.08}
\resizebox{\textwidth}{!}{%
\begin{tabular}{lccccccccccc}
\toprule
 &  & \multicolumn{5}{c}{$\Q^\star$ welfare split} & \multicolumn{5}{c}{Shapley welfare split} \\
\cmidrule(lr){1-1}\cmidrule(lr){2-2}\cmidrule(lr){3-7}\cmidrule(lr){8-12}
Coalition & Gain & 1 & 2 & 3 & 4 & 5 & 1 & 2 & 3 & 4 & 5 \\
\cmidrule(lr){1-1}\cmidrule(lr){2-2}\cmidrule(lr){3-7}\cmidrule(lr){8-12}
$\{5\}$        & 0.0000 & --     & --     & --     & --     & 0.0000 & --     & --     & --     & --     & 0.0000 \\
$\{2,5\}$      & 0.0903 & --     & 0.0903 & --     & --     & 0.0000 & --     & 0.0452 & --     & --     & 0.0452 \\
$\{2,3,5\}$    & 0.3482 & --     & 0.2072 & 0.1410 & --     & 0.0000 & --     & 0.1299 & 0.1552 & --     & 0.0631 \\
$\{1,2,3,5\}$  & 1.7772 & 1.4272 & 0.2078 & 0.1423 & --     & 0.0000 & 1.0497 & 0.2515 & 0.2813 & --     & 0.1947 \\
$[5]$          & 2.6722 & 1.4442 & 0.2394 & 0.1949 & 0.7937 & 0.0000 & 1.1114 & 0.3207 & 0.2948 & 0.6887 & 0.2567 \\
\bottomrule
\end{tabular}
}
\caption{Coalitionwise lower-envelope pricing and coalition-wise Shapley welfare gains along the chain $5\to2\to3\to1\to4$, one of the four maximal chains on which the coalition-wise Shapley rule is entry-monotone. Each row is the restricted game induced by the corresponding prefix coalition.}
\label{tab:pmas-gamma-shapley-monotone-chain}
\end{table}

\end{document}